%% file: NDSS_Full_version.tex
\documentclass[conference]{IEEEtran}

\usepackage{cite}

\ifCLASSINFOpdf
\else
\fi
\ifCLASSOPTIONcompsoc
 \usepackage[caption=false,font=normalsize,labelfont=sf,textfont=sf]{subfig}
\else
 \usepackage[caption=false,font=footnotesize]{subfig}
\fi

\usepackage{algpseudocode}
\usepackage[ruled,linesnumbered,boxed,commentsnumbered]{algorithm2e}
\usepackage{multirow} 
\usepackage{amsfonts}
\usepackage{graphicx}
\usepackage{makecell}
\usepackage{threeparttable}
\usepackage{tikz}
\usepackage{amsmath}
\usepackage{amsthm}
\usepackage{float}
\usepackage{diagbox}
\usepackage{soul}

\usepackage{tabularray}

\usepackage[table]{xcolor}
\usepackage{colortbl}

\usetikzlibrary{shapes.geometric, shapes.arrows, positioning, fit, calc, backgrounds, shadows}

\usepackage{pifont}

\usepackage{hyperref}

\usepackage{marvosym}

\newtheorem{definition}{Definition}
\newtheorem{theorem}{Theorem}
\newtheorem{lemma}{Lemma}

\begin{document}
%
\title{\textit{Aspire}: Practical Batch Queries on Dynamic Encrypted Database with Full Leakage Suppression
}



\author{\IEEEauthorblockN{Hao Nie$^{\dagger}$, Wei Wang$^{\dagger}$\textsuperscript{(\Letter)}, Peng Xu$^{\dagger}$\textsuperscript{(\Letter)}, Haoming Yan$^{\dagger}$, Laurence T. Yang$^{\ddagger}$, and Kaitai Liang$^{\S}$}
	\IEEEauthorblockA{$^{\dagger}$Huazhong University of Science and Technology, $^{\ddagger}$Zhengzhou University, \\$^{\ddagger}$St. Francis Xavier University, $^{\S}$University of Turku, and $^{\S}$TU Delft\\
		\{nie, viviawangwei, xupeng, yhm361426852\}@hust.edu.cn, ltyang@gmail.com, kaitai.liang@utu.fi}\thanks{\Letter Wei Wang and Peng Xu are the corresponding authors.}}


%


\IEEEoverridecommandlockouts
\makeatletter\def\@IEEEpubidpullup{6.5\baselineskip}\makeatother
\IEEEpubid{\parbox{\columnwidth}{
		This is the full version of the paper accepted at the Network and Distributed System Security (NDSS) Symposium 2027, 22 - 26 March 2027, Seoul, Republic of Korea.
}
\hspace{\columnsep}\makebox[\columnwidth]{}}

\maketitle

\begin{abstract}
Leakage-abuse attacks against encrypted database systems exploit semantic leakage and pose a severe threat to the confidentiality of plaintext data and queries. To address this problem, we propose a zero-leakage encrypted database that rules out leakage-abuse attacks, which fundamentally rely on semantic leakage. We first distinguish between two kinds of zero-leakage security. \textit{Ideal zero-leakage} forbids the disclosure of any information beyond public parameters, whereas \textit{restricted zero-leakage} additionally permits the leakage of metadata (such as database size and query timing). We show that restricted zero-leakage is sufficient to thwart leakage-abuse attacks and is achievable in practice. We then formally define adaptive restricted zero-leakage security and present \textit{Aspire}, an encrypted database that satisfies this notion. Aspire integrates probabilistic batch code and minimal perfect hashing to support efficient batch keyword queries and dynamic updates while revealing only non-semantic metadata. Experimental results demonstrate the practicality and efficiency of Aspire, showing that Aspire can achieve an amortized query latency of less than 1.38 ms (7.53 s for a batch of 5461 keyword queries), with an amortized communication cost of approximately 2.95 KB (16.1 MB in total), on a $2^{22}$ key-value pairs database.
\end{abstract}


%
\IEEEpeerreviewmaketitle

\section{Introduction}

 
Cloud storage systems such as Amazon Web Services~\cite{aws_edb} and MongoDB~\cite{MongoDB} widely rely on encrypted search primitives, including searchable symmetric encryption (SSE), to enable efficient query processing over encrypted data. 
State-of-the-art SSE schemes inevitably leak access and search patterns, which enable powerful leakage-abuse attacks (LAAs) that recover user queries and even reconstruct the whole database 
\cite{DBLP:conf/ndss/IslamKK12, DBLP:conf/ccs/CashGPR15, DBLP:conf/ndss/BlackstoneKM20, DBLP:conf/ccs/NingHPYL0D21}.

Prior defenses attempt to mitigate leakage without sacrificing efficiency. Forward and backward security reduce linkability across updates~\cite{DBLP:conf/ndss/StefanovPS14,DBLP:conf/ccs/Bost16,DBLP:conf/ccs/BostMO17}, yet recent works demonstrate that residual leakage remains sufficient for accurate inference attacks \cite{DBLP:conf/ccs/Nie00CYCL25,DBLP:conf/ccs/XuZXYW23}. 
This reveals a fundamental limitation: \textbf{Efficient encrypted search inherently leaks structural information that adversaries can exploit}. 

A natural countermeasure is to eliminate this leakage entirely by enforcing oblivious access. 
Oblivious RAM (ORAM) hides access patterns through continuous reshuffling and re-encryption \cite{DBLP:journals/jacm/GoldreichO96,DBLP:journals/jacm/StefanovDSCFRYD18}, and has been incorporated into leakage-suppressed SSE constructions \cite{DBLP:conf/crypto/KamaraMO18,DBLP:conf/eurocrypt/GeorgeKM21}. 
Alternatively, private information retrieval (PIR) enables data retrieval without disclosing specific access indices \cite{DBLP:conf/focs/ChorGKS95}. 
These approaches entail substantial overhead: ORAM constructions require at least $\Omega(\log N)$ bandwidth per access, while PIR shifts the burden to server-side computation and struggles to support dynamic or keyword-based queries efficiently \cite{DBLP:conf/uss/MahdaviK22,DBLP:conf/sp/LiuLWR24}. 
This reveals a critical trade-off: \textbf{Strong mitigation against LAA requires heavy cryptographic mechanisms, yet practical systems prioritize efficiency and flexibility}. 
Reconciling these competing objectives remains a formidable open problem.

In this work, we distinguish between two tiers of zero-leakage security: the idealized \textit{ideal zero-leakage} security and the more pragmatic \textit{restricted zero-leakage} security. 
Both notions are formalized to provide comprehensive defense against existing LAAs. 
We then present a concrete cryptographic construction that realizes these security guarantees. 


\noindent\textbf{Zero Leakage Notions: Ideal vs. Restricted.}
Ideal zero-leakage guarantees that no information is revealed beyond the system’s public parameters. Beyond thwarting LAA-related patterns, this notion encompasses the concealment of operational metadata, such as database size and the timing, number, and type of operations. However, such strong security incurs significant efficiency overhead, limiting practicality in large-scale or dynamic deployments. To achieve a better balance with practical efficiency, we introduce restricted zero-leakage, which permits the disclosure of operational metadata while strictly concealing all patterns related to the underlying plaintexts. We summarize the relationships between these security notions, leakage profiles, and attacks in Figure~\ref{fig:summary_of_notion}.

\begin{figure*}
    \centering

\definecolor{bgred}{RGB}{255, 235, 235}     
\definecolor{bggreen}{RGB}{235, 250, 235}   
\definecolor{arrowblue}{RGB}{180, 195, 225} 
\definecolor{mypurple}{RGB}{0, 180, 0}  
\definecolor{myorange}{RGB}{100, 0, 150}    
\definecolor{mycyan}{RGB}{0, 180, 200}      
\def\colTwoX{4.7}
\def\colThreeX{9.5}
\def\ThreesubTwo{3.8}
\begin{tikzpicture}[
    node distance=1cm,
    font=\rmfamily,
    txt/.style={align=center},
    redarrow/.style={
        single arrow, draw=none, fill=red, 
        single arrow tip angle=110,
        minimum width=0.7cm,
        single arrow head extend=0.15cm,
        minimum height=0.6cm, shape border rotate=270
    },
    bluearrow/.style={
        single arrow, draw=none, fill=cyan!80!blue, 
        single arrow tip angle=110,
        minimum width=0.7cm,
        single arrow head extend=0.15cm,
        minimum height=1cm, shape border rotate=270
    },
    schemebox/.style={draw=none, align=center, inner sep=2pt},
]

\node[single arrow, fill=arrowblue!50, draw=black!60, 
      minimum height=14.5cm, minimum width=0.8cm, 
      single arrow head extend=0.1cm, opacity=0.7,
      anchor=west] (headerarrow) at (-2.5, 0) {};
\node[txt,anchor=east,font=\footnotesize] at (-2.45,0) {Practical\\ but Leakage};
\node[txt,anchor=west,font=\footnotesize] at (11.8,0) {Zero-leakage but\\ Impractical};
\node[anchor=center, font=\small] (PL) at (0, 0) {Pattern Leakage};
\node[anchor=center, font=\small] (RZL) at (\colTwoX, 0) { Restricted Zero-Leakage};
\node[anchor=center, font=\small] (GZL) at (\colThreeX, 0) { Ideal Zero-Leakage};

\node[anchor=center] at (0, 0.5) {Vulnerable to LAA};
\node[anchor=center] at (7.2,0.5) {Fully Against LAA};


\node[txt, below=0.6cm of PL.south, anchor=center, text width=5cm,font=\footnotesize] (scheme_left) {
    Schemes under F\&B security\cite{DBLP:conf/ccs/Bost16,DBLP:conf/ccs/BostMO17},\\
    AZL\cite{DBLP:conf/crypto/KamaraMO18}, Pancake\cite{DBLP:conf/uss/GrubbsKLBL0R20},\\
    VLH\&AVLH\cite{DBLP:conf/eurocrypt/KamaraM19}
};

\node[txt,below = 0cm of scheme_left,xshift=-3cm,yshift=0.75cm] (LK) {System};

\node[txt,below=0.63cm of RZL.south,anchor=center,font=\footnotesize] (scheme_mid_top) {
    ORAM\&OMAP\cite{DBLP:journals/jacm/StefanovDSCFRYD18,DBLP:journals/jacm/GoldreichO96,DBLP:journals/pvldb/CaoFLZFWXYR24}\\
    FZL\&Dyn. FZL\cite{DBLP:conf/crypto/KamaraMO18,DBLP:conf/eurocrypt/GeorgeKM21}\\
    {\color{mypurple}\textit{Aspire}} (Ours)
};

\node[txt, below=0.57cm of GZL.south,anchor=center,font=\footnotesize] (scheme_right_top) {
    
    EXH (subliminal mode)\cite{DBLP:journals/iacr/BahiKMN24}\\
    
};
\node[txt, below=0.74cm of GZL.south,anchor=center,font=\footnotesize] (scheme_right_top) {
    
    { }\\
    { }\\
{\color{mypurple}\textit{Aspire}}+Metadata Hiding (Ours)
};

\coordinate (center_right_1) at ($(RZL)!0.5!(GZL)-(0,0.65)$);
\node[
    draw=myorange, 
    thick, 
    dashed,
    minimum width=8.7cm, 
    minimum height=0.65cm, 
    inner sep=0pt,
    anchor=center,
    yshift=-1.3
] (box_oram) at (center_right_1) {};
\node[text=myorange, anchor=west,font=\small] at (box_oram.east) {ORAM Based};

\coordinate (center_right_2) at ($(RZL)!0.5!(GZL)-(-0,1.23)$);
\node[
    draw=mypurple, 
    thick, 
    dashed,
    minimum width=8.7cm, 
    minimum height=0.3cm, 
    inner sep=0pt,
    anchor=center
] (box_PIR) at (center_right_2) {};

\node[text=mypurple, anchor=west, xshift = -1cm,yshift=-0.35cm,font=\small] at (box_PIR.east) {PIR or ORAM Based};

\node[redarrow, below=0cm of scheme_left] (arrow1) {};
\node[redarrow, below=-0.05cm of scheme_mid_top] (arrow2) {};
\node[redarrow, below=-0.05cm of scheme_right_top] (arrow3) {};


\node[txt, below=0cm of arrow1,font=\small] (leak_left_1) {Public Parameters\\Operation Metadata${}^1$};
\node[txt, text=red!80!black, draw=red, dashed, thick, inner sep=1pt, below=0.1cm of leak_left_1,font=\small] (leak_left_2) {AP, SP, QC, VP, FVP, RSP, UL${}^2$};

\node[txt,below = 0cm of leak_left_1,xshift=-3cm,yshift=0.5cm] (LK) {Leakage};
\node[txt,below = 1.3cm of LK] {Attack};
\node[txt, below=0cm of arrow2,font=\small] (leak_mid) {Public Parameters\\Operation Metadata};
\node[
    below=0cm of leak_left_1, 
    inner sep=1pt,           
    draw=none,               
    text opacity=0           
] (leak_left_2) {};        

\node[txt, below=0cm of arrow3,font=\small] (leak_right) {Public Parameters\\}; 
\node[
    below=0cm of leak_left_1, 
    inner sep=1pt,           
    draw=none,               
    text opacity=0           
] (leak_left_2) {};        

\node[draw=mycyan, dashed, thick, fit=(leak_left_1)(leak_mid)(leak_right), inner sep=0pt] (box_leakage) {};


\node[redarrow, below=0.5cm of leak_left_1.south,anchor=north] (arrow4) {};

\coordinate (center_right) at ($(leak_mid)!0.5!(leak_right)$);
\node[bluearrow, below=0.2cm of box_leakage.south -| center_right] (arrow_big_blue) {};

\node[txt, below=0cm of arrow4, font=\footnotesize] {Known-data\cite{DBLP:conf/ccs/CashGPR15,DBLP:conf/ccs/NingHPYL0D21,DBLP:conf/ndss/IslamKK12,DBLP:conf/ndss/BlackstoneKM20,DBLP:conf/uss/EspirituKMP25},\\ Similar-data\cite{DBLP:conf/ccs/Nie00CYCL25,DBLP:conf/uss/OyaK22,DBLP:conf/uss/Nie00ZYL24,DBLP:conf/ccs/PouliotW16,DBLP:conf/uss/Damie0P21},\\ Injection \cite{DBLP:conf/uss/Zhang00YL23,DBLP:conf/ndss/BlackstoneKM20,DBLP:conf/uss/ZhangKP16,DBLP:conf/eurosp/PoddarWLP20}, and\\Database Reconstruction\cite{DBLP:conf/sp/LachariteMP18,DBLP:conf/ccs/KellarisKNO16,DBLP:conf/ccs/FalzonMACRST20,DBLP:conf/ccs/GuiJW19,DBLP:conf/uss/EspirituKMP25}};

\node[txt, below=0.1cm of arrow_big_blue] {No known LAA};

\begin{scope}[on background layer]
    
    
    \fill[bgred!50] (-2.4,0.8) rectangle (2.4,-5.6);
    
    \fill[bggreen!50] (2.4,0.8) rectangle (7.1,-5.6);

    \fill[bggreen!90] (7.1,0.8) rectangle (11.8,-5.6);
    
    \draw[dashed, black] (2.4,0.8) -- (2.4,-5.6);
\end{scope}
    
    
    
\end{tikzpicture}
    \begin{tablenotes}
        \footnotesize
        \item[1] {1 The public parameters include  security parameter, cryptographic setting descriptions, and system-wide configurations that are fixed at setup stage, all independent of database contents and operations. The operation metadata includes database size and the timing, number, and type of operations, which are non-semantic properties/correlation to plaintext stored data.} 
        \item[2] 2 ``AP'', ``SP'', ``QC'', ``VP'', ``FVP'', ``RSP'', ``UL'' represent access pattern, search pattern, query correlation, volume pattern, file volume pattern, response size pattern, and update leakage, respectively. 
        We only list the well-studied patterns and leakage is not limited to the aforementioned. Systems can suffer from a subset of leakages. 
    \end{tablenotes}
    \caption{Summary of LAA security landscape: notions, systems, leakage and attacks.}
    \label{fig:summary_of_notion}
    
\end{figure*}

\emph{Ideal but expensive:} While ideal zero-leakage removes all exploitable information leakage, the computational and communication overhead required to obfuscate operational metadata is often prohibitive for real-world applications. 
For example, masking operation timing requires maintaining a constant execution frequency~\cite{DBLP:journals/iacr/BahiKMN24}, 
which entails the continuous injection of dummy operations regardless of actual demands. 
Furthermore, concealing the database size requires a fixed capacity at system setup, which is fundamentally incompatible with modern databases that require dynamic scaling and elasticity. 


\emph{Fully against LAA but practical:} Restricted zero-leakage avoids the overhead of metadata obfuscation while providing robust defense against LAAs. 
Because existing LAAs primarily exploit semantic leakage, such as access, search, and volume patterns from an encrypted database and queries, to correlate encrypted data with auxiliary information, the primary security is the elimination of all plaintext-dependent semantic leakage. 
We demonstrate that constructions satisfying this security requirement can be realized with high efficiency.  
Several prior security definitions, such as fully zero-leakage in STE~\cite{DBLP:conf/crypto/KamaraMO18}, ORAM~\cite{DBLP:journals/jacm/StefanovDSCFRYD18}, and Oblivious map (OMAP)~\cite{DBLP:journals/pvldb/CaoFLZFWXYR24}, can be interpreted as specialized instances of restricted zero-leakage security.
Thus, we propose restricted zero-leakage as a unifying and definitive security upper bound for defending against LAA.
We provide a formal definition and a rigorous security analysis in Section \ref{sec:zero_leakage_sec}. 

\noindent\textbf{New construction.} 
We formalize the system model of a zero-leakage encrypted key-value database, encompassing $\mathsf{Setup}$ and $\mathsf{Query}$ protocols, where $\mathsf{Query}$ can support batch keyword searches. 
We decouple the update process into $\mathsf{UpdateLocal}$ and $\mathsf{Update}$, ensuring that client-server interaction occurs only during the $\mathsf{Update}$ phase, while $\mathsf{UpdateLocal}$ manages immediate local state changes. We further introduce a formal security model that utilizes leakage functions $\mathcal{L}_\mathsf{Setup}$ and $\mathcal{L}_\mathsf{Update}$, which take only the database size as input, while $\mathcal{L}_\mathsf{Query}$ and $\mathcal{L}_\mathsf{UpdateLocal}$ take $\bot$, satisfying the restricted zero-leakage standard. {Here, $\bot$ denotes empty leakage, i.e., the server learns no information from the corresponding leakage function beyond the occurrence of the protocol invocation.}

Building on the model, we present \textit{Aspire}, an encrypted key-value database system that achieves restricted zero-leakage security and can serve as a foundation for other zero-leakage databases supporting various query functionalities.
Technically, \textit{Aspire} leverages probabilistic batch code (PBC) and minimal perfect hash (MPH) to transform a batch of keyword queries into calls to multiple oblivious index accesses (instantiated with PIR or ORAM). 
For client updates, \textit{Aspire} stashes them locally before pushing to the server via an independent update protocol that minimizes the impact on latency related to updates and search queries. 
\textit{Aspire} also employs a synchronization mechanism compatible with batched queries that reveals the encrypted database size only once every $\alpha$ $\mathsf{Update}$.
Because encrypted key-value stores are fundamental building blocks in both industrial and academic systems (Section~\ref{app:kvsystems}), we demonstrate the versatility of \textit{Aspire} by adapting it to the SSE setting (Section~\ref{sec:discussion}).

Finally, we conduct extensive experiments to evaluate the efficiency of \textit{Aspire}. 
We instantiate \textit{Aspire} on top of the PIR protocol Spiral\cite{DBLP:conf/sp/MenonW22}, the batch PIR protocol PIRANA\cite{DBLP:conf/sp/LiuLWR24}, and PathORAM\cite{DBLP:journals/jacm/StefanovDSCFRYD18}.  
On a database with $2^{22}$ key-value pairs, \textit{Aspire} (based on PIRANA) costs about 5.6 MB of local storage and achieves an amortized query latency of less than 1.38 ms (7.53 s for a batch of 5461 queries), with an amortized communication cost of approximately 2.95 KB (16.1 MB total).
Compared to the state-of-the-art OMAP~\cite{DBLP:journals/pvldb/CaoFLZFWXYR24}, \textit{Aspire} achieves significantly lower communication overhead and total latency for batch searches. 

{Our contributions can be summarized as follows:}
\begin{itemize}
    \item We introduce restricted zero-leakage security, which hides all plaintext-dependent leakage while allowing only non-semantic operation metadata, and we show that it is sufficient to provide robust protection against existing LAAs.

    \item We formalize a dynamic encrypted key-value database model that supports batch keyword queries and updates under restricted zero-leakage. We then propose Aspire, which integrates PBC, MPH, and oblivious index access to achieve high-performance batch keyword queries and dynamic updates.
    
    \item We demonstrate through experiments that Aspire yields superior efficiency over prior state-of-the-art secure baselines.
\end{itemize}


{

\section{Related Work}
\label{app:related_work}

\begin{table}[tp]
    \centering
    \begin{threeparttable}
    \caption{Existing leakages and attacks.}
    \label{tab:summary_of_leakage_and_attacks}
    \begin{tabular}{c@{\hspace{3pt}}c}
    \hline
        Leakage & Attacks \\
        \hline
        AP  & \makecell{Peekaboo\cite{DBLP:conf/ccs/Nie00CYCL25},IHOP\cite{DBLP:conf/uss/OyaK22},Jigsaw\cite{DBLP:conf/uss/Nie00ZYL24},Graphm\cite{DBLP:conf/ccs/PouliotW16},
        \\ 
        Score Attack\cite{DBLP:conf/uss/Damie0P21},Count\cite{DBLP:conf/ccs/CashGPR15},Leap\cite{DBLP:conf/ccs/NingHPYL0D21},Subgraph\cite{DBLP:conf/ndss/BlackstoneKM20}, \\  ZKP\cite{DBLP:conf/uss/ZhangKP16},LMP\cite{DBLP:conf/sp/LachariteMP18},KKNO\cite{DBLP:conf/ccs/KellarisKNO16},FMAC\cite{DBLP:conf/ccs/FalzonMACRST20},\\IKK\cite{DBLP:conf/ndss/IslamKK12},PolySys\cite{DBLP:conf/uss/EspirituKMP25}
        }\\
        
        \hline
        SP & \makecell{Sap\cite{DBLP:conf/uss/OyaK21}, Freq\cite{DBLP:journals/isci/LiuZWT14}, BVMA\cite{DBLP:conf/uss/Zhang00YL23},\\IHOP\cite{DBLP:conf/uss/OyaK22},PolySys\cite{DBLP:conf/uss/EspirituKMP25}}\\
        \hline
         QC & IHOP\cite{DBLP:conf/uss/OyaK22}\\
        \hline
        VP  & \makecell{ Jigsaw\cite{DBLP:conf/uss/Nie00ZYL24},IHOP\cite{DBLP:conf/uss/OyaK22},Sap\cite{DBLP:conf/uss/OyaK21}, BVMA\cite{DBLP:conf/uss/Zhang00YL23},\\PWLP\cite{DBLP:conf/eurosp/PoddarWLP20},Count\cite{DBLP:conf/ccs/CashGPR15}, GLMP\cite{DBLP:conf/ccs/GrubbsLMP18}, KKNO\cite{DBLP:conf/ccs/KellarisKNO16},\\GJW\cite{DBLP:conf/ccs/GuiJW19},PolySys\cite{DBLP:conf/uss/EspirituKMP25}}\\
        \hline
       
        FVP  & Peekaboo\cite{DBLP:conf/ccs/Nie00CYCL25}, FMA\cite{DBLP:conf/ccs/XuZXYW23}\\
        \hline
        RSP& BVA\cite{DBLP:conf/uss/Zhang00YL23}, Decoding\cite{DBLP:conf/ndss/BlackstoneKM20}\\
        \hline
        UL & VIA\cite{DBLP:conf/ccs/XuZXYW23}, UF\&UFID\cite{DBLP:conf/codaspy/HaltiwangerH24}\\
        \hline
        \makecell{Operation\\ Metadata} & No known leakage-abuse attack\\
        \hline
    \end{tabular}

    
    \end{threeparttable}
\end{table}

\subsection{Prior LAAs}
\label{app:review_of_LBA}
This section presents a comprehensive overview of previous LAAs and demonstrates that such attacks are ineffective against systems that satisfy restricted zero-leakage security.

Currently, LAAs primarily comprise known-data attacks\cite{DBLP:conf/ccs/CashGPR15,DBLP:conf/ccs/NingHPYL0D21,DBLP:conf/ndss/IslamKK12,DBLP:conf/ndss/BlackstoneKM20,DBLP:conf/uss/EspirituKMP25}, similar-data attacks\cite{DBLP:conf/ccs/Nie00CYCL25,DBLP:conf/uss/OyaK22,DBLP:conf/uss/Nie00ZYL24,DBLP:conf/ccs/PouliotW16,DBLP:conf/uss/Damie0P21}, injection attacks\cite{DBLP:conf/uss/Zhang00YL23,DBLP:conf/ndss/BlackstoneKM20,DBLP:conf/uss/ZhangKP16,DBLP:conf/eurosp/PoddarWLP20}, and database reconstruction attacks targeting range queries\cite{DBLP:conf/sp/LachariteMP18,DBLP:conf/ccs/KellarisKNO16,DBLP:conf/ccs/FalzonMACRST20,DBLP:conf/ccs/GuiJW19,DBLP:conf/uss/EspirituKMP25}.  
In known-data attacks, the adversary is assumed to possess partial or complete plaintexts of documents stored in the client’s database. The adversary exploits the leakage patterns revealed during search operations to correlate queries with keywords and to match encrypted documents with the corresponding known plaintexts.  
Similar-data attacks allow the adversary to have a set of documents that are statistically similar to those in the client’s database or estimations of the client’s query distribution. These attacks rely on aligning statistical information derived from auxiliary, similar datasets with the leakage patterns produced by query and update operations.  
In injection attacks, the adversary is assumed to be capable of inserting carefully crafted files exhibiting particular patterns into the client’s database. By monitoring the resulting changes in the leakage patterns, the adversary can infer which keywords are being queried.  
Against range queries, database reconstruction attacks exploit pattern leakage to infer or reconstruct the distribution of records in the database. The effectiveness of these attacks can be further enhanced when the adversary has access to additional auxiliary information about the encrypted database.
The patterns and corresponding attacks utilizing them are listed in Table \ref{tab:summary_of_leakage_and_attacks}.

As shown in Figure \ref{fig:attacker}, LAAs correlate observed leakage patterns of $op$ (i.e., search, addition, deletion, query-response, etc.) with known patterns from auxiliary information. This correlation enables the adversary to recover client queries or use leaked patterns to reconstruct the structure of the range-query-supporting encrypted database.
Under the restricted zero-leakage conditions, the adversary can only obtain operation metadata, such as execution timing and the number of operations performed. This information cannot be systematically associated with the specified auxiliary knowledge and, consequently, does not contribute to the database reconstruction or query recovery.

We note, however, that a very strong adversary possessing substantial auxiliary information, such as prior knowledge that the client is likely to query specific restaurants during typical dining hours, may still infer the underlying query with high accuracy. Such inference derives from information external to the cryptographic mechanism itself and is orthogonal to the security guarantees provided by the applied encryption scheme.

\begin{figure}[tp]
    \centering
    \includegraphics[width=0.9\linewidth]{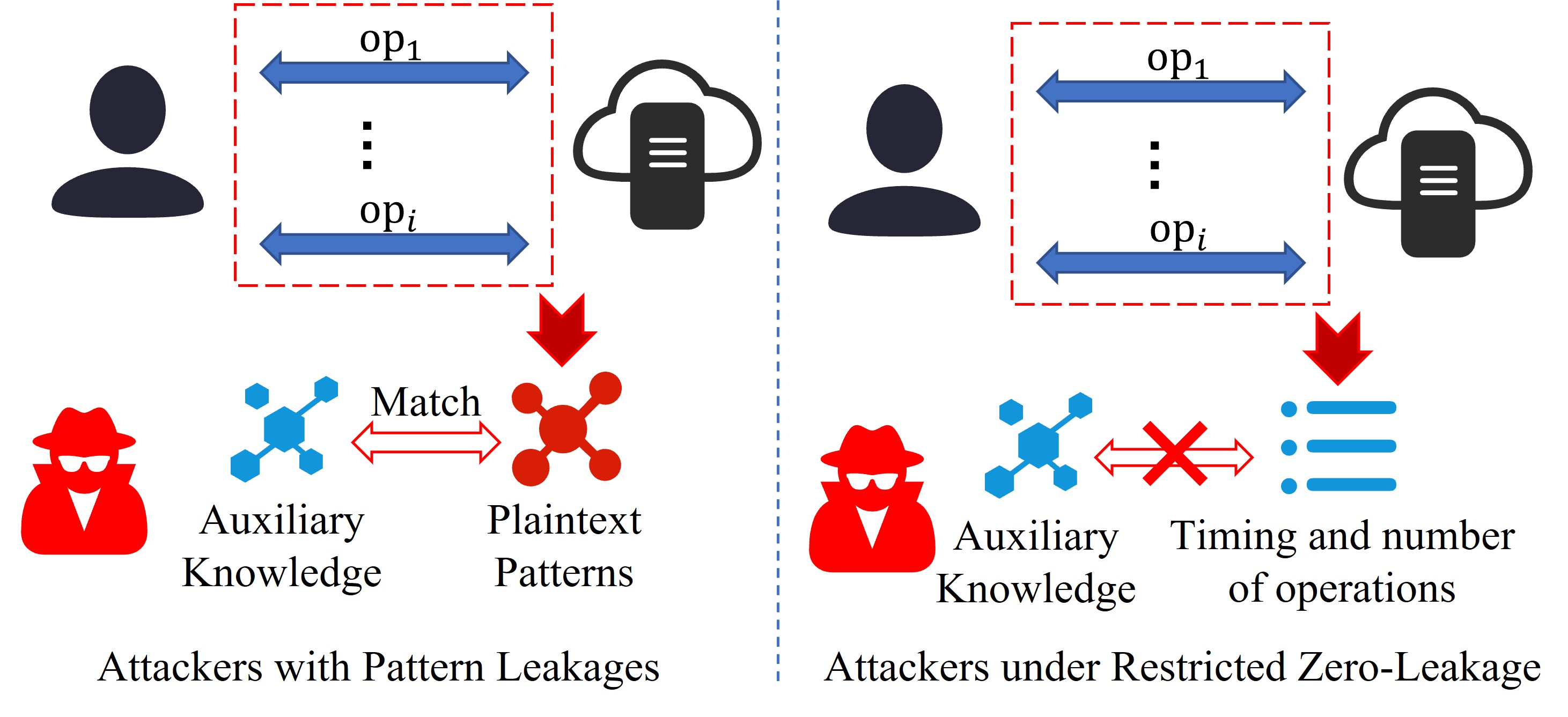}
    \caption{LAA with pattern leakages vs. restricted zero-leakage.}
    \label{fig:attacker}
\end{figure}

\subsection{Review of Key-value Storage Systems}
\label{app:kvsystems}
Key-value storage systems~\cite{cloudfront_kvs,redis,DBLP:conf/sosp/DeCandiaHJKLPSVV07,DBLP:conf/uss/GrubbsKLBL0R20,DBLP:conf/icdcs/ZhangXMJ22,DBLP:journals/tifs/ZhangMWGJ24,DBLP:journals/pvldb/ZhengXWWHQLR24,DBLP:journals/pacmmod/MaiyyaVAAK23} serve as the architectural backbone for large-scale platforms like DynamoDB~\cite{aws_dynamodb}, Bigtable~\cite{google_bigtable}, and RocksDB~\cite{rocksdb}. With security considerations, SSE and encrypted document databases\cite{DBLP:conf/ccs/BostMO17,DBLP:conf/ccs/Bost16,DBLP:journals/tkde/VoYSLNW23,DBLP:journals/popets/GuiPPW24,DBLP:conf/uss/Chamani0KD22,DBLP:conf/ccs/ChamaniPPJ18,DBLP:conf/ndss/WuL23} adopt encrypted key-value databases as the underlying foundational layer. In most prior studies, the access patterns of encrypted KV stores are exposed and can therefore be exploited under LAA.

To hide the access pattern of key-value storage, Pancake\cite{DBLP:conf/uss/GrubbsKLBL0R20} uses frequency smoothing to transform plaintext accesses into uniformly distributed encrypted accesses. Assuming client queries follow a fixed distribution, an adversary cannot distinguish the transcript from a sequence of uniformly distributed accesses to random bit strings. However, when queries are correlated, IHOP can successfully attack Pancake\cite{DBLP:conf/uss/OyaK22}.  
Building on Pancake, Zhang et al.\cite{DBLP:conf/icdcs/ZhangXMJ22,DBLP:journals/tifs/ZhangMWGJ24} proposed $K$-indistinguishable frequency smoothing and key–value packing, which reduce costs while allowing additional frequency leakage.
Jia et al.\cite{DBLP:conf/uss/Jia0K24} further explored a three-way tradeoff between storage, bandwidth, and security on Pancake and proposed variants of Pancake to meet the corresponding trade-offs.
Zheng et al.\cite{DBLP:journals/pvldb/ZhengXWWHQLR24}  proposed SWAT, introducing $\theta$-query decorrelation and a sampling-pool technique to mitigate query-correlated leakage in Pancake.  
Beyond frequency-smoothing schemes, Maiyya et al.\cite{DBLP:journals/pacmmod/MaiyyaVAAK23} introduced $(\alpha,\beta)$-uniformity, where data written to the server is accessed before $\alpha$ other accesses, and cache data is not written back until after $\beta$ other accesses, and proposed Waffle to achieve this.}

\section{Zero-Leakage in Encrypted Database}
\label{sec:zero_leakage_sec}

Most existing methods against leakage primarily depend on oblivious access techniques and structured-encryption compilers. 
While a thorough discussion is postponed to Section \ref{sec:existing_zero_leakage}, we begin by formalizing the notion of “zero-leakage” in the setting of encrypted databases.

\subsection{Ideal and Restricted Zero-Leakage}
\label{sec:zero_leakage}

Kamara et al. define \textit{zero-leakage} encrypted databases as those whose leakage depends only on the security parameter or other public parameters~\cite{DBLP:conf/crypto/KamaraMO18}.
They further define fully zero-leakage for structured encryption and dynamic structured encryption, which informally allows database size leakage during setup and rebuild but leaks nothing during queries~\cite{DBLP:conf/crypto/KamaraMO18,DBLP:conf/eurocrypt/GeorgeKM21}.
Building on this, Bahi et al. introduce \textit{subliminal leakage}, which treats operation timing as leakage~\cite{DBLP:journals/iacr/BahiKMN24}, 
and propose subliminal encrypted multi-maps that mix real and dummy operations sent at fixed, independent time intervals.

Inspired by these works, we introduce \textit{ideal zero-leakage} as an ideal security baseline that rules out all forms of leakage from the encrypted database.
\\ 
\emph{\textbf{Ideal Zero-Leakage (Informal).}
An encrypted database exhibits \textit{ideal zero-leakage} if the execution of any sequence of database operations reveals no observable information beyond what is explicitly allowed by the public parameters.
}

\textbf{Public parameters} include the security parameter, cryptographic scheme descriptions, and system-wide configurations that are fixed at setup time and are independent of the database contents and operations.
The ideal zero-leakage does not allow any information about the plaintext to be leaked. 
Moreover, it excludes the operation metadata from the client's behavior and the database parameters.
\textbf{Operation metadata} includes database size, the types, timing, and number of all operations (e.g., search queries, update queries, and other interactions with the server).

Existing LAAs primarily exploit \textit{semantic} leakage, such as access, search, and volume patterns, rather than operation metadata, as it does not directly convey semantic information about the underlying plaintext. 
Operation metadata is largely driven by external and application-dependent factors and, on its own, does not provide a stable plaintext-dependent structure that enables query and data reconstruction.
Thus, hiding the metadata is unnecessary for defending against LAA. 
More discussions are in Section \ref{app:review_of_LBA}.

Also, suppressing operation metadata would require strong countermeasures that impose significant performance overhead and may fundamentally alter system behavior. 
For example, concealing operation timing would necessitate continuously executing dummy operations~\cite{DBLP:journals/iacr/BahiKMN24}. {Related works have explored hiding individual dimensions of operation metadata. For example, \cite{DBLP:conf/eurocrypt/GeorgeKM21} includes mechanisms that make searches and updates indistinguishable, while ObliviSync~\cite{DBLP:conf/ndss/AvivCMR17} hides the timing and occurrence of writes through fixed-interval, fixed-rate synchronization.}

Thus, we consider a more practical yet still strong security notion, which we refer to as \textit{restricted zero-leakage}.
\\
\emph{\textbf{Restricted Zero-Leakage (Informal).}
An encrypted database exhibits \textit{restricted zero-leakage} if the execution of any sequence of database operations reveals no observable information beyond what is explicitly allowed by the public parameters and operation metadata.
}

Restricted zero-leakage does not constrain an encrypted database scheme to eliminate specific leakage types, instead, it enforces that no semantic information about the plaintext is revealed. 
The prior works on zero-leakage are consistent with restricted zero-leakage~\cite{DBLP:conf/crypto/KamaraMO18,DBLP:conf/eurocrypt/GeorgeKM21}.

The restricted zero-leakage, therefore, provides a strong measure against LAAs.
Existing attacks on encrypted databases fundamentally rely on semantic information exposed through the leakage, which enables the attacker to correlate the observed leakage with prior knowledge of plaintext distributions. 
By eliminating all meaningful semantic clues about the plaintext, restricted zero-leakage rules out attacks exploiting such leakage (see Section \ref{app:review_of_LBA} for more discussion). 

To summarize, we distinguish between ideal zero-leakage, an ideal yet impractical notion, and restricted zero-leakage, a practical yet strong security bound against LAAs. 
The rest of this paper centers on the formalization and realization of restricted zero-leakage. 

\subsection{{Threat Model} and Restricted Zero-Leakage Security}
\label{sec:restricted_zero_leakage}
Based on the above discussions, we provide the formal definition of restricted zero-leakage security.
We consider the scenario of an encrypted database where a client wants to outsource its database (such as a set of key-value pairs or a set of files) to an untrusted server. 
The client interacts with the server to initialize the encrypted database and subsequently performs searches, insertions, deletions, updates, and other database operations.
The server is honest-but-curious; it follows the protocol correctly but observes all ciphertexts transmitted by the client\footnote{This work focuses on the confidentiality of the database and queries, rather than the correctness or availability of outsourced computation. Arbitrary deviations of malicious servers are outside the scope.}.
The server tries to infer information about the database's contents or the operations.

Restricted zero-leakage requires that any information observable to the server during system execution be non-semantic. 
We therefore define the notion through a real/ideal simulation experiment. In the real experiment, the adversary observes the actual server-side view of the protocol.
In the ideal experiment, a simulator is given only the restricted leakage and must generate an indistinguishable server-side view. If such a simulator exists, then the real execution reveals no plaintext-dependent semantic information beyond the allowed metadata.
{
\begin{definition}[Adaptive Restricted Zero-Leakage Security] Let $\Sigma$ be an encrypted database scheme.
Let $\mathsf{Op}_{\Sigma}$ denote the set of all operations supported by $\Sigma$ after setup.
Let $\mathcal L_{\mathsf{rz}}=(\mathcal L_{\mathsf{Setup}},\mathcal L_{\mathsf{Op}})$ be a restricted leakage function, where $\mathcal L_{\mathsf{Setup}}$ captures the leakage during setup and $\mathcal L_{\mathsf{Op}}$ captures the leakage of an arbitrary operation $op\in\mathsf{Op}_{\Sigma}$.
$\mathcal L_{\mathsf{rz}}$ outputs only non-semantic metadata, such as public parameters, operation types, operation timing, and the database size.
For a PPT adversary $\mathcal A$, we define the following two experiments.

\textbf{Real}$_{\Sigma,\mathcal A}(1^\lambda)$: The adversary $\mathcal A$ chooses an initial database $DB$\footnote{Allowing the adversary to choose the initial database strengthens the security definition by covering the plaintext instance most favorable to distinguishing.}. The challenger runs the real setup protocol of $\Sigma$ on $DB$ and gives $\mathcal A$ the complete server-side view of the setup. Then, for $i=1,\ldots,q(\lambda)$, {where $q$ is polynomial in $\lambda$}, $\mathcal A$ adaptively chooses an operation $op_i$ of $\Sigma$ {based on the previous transcripts}, such as a search, insertion, deletion, or update. The challenger executes the real protocol on the current database state and returns to $\mathcal A$ the complete server-side view of this operation. Finally, $\mathcal A$ outputs a bit $b$.

\textbf{Ideal}$_{\mathcal S,\mathcal A,\mathcal L_{\mathsf{rz}}}(1^\lambda)$: The adversary $\mathcal A$ chooses an initial database $DB$. The simulator $\mathcal S$ is given only $\mathcal L_{\mathsf{Setup}}(DB)$ and generates a simulated server-side setup view for $\mathcal A$. Then, for each adaptive operation $op_i$ chosen by $\mathcal A$, the simulator is given only the corresponding leakage $\mathcal L_{\mathsf{Op}}(op_i)$, and generates a simulated server-side view of the operation. Finally, $\mathcal A$ outputs a bit $b'$.

We say that $\Sigma$ satisfies adaptive restricted zero-leakage security with respect to $\mathcal L_{\mathsf{rz}}$ if, for every PPT adversary $\mathcal A$, there exists a PPT simulator $\mathcal S$ such that
\[
\left|
\Pr[\textbf{Real}_{\Sigma,\mathcal A}(1^\lambda)=1]
-
\Pr[\textbf{Ideal}_{\mathcal S,\mathcal A,\mathcal L_{\mathsf{rz}}}(1^\lambda)=1]
\right|
\leq \mathsf{negl}(\lambda).
\]
\label{def:rzl}
\end{definition}
The leakage function $\mathcal L_{\mathsf{rz}}$ is required to be non-semantic. It may depend on public parameters and operation metadata, but not on the plaintext contents or  queried keywords. Any server-side view that can be simulated from $\mathcal L_{\mathsf{rz}}$ contains no plaintext-dependent semantic information. {In a scheme, if the adversary can infer any information not included in the permitted leakage defined by the leakage function, the scheme violates the security definition.}
}

Restricted zero-leakage does not limit the types of concrete leakage but ensures no semantic leakage about the plaintext. As for different schemes, the same leakage type may or may not relate to the plaintext.
For example, the current database size generally does not contain semantic information about operations if each addition (or deletion) operation results in the same change in database size.
However, if a scheme adopts some form of compression and adding different content of the same size changes the database size differently, the current database size actually conveys semantic information about the added content relative to the original database. 
Additionally, the restricted zero-leakage does not require hiding when an operation is occurring or the number of operations.

The definition is intentionally independent of specific system interfaces or leakage representations. 
The following sections formalize the system model for encrypted key-value databases alongside concrete leakage functions, establishing that any scheme $\Sigma$ adhering to these functions satisfies the restricted zero-leakage. 

{
\subsection{Existing Zero-Leakage Techniques and Their Limitations}
\label{sec:existing_zero_leakage}
Oblivious access is the standard way to suppress leakage and achieve restricted zero-leakage. 
ORAM, first formalized by Goldreich and Ostrovsky~\cite{DBLP:journals/jacm/GoldreichO96}, hides a client’s access pattern to remote storage, and many works improve its efficiency~\cite{DBLP:journals/jacm/StefanovDSCFRYD18,DBLP:conf/uss/RenFKSSDD15}. 
Tree-based schemes, especially Path ORAM~\cite{DBLP:journals/jacm/StefanovDSCFRYD18}, are among the most practical, but each access still reads and writes a root-to-leaf path, incurring $O(\log n)$ bandwidth. 
Non-recursive ORAM also requires a client-side position map of size $\Theta(n)$, while recursive ORAM lowers this storage at the cost of extra recursion levels and sufficiently large blocks. 
Later schemes such as Ring ORAM~\cite{DBLP:conf/uss/RenFKSSDD15} and Circuit ORAM~\cite{DBLP:conf/ccs/WangCS15} further optimize bandwidth or circuit complexity, but still incur significant communication overhead. 
Thus, ORAM strongly suppresses leakage but remains expensive for large-scale or batched keyword queries.
OMAP generalizes ORAM from array-like storage to oblivious key-value maps and dynamic data structures~\cite{DBLP:journals/pvldb/CaoFLZFWXYR24}. 
A natural construction builds a balanced search tree over ORAM, incurring $O(\log n)$ logical node accesses per map operation and thus about $O(\log^2 n)$ communication and often $O(\log n)$ interaction rounds~\cite{DBLP:conf/ccs/WangNLCSSH14,DBLP:conf/sp/RocheAC16,DBLP:conf/uss/TinocoGS23}. 
Specialized OMAP schemes reduce this to roughly $O(\log n)$ by embedding position information into oblivious data structures~\cite{DBLP:journals/pvldb/CaoFLZFWXYR24}, but they still incur multiple rounds of interaction and have weak support for efficient batch keyword queries.
ORAM and OMAP can be viewed as specialized instances of restricted zero-leakage, as they may reveal non-semantic operation metadata such as the number and timing of accesses, but they hide the plaintext-dependent access pattern.

For structured encryption, leakage suppression is often formulated as a compiler or suppressor that transforms an existing scheme into another scheme with a reduced leakage profile~\cite{DBLP:conf/crypto/KamaraMO18,DBLP:conf/eurocrypt/GeorgeKM21,DBLP:conf/asiacrypt/GeorgeKME25}. 
These transformations are largely inspired by the Square-Root ORAM paradigm~\cite{DBLP:journals/jacm/GoldreichO96,DBLP:conf/crypto/KamaraMO18}. 
While such compilers can provide strong leakage-suppression guarantees and avoid using full ORAM as a black-box layer, they still inherit square-root-ORAM-style costs, including non-trivial storage, cache management, and expensive rebuilds~\cite{DBLP:conf/crypto/KamaraMO18,DBLP:conf/eurocrypt/GeorgeKM21}. 
Some of the constructions, such as FZL and Dynamic FZL\cite{DBLP:conf/crypto/KamaraMO18,DBLP:conf/eurocrypt/GeorgeKM21}, fall within restricted zero-leakage.

Some works further hide operation metadata to approach ideal zero-leakage. 
Bahi et al.~\cite{DBLP:journals/iacr/BahiKMN24} proposed subliminal encrypted multi-maps, which hide operation time, operation count, and even whether an operation occurs by continuously sending real and dummy operations. 
This stronger metadata hiding requires persistent cover traffic and client-side queue management, and in subliminal mode, it may sacrifice full correctness or incur high overhead and unbounded latency. 
Therefore, fully hiding operation metadata is often too expensive for practical dynamic databases.

}

\section{Encrypted Key-Value Database}
{
In this section, we present the restricted zero-leakage encrypted key-value database, Aspire.
Instead of using oblivious access mechanisms directly as complete database solutions, Aspire is designed as an encrypted key-value database system that treats them as building blocks and provides an upper-layer encrypted database system. 
Built upon oblivious access mechanisms (such as PIR and ORAM), Aspire preserves restricted zero-leakage while improving efficiency and functionality, especially for batch keyword queries and dynamic workloads.}

{Our system consists of a client accessing a database hosted by a server, where the client maintains the cryptographic keys and local state. This setting can also support multiple clients through a trusted proxy that collects their queries and updates and executes the Aspire protocol on their behalf\cite{DBLP:conf/uss/GrubbsKLBL0R20}. From the server's perspective, the proxy acts as a single logical client. A representative deployment is an enterprise gateway that privately accesses an outsourced database on behalf of multiple employees.}

\subsection{System model}
\label{sec:system_model}
An encrypted key-value database scheme $\Sigma$, over the keyword space $\mathcal{K}$ and value space $\mathcal{V}$, consists of the following four routines, all taking the computational security parameter $\lambda$ as an implicit input.
\begin{itemize}
    \item $(\mathsf{st_C};EDB)\gets\mathsf{Setup}(DB)$, which takes a key-value database $DB=\{(k_1,v_1),...,(k_n,v_n)\}$ as input and outputs an encrypted database $EDB$ for the server and the local state $\mathsf{st_C}$ for the client. 
    \item $V\gets \mathsf{Query}(\mathsf{st_C},K;EDB)$, which takes a keyword set $K$ of size {$|K|$} and the $\mathsf{st_C}$ from the client and the $EDB$ from the server. The client receives a value set $V$.
    \item $\mathsf{st_C}'\gets \mathsf{UpdateLocal}(\mathsf{st_C},op,k,v)$, which takes an operation $op$ that could be $ADD$ or $DEL$, a keyword $k$, a value $v$, and the $\mathsf{st_C}$ as input. It outputs a new state $\mathsf{st_C}'$. The input $v$ {is} $\bot$ if the $op$ is $DEL$ {and is not $\bot$ if the $op$ is $ADD$}.
    \item $(\mathsf{st_C}';EDB')\gets \mathsf{Update}(\mathsf{st_C};EDB)$, which takes the $\mathsf{st_C}$ from the client and the $EDB$ from the server as input. The client gets a new $\mathsf{st_C}'$ and the server takes a new $EDB'$.
\end{itemize}

The $\mathsf{Query}$ allows a batch query for a set of keywords $K$.
We also separate the updates to $\mathsf{UpdateLocal}$ and $\mathsf{Update}$ to accommodate batch updates. The $\mathsf{UpdateLocal}$ function takes a single update request from the client and updates only the client's local $\mathsf{st_C}$. 
And the $\mathsf{Update}$ will synchronize part of the updates with the server.
\subsection{Security Goal}
\label{sec:security_goal}
The restricted zero-leakage encrypted key-value database should satisfy correctness in query processing and ensure restricted zero-leakage security.

\noindent\textbf{Correctness}.
An encrypted key-value database scheme $\Sigma$ is said to be correct if, for all $i \in [|K|]$, $k_j=K[i]$, and $k_j\in DB$, $\Sigma.\mathsf{Query}$ returns the correct result $V$, i.e., $V[i]=v_j$, with a probability of $1-\mathsf{negl}(\lambda)$. 

Our encrypted key-value database aims to achieve the restricted zero-leakage security as defined in Section \ref{sec:restricted_zero_leakage}.

\noindent\textbf{The restricted zero-leakage security}.
{For the encrypted key-value database considered in this work, we instantiate the restricted leakage function as follows:}
\[
\mathcal L_{\mathsf{Setup}}(DB)=L'(|DB|), \quad
\mathcal L_{\mathsf{Query}}(K)=\bot,
\]
\[
\mathcal L_{\mathsf{UpdateLocal}}(op,k,v)=\bot,\quad
\mathcal L_{\mathsf{Update}}()=L''(|DB|).
\]
 %
  %
   %
    %
%
{The $\bot$ means that, beyond the fact that the corresponding operation is invoked, the server learns no additional information, {while $\mathcal L'$ and $\mathcal L''$ denote arbitrary functions that depend only on $|DB|$.}}

{We provide a formal description of the leakage functions and proof that an encrypted key-value database scheme with the above leakage functions achieves the restricted zero-leakage security in Appendix \ref{app:security_proof}.}



{
\subsection{Preliminaries}
Aspire builds on private information retrieval (PIR), probabilistic batch code (PBC), and minimal perfect hashing (MPH). We briefly describe their functionality below and provide formal definitions in Appendix \ref{app:notations}.

\noindent\textbf{Private information retrieval}. PIR\cite{DBLP:conf/focs/ChorGKS95} allows a client to retrieve an element at a specified index $i$ from a server-held database $DB=(x_1,\ldots,x_n)$ without revealing the index $i$ to the server.

\noindent\textbf{Probabilistic batch codes}. An $(n,m,s,b)$-PBC\cite{DBLP:conf/sp/AngelCLS18} encodes $n$ elements into $b$ buckets $B_1,\ldots,B_b$, containing a total of $m$ encoded elements. For any batch $I={i_1,\ldots,i_s}\subseteq[n]$, the retrieval algorithm assigns the requested elements to distinct buckets such that $x_{i_j}$ can be recovered by accessing at most one element from its assigned bucket. 

\noindent\textbf{Minimal perfect hashing}. For a keyword set $K$, an MPH function\cite{DBLP:conf/alenex/EspositoGV20} is a collision-free mapping in which each keyword is assigned a different number in $[|K|]$.
Thus, each keyword $k\in K$ is mapped to a unique position $h_K(k)$, without any unused positions in the range. MPH supports expected linear construction time and constant lookup time, using under 2 bits per key for storage (about 0.22 MB for $2^{20}$ arbitrarily long keywords in our experiments).

}

\subsection{Design Overview}

{Aspire} comprises three main parts: batch queries via PIR, keyword-to-index transformation, and asynchronous client updates.

\noindent \textbf{Batch queries via PIR.} We use PIR to access the database while hiding everything about the queried keywords and results (the PIR can also be replaced with ORAM).
Meanwhile, we adopt the PBC to amortize the costs of batch queries (detailed notations are in Appendix \ref{app:notations}).

The overview of \textit{Aspire} is in Figure \ref{fig:overview}.
{Aspire uses an IND-CPA-secure randomized symmetric encryption scheme with fresh randomness for each encryption and re-encryption.}
In the setup, the client distributes its database $DB=\{(k_1,v_1),...,(k_n,v_n)\}$ to $b$ buckets via PBC.
For each $(k_i,v_i)$, the PBC uses 3 hash functions to map $k_i$ to 3 buckets and store $(k_i,v_i)$ in those buckets. 
After mapping, the client encrypts each pair in the buckets and sends the encrypted buckets to the server as the encrypted database.
To query a set of keywords $K$, the client uses the PBC to find the bucket for each keyword in $K$.
The PBC uses cuckoo hashing with the same 3 hash functions as in the setup, and each bucket has at most one keyword mapped to it.
Finally, the client adopts multiple keyword PIRs or a batch keyword PIR to retrieve the key-value pair in each bucket. 
As shown in ~\cite{DBLP:conf/sp/AngelCLS18,DBLP:conf/sp/LiuLWR24}, the utilization of PBC can save a lot of computation and communication costs for batch PIR.


\begin{figure}[htbp]
    \centering
    \subfloat[Setup]
	{
 \label{fig:setup_overview}
		\begin{minipage}[b]{0.4\textwidth}
        \centering
        \includegraphics[width=\textwidth]{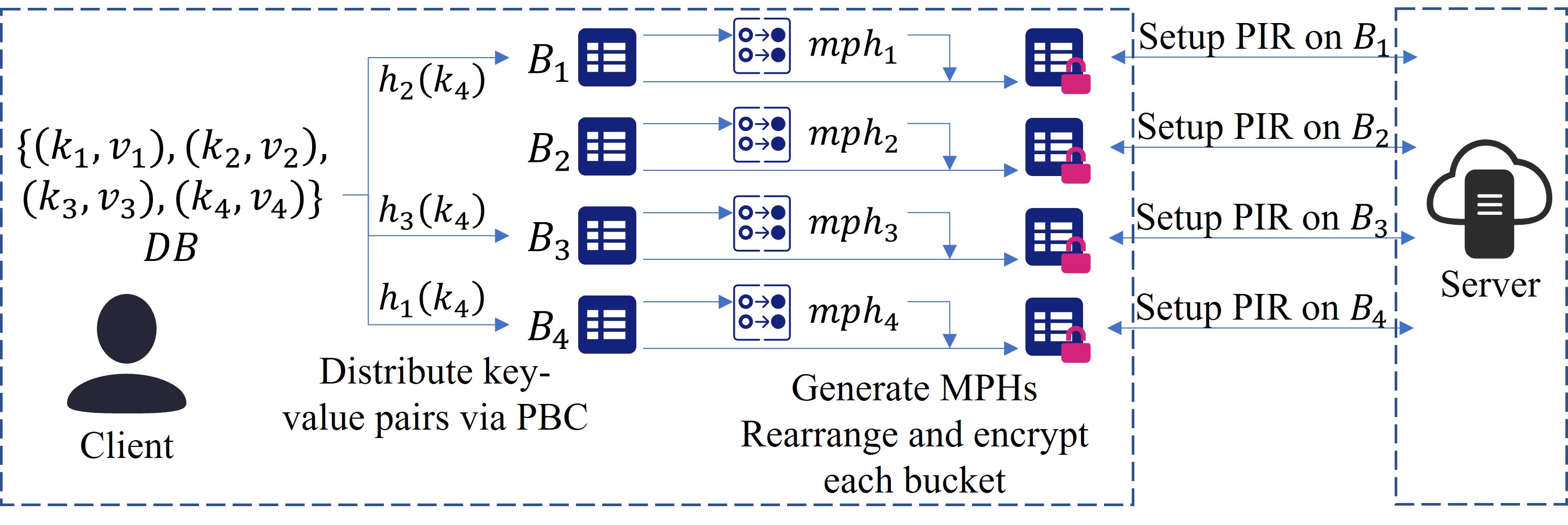}
    \end{minipage}
	}
    
    \subfloat[Query]
	{
 \label{fig:query_overvie}
		\begin{minipage}[b]{0.4\textwidth}
        \centering
        \includegraphics[width=\textwidth]{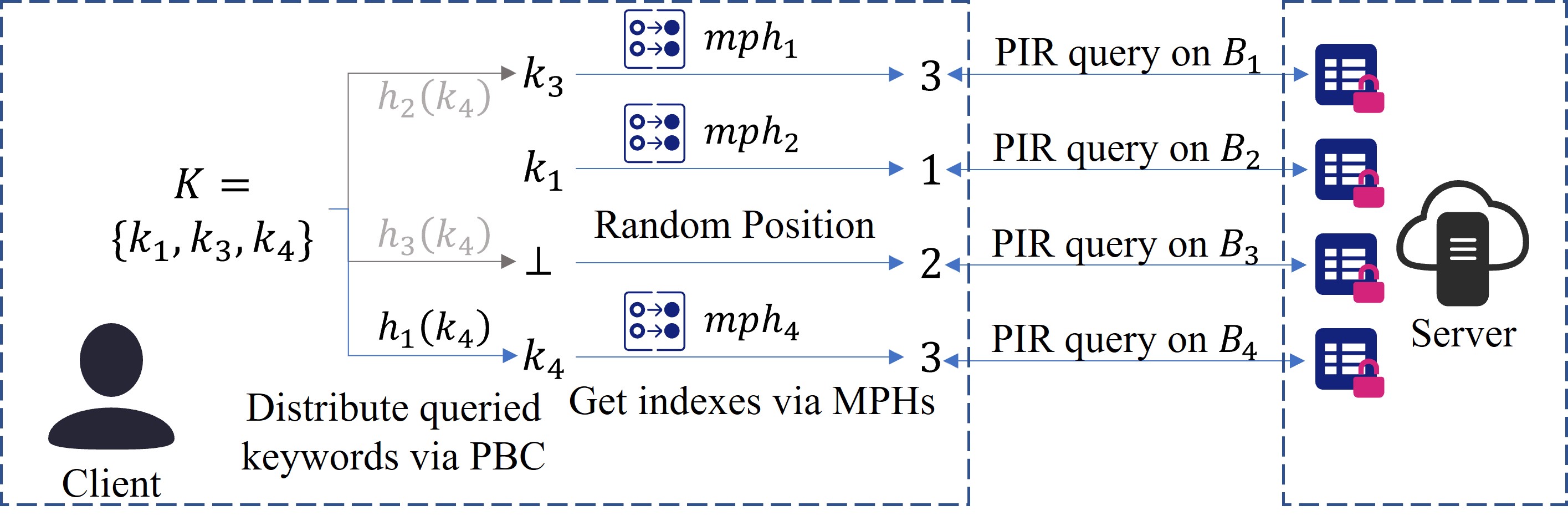}
    \end{minipage}
	}

    \subfloat[Update]
	{
 \label{fig:update_overvie}
		\begin{minipage}[b]{0.4\textwidth}
        \centering
        \includegraphics[width=\textwidth]{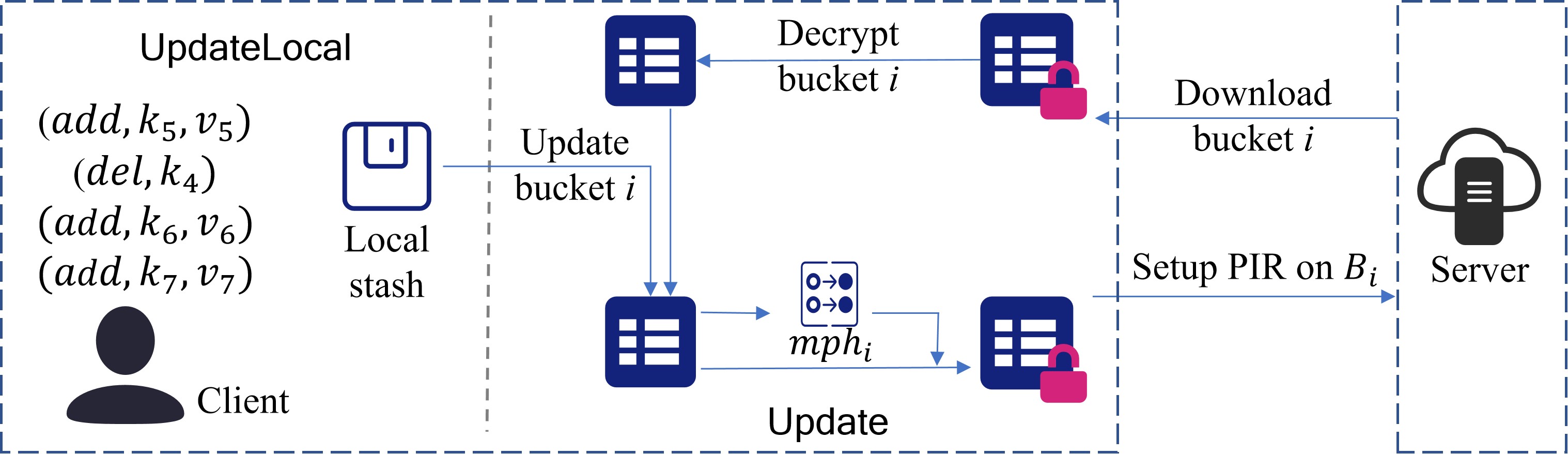}
    \end{minipage}
	}
    
    \caption{System overview of \textit{Aspire}.}
    \label{fig:overview}
\end{figure}

\noindent \textbf{Keyword search to index access.}
As detailed in Appendix \ref{app:review_of_KPIR}, 
existing keyword PIR constructions incur significantly higher overheads than standard index PIR. 
We circumvent these costs by using an MPH-based mapping that reduces keyword searches to index accesses.

The MPH maps a keyword set $K$ to the first $|K|$ natural numbers (see Appendix \ref{app:MPH_preliminary}). 
In the setup, the client collects keywords from each bucket, builds an MPH $mph_i$ from each bucket’s keyword set $B_i$, then realigns the bucket’s key-value pairs according to the MPH before encryption and upload.
In the batch query, the client uses the stored MPH functions to locate a keyword’s position in a bucket, then uses standard PIR or batch PIR to retrieve the key-value pairs.
{For each bucket $B_i$, the client determines the PIR query position $\mathit{pos}_i$. If a queried key is assigned to $B_i$, the client obtains $\mathit{pos}_i$ by evaluating the corresponding MPH. Otherwise, it selects $\mathit{pos}_i$ uniformly at random. The client then uses $\mathit{pos}_i$ as the index input to $\mathsf{PIR.Query}$.}
Appendix~\ref{app:keyword_PIR} details a keyword PIR construction leveraging an MPH function. 

\noindent \textbf{Independent client updates.}
The client's updates are independent of update interactions with the server.
For each addition and deletion issued by the client, we directly store the updates locally. 
Later, the client synchronizes the updates with the server by $\mathsf{Update}$.
In the $i$-th $\mathsf{Update}$, the client downloads and decrypts the $(i\%b)$-th bucket. 
The client identifies which additions and deletions correspond to the keywords mapping the $(i\%b)$-th bucket and thus updates the bucket accordingly.
The client then constructs a new MPH function according to the updated keywords, rearranges, re-encrypts, and uploads the bucket to the server.

To handle the changing size of the database, when updating the first bucket, the client sets a proper size $B_{size}$ for the buckets according to the current number of pairs in the database; thus, it does not leak the actual size of each bucket.
When updating a bucket, if the updated bucket size is less than $B_{size}$, the client pads it to $B_{size}$ . If the size is larger, the client randomly removes some key-value pairs from the bucket to the stash until the bucket size is $B_{size}$.
The client reinitializes $B_{size}$ at the time of the next update to the first bucket.
If the number of buckets is about $n^{1/2}$, and the client calls $\mathsf{Update}$ for each addition and deletion, the amortized communication cost of an update is $O(n^{1/2})$. Also, the storage cost of the stash is $O(n^{1/2})$.
We note that invocations of $\mathsf{Update}$ are decoupled from the client-side update, allowing them to be executed during idle periods. 
As a result, these synchronizations do not contribute to the latency of search queries or update requests.

\section{The Construction of \textit{Aspire}}
We first introduce the concrete construction of \textit{Aspire} and then provide a detailed overhead and security analysis.

\subsection{Modules}
\label{sec:constructionofstorage}

\textit{Aspire} includes $\mathsf{Setup}$, $\mathsf{Query}$, $\mathsf{UpdateLocal}$, and $\mathsf{Update}$, which are detailed in Algorithm \ref{alg1}, \ref{alg2}, \ref{alg3_v2}, and \ref{alg4_v2}, respectively.

\input{Alg_Setup}
\input{Alg_Query}
\subsubsection{Setup}
In the beginning of the setup, the client owns a database $DB=\{(k_1,v_1),...,(k_n,v_n)\}$. 
The client initializes an empty stash $S_1$ to store the local updates and a $Count=n$ to record the number of pairs in the database. The client and the server both initialize a $Rcount=0$ to record the times of $\mathsf{Update}$ being issued.

The client first extracts all the keywords $K=\{k_1,...,k_n\}$ from the database and calls $\mathsf{PBC.Encode}$ to distribute all the keywords in the $K$ into $b$ buckets as $\{K_0,...,K_{b-1}\}$ (line \ref{alg1:ekeywords}-\ref{alg1:PBCencode}). 
The number of buckets $b$ is a public parameter.
The $\mathsf{PBC.Encode}$ adopts three hash functions to map each keyword $k\in K$ into three buckets, where $k$ is replicated three times.
After the mapping, a bucket may contain duplicate keywords, as two hash functions may map a keyword to the same bucket. We directly retain the unique keywords in each bucket.

The client then sets the size of every bucket to be the same $B_{size}$. 
Finding a small $B_{size}$ while ensuring that the probability of a bucket exceeding the size of $B_{size}$ is relatively small is a $3n$-balls-into-$b$-bins problem. {Following ~\cite{DBLP:conf/random/RaabS98}, we choose $B_{size}$ as $3n/b+\sqrt{(6n\log b)/b}$ (line \ref{alg1:set_bsize}), ensuring that the probability of overflow remains acceptable while keeping the number of dummy items limited.
If a bucket has a size larger than $B_{size}$, the client will put random elements from the bucket into the stash without affecting correctness. 
}
For each bucket $i$, while $|K_i|>B_{size}$, the client removes a random keyword $k$ from $K_i$ and adds $(k,j,v,ADD)$ to $S_1[i]$ until $|K_i|=B_{size}$, where the $j$-th hash function in PBC maps $k$ to $i$, and $v$ is the corresponding value of $k$ (line \ref{alg1:bsize_b}-\ref{alg1:bsize_e}).

After that, for each bucket $i$, the client calls the $\mathsf{MPH.Setup}$ to generate a minimal perfect hash function $mph_i$ for $K_i$ (line \ref{alg1:setmph}). 
The client then rearranges $K_i$ according to $mph_i$, i.e., places the keyword $k\in K_i$ at the position of $\mathsf{MPH.Evaluate}(mph_i,k)$ (line \ref{alg1:rearrange}). 
After the rearrangement, $mph_i$ will map $K_i[j]$ to $j$. 

Then, for each keyword $k\in K_i$, the client encrypts $k||v||p$, where $v$ is the corresponding value of the keyword $k$, and the $p$-th hash function in PBC maps $k$ to the $i$-th bucket (line \ref{alg1:enc1}-\ref{alg1:enc2}). 
Thus, although the same key-value pairs appear in multiple buckets, they are indistinguishable. 
The encrypted key-value pairs of each bucket are then padded to the size of $B_{size}$ and sent to the server.

Finally, the server and the client run the setup of PIR on every bucket, and the client receives a hint for querying each bucket using PIR (line \ref{alg1:server1}-\ref{alg1:server2}). 
We note that PIR can be replaced by BatchPIR to save querying time, and if the adopted PIR scheme does not require the client to store a hint, the step of sending hints to the client is not needed.

\subsubsection{Query}
To query a set of keywords $\{q_1,...,q_s\}$, the client calls the $\mathsf{PBC.GenSchedule}$ to distribute the keywords to each bucket. 
$\mathsf{PBC.GenSchedule}$ adopts the 3-way cuckoo hash with the same three hash functions as in the setup (line \ref{alg2:pbcgen}).
Each bucket receives a keyword or nothing. 
Thus, the client can query each bucket for the value of a single keyword to obtain all the desired values.
If the number of queried keywords $s$ is smaller than $2b/3$, the cuckoo hash would fail with a probability of around $2^{-40}$~\cite{DBLP:conf/sp/AngelCLS18}. {Fewer hash functions make it difficult to achieve this failure rate under the same load, while more hash functions increase database replication and storage overhead.} Note that the client can rerun the cuckoo hash with fewer queried keywords if a failure occurs.

For each bucket $i$, if a keyword is distributed, the client uses the $mph_i$ to obtain the position of the stored key-value pair on the server; otherwise, the client just uses a random position (line \ref{alg2:mph1}-\ref{alg2:mph2}).
The client and the server then run a PIR query in each bucket to retrieve the corresponding values for the keywords (line \ref{alg2:pir1}-\ref{alg2:pir2}). 
We note that the PIR here can be replaced by a BatchPIR to save query costs.
Finally, the client checks the retrieved key-value pairs for correctness and the local stash $S_1$ for potential updates (line \ref{alg2:check1}-\ref{alg2:check2}).

\input{Alg_Update_v2}
\subsubsection{UpdateLocal}
\label{sec:updatelocal}
For every update request, i.e., to add a key-value pair or delete a keyword from the database, the client stashes the updates locally. 
When encountering an update $(k,v,op)$, where $k,v$ is the updated key-value pair and $op$ is the $ADD$ or $DEL$, the client first calculates the buckets to which the keyword has been hashed (line \ref{alg3:pbc}).
For each hash function $hash_i$, the client adds $(k,i,v,op)$ to $S_1[pos_i]$
and changes the number of key-value pairs in the database $Count$ accordingly (line \ref{alg3:up1}-\ref{alg3:up2}).

To save the storage of the client, the client can split the stash into two parts, $S_1$ and $S_2$. The $S_1$ stores the keywords, the hash function number, and the corresponding operation, while $S_2$ stores the key-value pairs from the addition. Thus, the client does not need to store repeated values.

\input{Alg_Reconstruct_v2}
\subsubsection{Update}
To synchronize the updates in the stash with the server, the client and server execute $\mathsf{Update}$.
We set that after at most $a$ executions of $\mathsf{UpdateLocal}$, an $\mathsf{Update}$ will occur, where $a$ is a constant value.
The client and server hold a $Rcount$, which records the times of $\mathsf{Update}$ and is initialized as $0$. 
In each $\mathsf{Update}$, the client and server reconstruct the $Rcount\%b$-th bucket.

As the database size may change over time, the bucket size will also change. 
When updating the first bucket, i.e., $Rcount\%b=0$, the client re-sets $B_{size}$ to handle the change in database size. 
Similar to the $\mathsf{Setup}$, the client re-sets $B_{size}$ to $3Count/b+$ $\sqrt{(6Count\log b)/b}$, where $Count$ is the current database size (line \ref{alg4:setBsize1}-\ref{alg4:setBsize2}). Until the next time the first bucket is updated, all the buckets will be reconstructed to the same size $B_{size}$.

The client downloads the encrypted bucket $\widetilde{KV}_{Rcount\%b}$ and decrypts it as $KV'$ (line \ref{alg4:download}). 
The client then updates the stashed updates of bucket $Rcount\%b$ by performing the updates in $S_1[Rcount\%b]$ to $KV'$ (line \ref{alg4:updates_1}-\ref{alg4:updates_2}).

If the updated $KV'$ is larger than $B_{size}$, random elements will be removed from $KV'$ to stash $S_1$ until $|KV'|=B_{size}$ (line \ref{alg4:remove1}-\ref{alg4:remove2}).
After which, the client re-generates the $mph_{Rcount\%b}$ for the bucket and rearranges the $KV'$ according to the generated MPH (line \ref{alg4:genmph1}-\ref{alg4:genmph2}). 
The client encrypts the $KV'$ as $\widetilde{KV}'$.
If $|\widetilde{KV}'|<B_{size}$, the client pads $\widetilde{KV}'$ to the size of $B_{size}$ with random numbers.
Finally, the client uploads $\widetilde{KV}'$ to the server and updates $Rcount$.
The server and client re-run the setup of PIR on bucket $Rcount\%b$ with $\widetilde{KV}'$ to replace the original PIR.

The $\mathsf{Update}$ can be executed independently of $\mathsf{Query}$ and $\mathsf{UpdateLocal}$, which provides the following three benefits.
\textcircled{1} Enables the client to flexibly choose the execution time of $\mathsf{Update}$, for example, during the time when the server is idle.
\textcircled{2} If a $\mathsf{Query}$ occurs during the $\mathsf{Update}$, the $\mathsf{Update}$ does not increase the latency of $\mathsf{Query}$, as the $\mathsf{Update}$ process can suspend until the $\mathsf{Query}$ is complete. 
\textcircled{3} $\mathsf{UpdateLocal}$ can be executed multiple times on average before an $\mathsf{Update}$ is executed once. 
Thus, an $\mathsf{Update}$ can upload many additions and deletions of key-value pairs.

\subsection{Performance Analysis}
We analyze the performance of \textit{Aspire} in client storage, query and update overheads.

\noindent\textbf{Client storage}.
We analyze the client-side storage of \textit{Aspire}, which consists of three parts: the MPH storage, the client-side storage of the underlying $\mathsf{PIR}$ scheme, and the local stash.
{
\begin{theorem}
\label{thm:client-storage}
Let $n$ be the number of key-value pairs and $b$ be the number of buckets. Assume that at most $a$ executions of $\mathsf{UpdateLocal}$ occur between two consecutive executions of $\mathsf{Update}$. Let $S_{\mathsf{MPH}}$ be the total storage of all MPH functions, $H_{\mathsf{PIR}}$ be the total client-side storage of the underlying $\mathsf{PIR}$ scheme, and $S_{\mathsf{ov}}$ be the size of the overflowed items. Then the client storage of \textit{Aspire} is
$$
S_{\mathsf{MPH}}+H_{\mathsf{PIR}}+O(ab)+S_{\mathsf{ov}}.
$$
In particular, if $ab=O(\sqrt n)$, then the client storage is
$$
S_{\mathsf{MPH}}+H_{\mathsf{PIR}}+O(\sqrt n)+S_{\mathsf{ov}}.
$$
\end{theorem}

\begin{proof}
It remains to bound the local stash caused by pending updates. Each execution of $\mathsf{UpdateLocal}$ creates at most three update records, one for each hash function used by the PBC encoding. Since at most $a$ executions of $\mathsf{UpdateLocal}$ occur between two consecutive executions of $\mathsf{Update}$, and each modified item is synchronized with the server within the next $b$ executions of $\mathsf{Update}$, the number of pending update records is at most $3ab$.

Each update record stores a keyword, a value, and a constant-size tag, so the pending-update stash size is at most $3ab(|k|+|v|+3)$ bits, which is $O(ab)$ when $|k|$ and $|v|$ are fixed. If $ab=O(\sqrt n)$, this term becomes $O(\sqrt n)$. The overflowed items contribute additional $S_{\mathsf{ov}}$ bits. 
Combining the MPH and client-side $\mathsf{PIR}$ storage overheads establishes the claimed bound. 
\end{proof}

We now instantiate each term in Theorem~\ref{thm:client-storage}. First, the client stores one MPH function for each bucket. Since an MPH can be represented using fewer than $2$ bits per keyword~\cite{DBLP:conf/alenex/EspositoGV20}, and each key-value pair is mapped to at most three buckets, the total MPH storage is less than $6n$ bits.
Second, $H_{\mathsf{PIR}}$ depends on the underlying $\mathsf{PIR}$ or $\mathsf{BatchPIR}$ scheme. For PIR schemes without client-side hints, such as YPIR~\cite{DBLP:conf/uss/Menon024} and PIRANA~\cite{DBLP:conf/sp/LiuLWR24}, this term is $O(1)$ or zero.
Finally, the overflow term $S_{\mathsf{ov}}$ follows from the standard balls-into-bins analysis~\cite{DBLP:conf/random/RaabS98}. 
We set $B_{\mathsf{size}}= {3n}/{b}+\sqrt{{6n\log b}/{b}}$ to balance the stored pairs in the server and client storage. 
For example, when $n=2^{20}$ and $b=2^{10}$, the expected overflow is about $1.3$, indicating that the expected number of elements moved to the stash is very small under this parameter setting.
A detailed analysis of the expected overflow is in Appendix~\ref{app:overflow}.


}

\noindent\textbf{Query costs}.
For the $\mathsf{Query}$, the costs are mainly dominated by the costs of the $\mathsf{PIR}$ used.
Although the PIR requires a linear scan over the entire database to answer a $\mathsf{Query}$, the costs are amortized over the search for $2b/3$ keywords. Assuming the bucket number $b$ is $O(\sqrt{n})$, the amortized computation cost for a keyword is $O(\sqrt{n})$.  
As for the communication cost, since PIR can retrieve only a single encrypted element from a bucket, \textit{Aspire} needs amortized constant communication overhead for sufficiently large values.

\noindent\textbf{Update costs}.
$\mathsf{UpdateLocal}$ requires minimal computation costs as it only inserts items into the local stash.
The $\mathsf{Update}$ needs to download a bucket and upload the updated bucket, which has a communication cost of $O(n/b)$. 
If we assume that there are $a$ executions of $\mathsf{UpdateLocal}$ between two $\mathsf{Update}$, the communication cost is $O(n/(ab))$ amortized by $a$ $\mathsf{UpdateLocal}$. 
Assuming $ab$ has the size of $O(\sqrt{n})$, the amortized communication cost for each $\mathsf{UpdateLocal}$ is also $O(\sqrt{n})$.
Similar to the communication costs, $\mathsf{Update}$ needs to re-encrypt a bucket of size $O(n/b)$ and regenerate the MPH for the bucket; thus, its computation cost is $O(n/b)$.
In certain batch PIR schemes, the contents of multiple buckets may be encoded into a single codeword. Consequently, modifying the data in any bucket necessitates re-encoding the entire database. In this case, the computational complexity of the $\mathsf{Update}$ operation is $O(n)$.

\subsection{Security Analyses}
\textit{Aspire} achieves the security goals with the specified leakage functions in Section \ref{sec:security_goal}. 
We provide a proof sketch of the correctness and security of \textit{Aspire}. The formal security model and proof are in Appendix \ref{app:security_proof}. 

\noindent \textbf{Correctness.}  The correctness arises directly from the accurate mapping of keywords in the query set to buckets and the retrieval of the keyword in each bucket. 
The former relies on the correctness of $\mathsf{PBC.GenSchedule}$, which has an overwhelming success probability. 
For the retrieval of the keyword in each bucket, if the keyword is in the encrypted bucket on the server, with the correctness of $\mathsf{MPH}$ and $\mathsf{PIR}$, the value will be retrieved and recovered accurately. 
If it is in the local stash, the value will also be verified and returned. 

\noindent \textbf{Security.} 
In the $\mathsf{Setup}$, the client interacts with the server by uploading the padded encrypted buckets to the server. The server only sees the size of the buckets, which is determined by the number of key-value pairs in the database and the public parameter $b$.
In the $\mathsf{Query}$, the server only sees a PIR query for each bucket. 
As PIR hides any information about the queried keywords, the server of \textit{Aspire} also receives no information about which encrypted key-value pair has been accessed.
The $\mathsf{UpdateLocal}$ occurs locally; thus, the server knows nothing during $\mathsf{UpdateLocal}$.
As for the $\mathsf{Update}$, the server sees that the client downloads the $Rcount\%b$-th encrypted bucket and uploads a new encrypted bucket to replace it.
When updating the first bucket (i.e., $Rcount\%b=0$), the server knows nothing except for the current number of key-value pairs in the database based on the size of the new encrypted bucket.
When updating other buckets, the server knows no additional information.
Thus, $\mathsf{Setup}$, $\mathsf{Query}$, $\mathsf{UpdateLocal}$, and $\mathsf{Update}$ all satisfy the leakage function in the security goals.
We provide the formal security proof in Appendix \ref{app:security_proof}.

We note that the restricted zero-leakage security does not require hiding the time and number of operations, as discussed in Section \ref{sec:zero_leakage}.
The leakage functions also do not capture information about the timing and number of operations~\cite{DBLP:journals/iacr/BahiKMN24}.




\subsection{Discussion}
\label{sec:discussion}
Since ORAM and OMAP offer restricted zero-leakage and PIR-like functionality, we examine replacing PIR with ORAM in \textit{Aspire} and compare the \textit{Aspire} instances with OMAP. We then describe how to extend \textit{Aspire} to achieve ideal zero-leakage security and present a technique for \textit{Aspire} to support searchable symmetric encryption (SSE).

\noindent \textbf{Replacing PIR with ORAM.} The PIR in \textit{Aspire} can be replaced by ORAM while maintaining the same functionality and security.
In the $\mathsf{Setup}$ and $\mathsf{Update}$, instead of initializing PIRs, the client and server can run ORAM instances for each bucket.
Also, in the $\mathsf{Query}$, the retrieval of values by PIR can be replaced by ORAM accesses.
Since ORAM does not reveal any additional information beyond what is disclosed by PIR during data access, substituting PIR with ORAM in \textit{Aspire} yields the same leakage function.

However, the widely used tree-based ORAM scheme needs $O(\log n)$ rounds of communication, where $n$ is the size of the database, or it requires the client to store a position map, which needs $O(n)$ storage~\cite{DBLP:journals/jacm/StefanovDSCFRYD18,DBLP:journals/pvldb/CaoFLZFWXYR24}. 
As we show in Section \ref{sec:evaluation}, an \textit{Aspire} instance replacing PIR with PathORAM~\cite{DBLP:journals/jacm/StefanovDSCFRYD18} (denoted as \textit{Aspire-PO}), which includes a position map, dramatically increases the client's storage (104.4 MB when using PathORAM vs. 8.4 MB when using PIR for $2^{23}$ key-value pairs).
Additionally, the ORAM requires $O(\log n)$ additional communication.
In terms of computational complexity, an ORAM access avoids linear cost, while PIR requires a linear scan to answer a query.
We demonstrate that, in \textit{Aspire}, the computation of the PIR query can be parallelized between buckets; ultimately, the total time for utilizing ORAM and PIR is comparable. 

\noindent \textbf{Comparison to OMAP.} 
Oblivious map (OMAP) is designed to facilitate oblivious access to key-value stores~\cite{DBLP:journals/pvldb/CaoFLZFWXYR24}.
It also achieves restricted zero-leakage security. 
Comparing with \textit{Aspire}, the state-of-the-art OMAP~\cite{DBLP:journals/pvldb/CaoFLZFWXYR24} requires multiple rounds of communication and more bandwidth. 
Additionally, most OMAP schemes do not support batch queries~\cite{DBLP:journals/pvldb/CaoFLZFWXYR24}, while \textit{Aspire} can amortize the communication and computation with batch queries for a set of keywords. 
We provide detailed experimental results in Section \ref{sec:evaluation}.

\noindent\textbf{Ideal zero-leakage security for \textit{Aspire}.} 
Recall that \textit{Aspire} is designed for restricted zero-leakage security. 
It can be straightforwardly extended for ideal zero-leakage security by incorporating metadata-hiding techniques (see \textit{Aspire} with metadata-hiding in Figure \ref{fig:summary_of_notion}).
To conceal the true size of the database, a client may choose a fixed, predefined size parameter and maintain it as an invariant across all subsequent $\mathsf{Update}$ operations. 
This parameter thus becomes a public value; however, it is impractical once the database exceeds this limit. Then the client must store overflow items locally.
Likewise, if the database shrinks, the impracticality arises from having many unnecessary ``dummy'' items. 
To hide the metadata associated with $\mathsf{Query}$, the client predefines a query schedule during the setup that is independent of the actual workload. 
For instance, the client may issue either real or dummy $\mathsf{Query}$ at a fixed frequency. Similar to $\mathsf{Query}$, the client establishes an independent, fixed schedule for $\mathsf{Update}$.
A drawback of this approach is that it inevitably leads to a substantial increase in both communication overhead and latency.

We refer the reader to \cite{DBLP:journals/iacr/BahiKMN24} for detailed techniques related to the schedule.
Note that by executing $\mathsf{Query}$ and $\mathsf{Update}$ operations according to fixed, mutually independent schedules, we can prevent the server from inferring any additional information from the two operation types at each specific time interval. 
In this setting, the server’s visibility is restricted exclusively to the public parameters, thereby ensuring that \textit{Aspire} achieves ideal zero-leakage security.

\noindent \textbf{Extending to SSE.} SSE enables a client to search the encrypted database $\{(ind_1,W_1),...\}$ outsourced to the server \cite{DBLP:conf/ccs/Bost16,DBLP:conf/ccs/BostMO17,DBLP:conf/ccs/CashGPR15}, where $ind$ is the identity of a file and $W$ is the set of keywords related to the file.
The client can search for a keyword in the database and retrieve the corresponding file identities from the server.

One can extend \textit{Aspire} to SSE through the following steps:
\begin{itemize}
    \item The client extracts an inverted index $\{(k_1,(ind^1_1,...,ind_k^1)),$ $...,(k_x,(ind^x_1,...,$ $ind^x_l))\}$, where $ind^i_j$ is the identity of a file related to the keyword $k_i$.
    The client maintains a count table $\{(k_1,c_1),$ $...,$ $(k_x,c_x)\}$, where $c_i$ records how many files are related to the keyword $k_i$.
    \item In the setup, the client generates $c_i$ pairs $\{(k_i||1,$ $ind^i_1)$ $,...,$ $(k_i||c_i,$ $ind_i^{c_i})\}$ for every keyword $k_i$. The client collects all the pairs and uses the $\mathsf{Setup}$ of \textit{Aspire} to initialize.
    \item To search for a keyword $k_i$, the client uses \textit{Aspire}'s $\mathsf{Query}$ to retrieve all values for $\{k_i||1,\dots,k_i||c_i\}$.
    \item To add a $ind$ related to $k_i$, the client can simply increase the count $c_i$ and add the pair $(k_i||c_i,ind)$ via the $\mathsf{UpdateLocal}$. 
    Deleting an identity $ind$ related to a keyword $k_i$ is more complicated, as the client does not know the $j$ of $(k_i||j,ind)$ to be deleted.
    To delete an identity related to $k_i$, the client stashes the deletion until the next time searching for $k_i$, so that the client can adjust the key-value pair related to $k_i$ and delete the identity. 
    This kind of lazy deletion is widely used in SSE~\cite{DBLP:conf/ccs/ChamaniPPJ18}.
    
\end{itemize}

This extension achieves forward and the strongest backward security of SSE~\cite{DBLP:conf/ccs/BostMO17}. 
The setup and update will only reveal the number of key-identity pairs in the database.
If the maximum count of all keywords does not exceed the $2b/3$, i.e., if searching for any keyword requires a single $\mathsf{Query}$ in \textit{Aspire}, then the search will reveal nothing to the server.

\section{Experimental Evaluation}
We evaluate \textit{Aspire} based on PIR across multiple dimensions and compare it against ORAM~\cite{DBLP:journals/jacm/StefanovDSCFRYD18} and OMAP~\cite{DBLP:journals/pvldb/CaoFLZFWXYR24}. 
Furthermore, we simulate \textit{Aspire} based on ORAM instead of PIR to demonstrate the extensibility of our design.
The source code accompanying this work is publicly accessible at \href{https://doi.org/10.5281/zenodo.22813536}{https://doi.org/10.5281/zenodo.22813536}.
\subsection{Experimental Setup}

\noindent \textbf{Environment.}
Our implementation is developed in C++ and evaluated on an Ubuntu 22.04.5 system equipped with an AMD EPYC 9354 32-core processor and 1 TB of main memory. 
To emulate the network conditions between the client and the server, we configure a communication channel with a bandwidth of 30 Mbps and a latency of 20 ms.
{We implemented AES-128 using Crypto++ with a fresh random IV for each encryption, making the encryption non-deterministic.}

\noindent \textbf{Dataset and query.}
We use a synthetic dataset under varying configurations, including the number of key-value pairs, the bit length of values, and the number of queries per batch.
{Synthetic data suffices because Aspire’s performance depends on database size, value length, and batch size, rather than plaintext contents or keyword distributions.}
Queries are generated by randomly sampling keywords from the dataset.
Unless otherwise stated, the dataset is configured to contain $2^{22}$ key-value pairs, with 32-bit keywords and 2048-bit (256 B) values {(about 1 GB of raw data)}. The {default} number of buckets is set to 512 {(approximately 6 MB per bucket)}, which supports batches of up to 341 queries. {We note that batched queries are representative of practical encrypted key–value workloads. For example, RocksDB\cite{rocksdb} supports batched retrieval through MultiGet, and SSE may translate one keyword search into thousands of key–value lookups (Section \ref{app:kvsystems}).}

\noindent \textbf{Implementation and baselines.}
We implement \textit{Aspire} based on the single-query PIR scheme Spiral~\cite{DBLP:conf/sp/MenonW22} (denoted as \textit{Aspire-S}), as it compresses the query and response to achieve communication efficiency. 
To adopt the advantages of batch PIR, we also implement \textit{Aspire} based on PIRANA~\cite{DBLP:conf/sp/LiuLWR24} (denoted as \textit{Aspire-P}), which achieves fast batch PIR queries. 
For a better understanding of the difference between \textit{Aspire} based on ORAM and PIR, we also replace the PIR in \textit{Aspire} with PathORAM~\cite{DBLP:journals/jacm/StefanovDSCFRYD18} (denoted as \textit{Aspire-PO}). 
We note that \textit{Aspire-PO} still requires a position map for each bucket, even with an MPH to document the index of each keyword.
As the PIRANA server encodes the data in each batch together, updating a bucket will cause a re-encoding of all buckets. 
Thus, in $\mathsf{Update}$, the server re-runs the setup of PIRANA for all the buckets.
For Spiral and PIRANA, we utilize the recommended parameters from their original work. We use 3 hash functions for the PBC and set the bucket number to be no less than 1.5 times the number of queries in a batch.

For baselines, we implement a key-value database based on PathORAM~\cite{DBLP:journals/jacm/StefanovDSCFRYD18} (denoted as PathORAM), in which the client stores a keyword-to-index mapping and uses the PathORAM to retrieve the corresponding values. 
Since PathORAM~\cite{DBLP:journals/jacm/StefanovDSCFRYD18} does not specify a setup protocol, we let the client first construct the data structure locally and subsequently upload it to the server during the initialization of PathORAM and \textit{Aspire-PO}.
As OMAP supports keyword querying over key-value databases, we also experiment with the state-of-the-art OMAP scheme~\cite{DBLP:journals/pvldb/CaoFLZFWXYR24} (denoted as OMAP24).  

To mitigate the computational bottleneck of linear scanning in PIR querying, we parallelize {the servers of} \textit{Aspire-S}, \textit{Aspire-P}, and \textit{Aspire-PO} using 16 threads, denoted as \textit{Aspire-S} Parallel, \textit{Aspire-P} Parallel, and \textit{Aspire-PO} Parallel. 
This parallelization occurs naturally because our PBC-based batch query structure issues independent queries to each bucket, thereby supporting concurrent execution. 
Conversely, standard PathORAM and OMAP24 do not provide the mechanisms necessary to support parallelism during the query phase, as each individual access requires a subsequent reshuffling step, thereby preventing concurrent operations.

\subsection{Evaluation}
\label{sec:evaluation}

\begin{figure}[tp]
    \centering
    \includegraphics[width=0.6\linewidth]{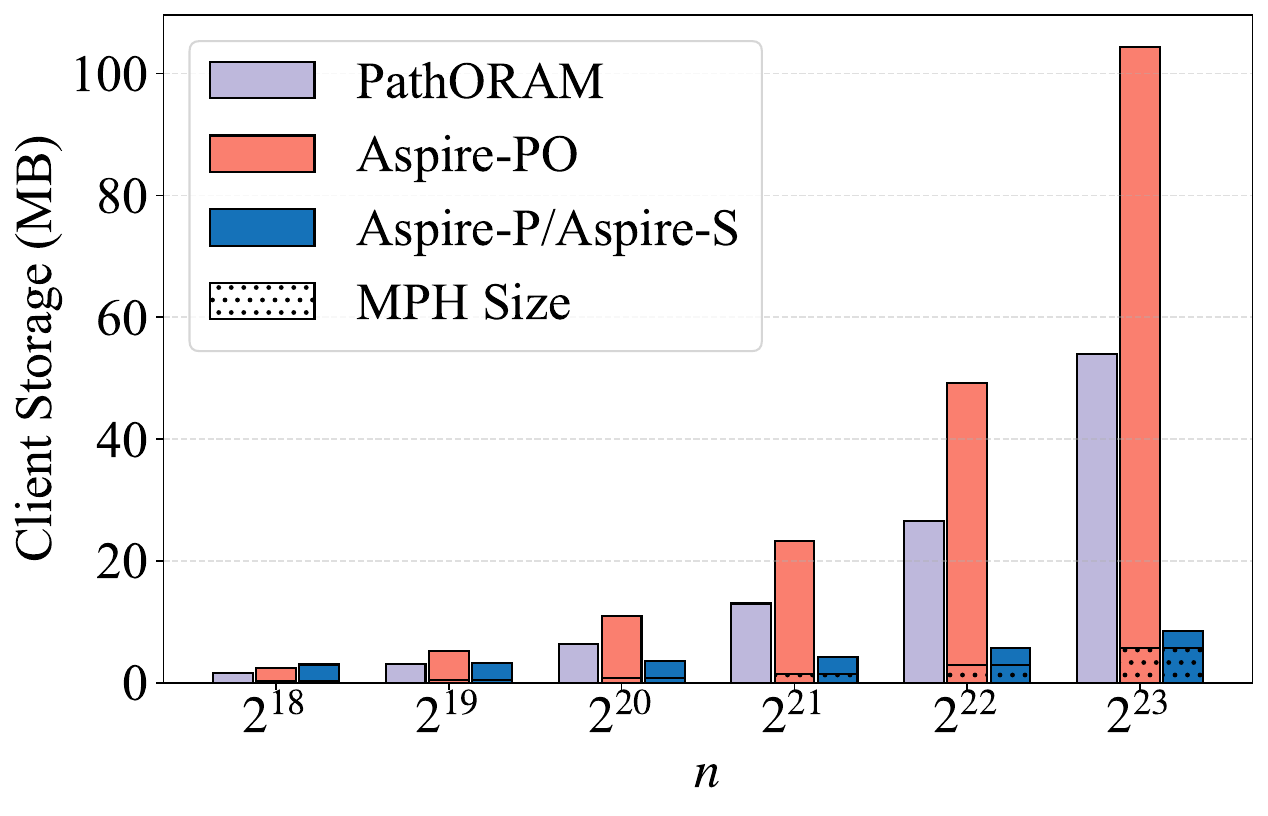}
    \caption{Local storage with different numbers of key-value pairs. }
    \label{fig:locol_storage_size}
\end{figure}

\begin{figure}[tbp]
    \centering
    \subfloat[Communication Cost]
	{
 \label{fig:setup_a}
		\begin{minipage}[b]{0.23\textwidth}
        \centering
        \includegraphics[width=\textwidth]{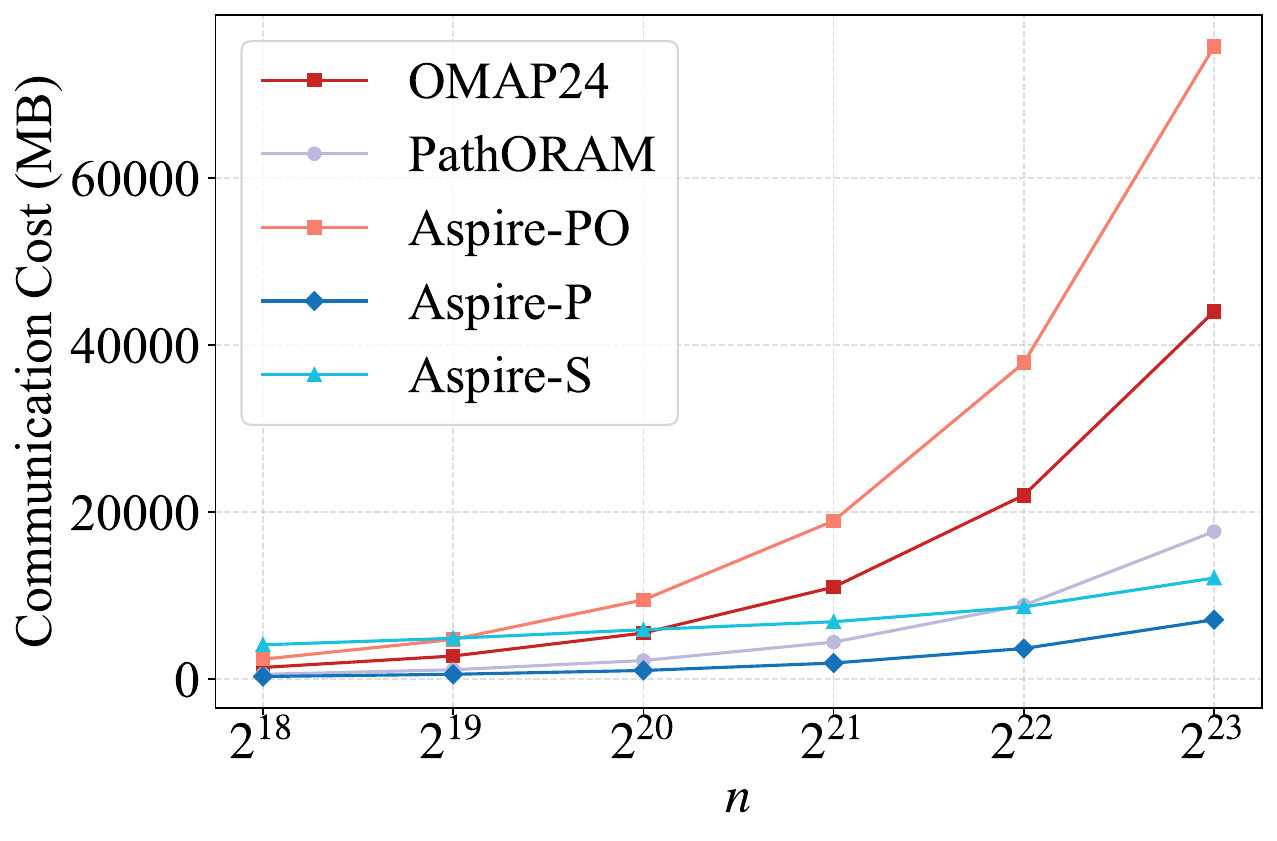}
    \end{minipage}
	}
    \subfloat[Computation Cost]
	{
 \label{fig:setup_b}
		\begin{minipage}[b]{0.23\textwidth}
        \centering
        \includegraphics[width=\textwidth]{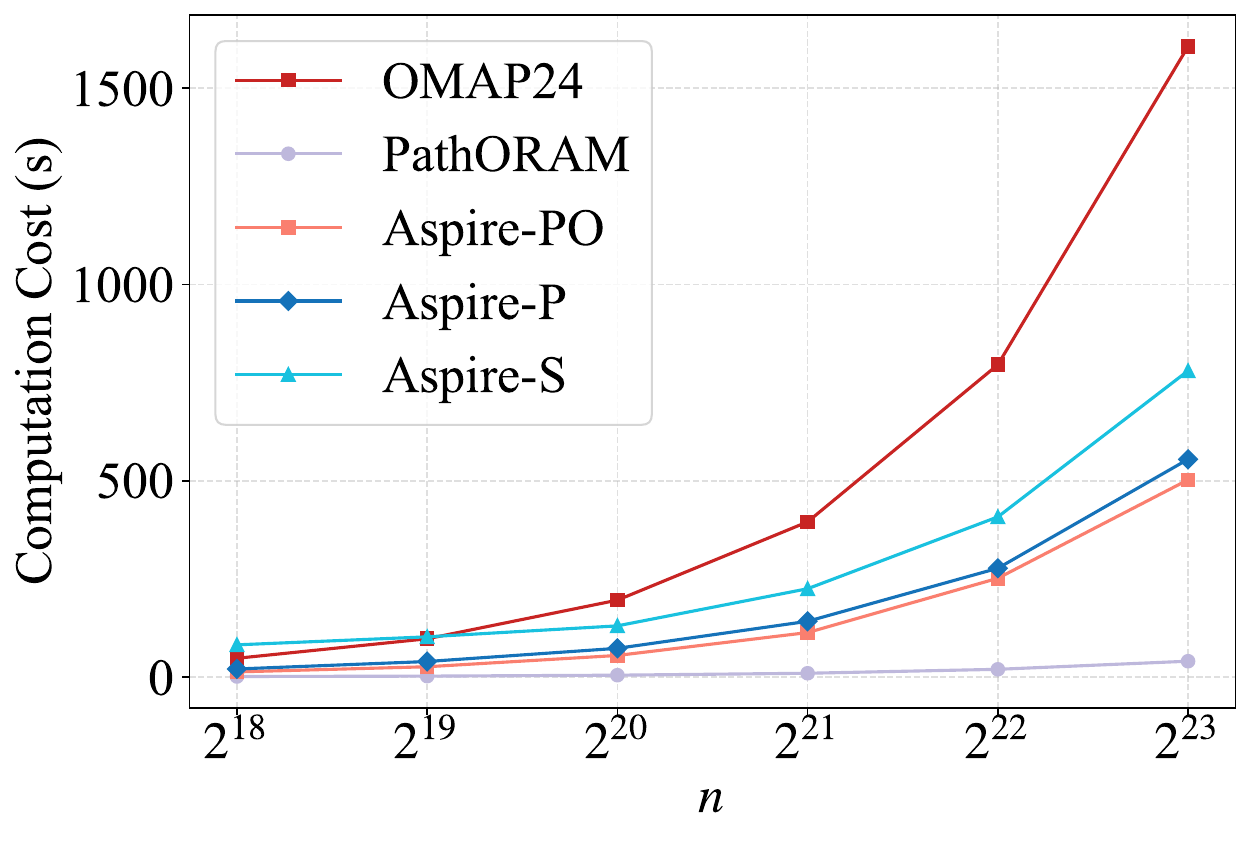}
    \end{minipage}
	}
    
    \caption{Setup performance with different numbers of key-value pairs.}
    \label{fig:setup}
\end{figure}

\begin{figure*}[tbp]
    \centering
     \subfloat[Communication Cost]
	{
 \label{fig:data_size_a}
		\begin{minipage}[b]{0.23\textwidth}
        \centering
        \includegraphics[width=\textwidth]{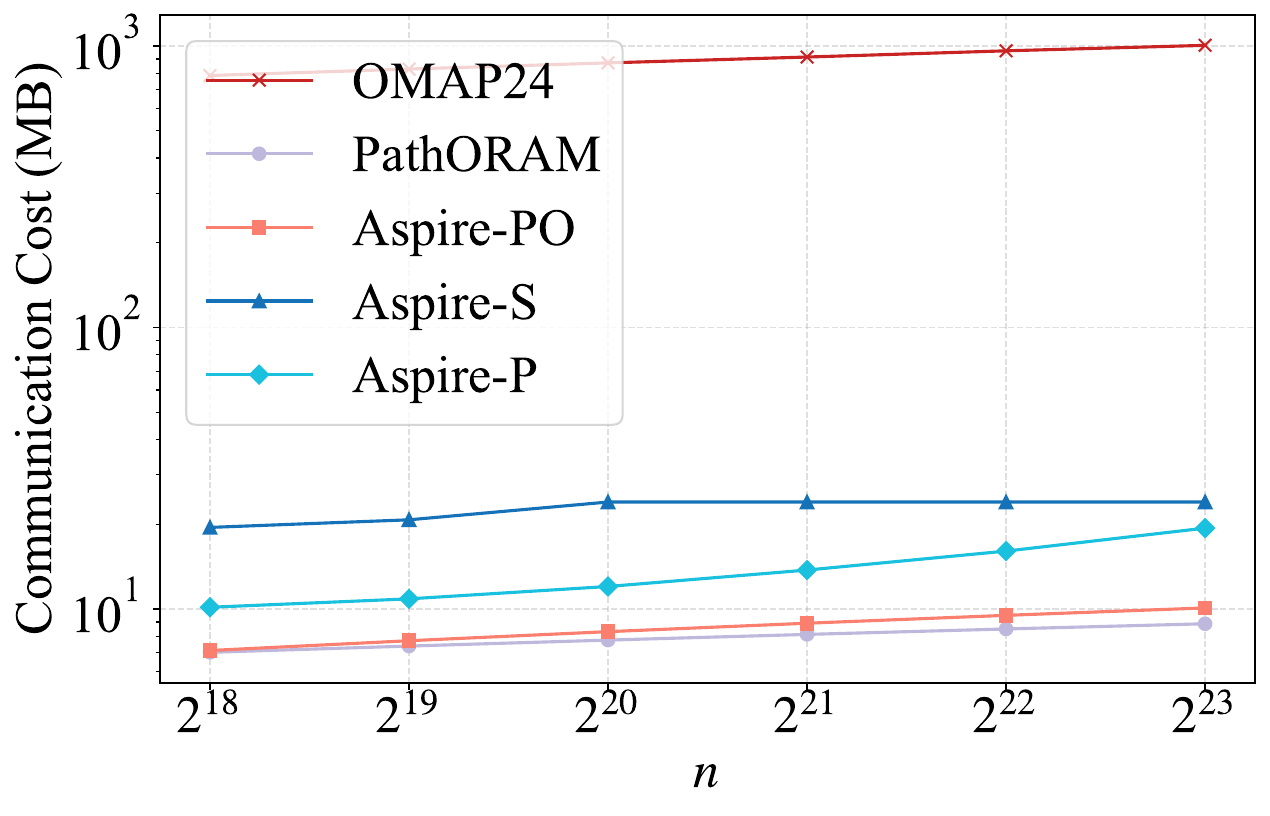}
    \end{minipage}
	}
    \subfloat[Computation Cost]
	{
 \label{fig:data_size_b}
		\begin{minipage}[b]{0.23\textwidth}
        \centering
        \includegraphics[width=\textwidth]{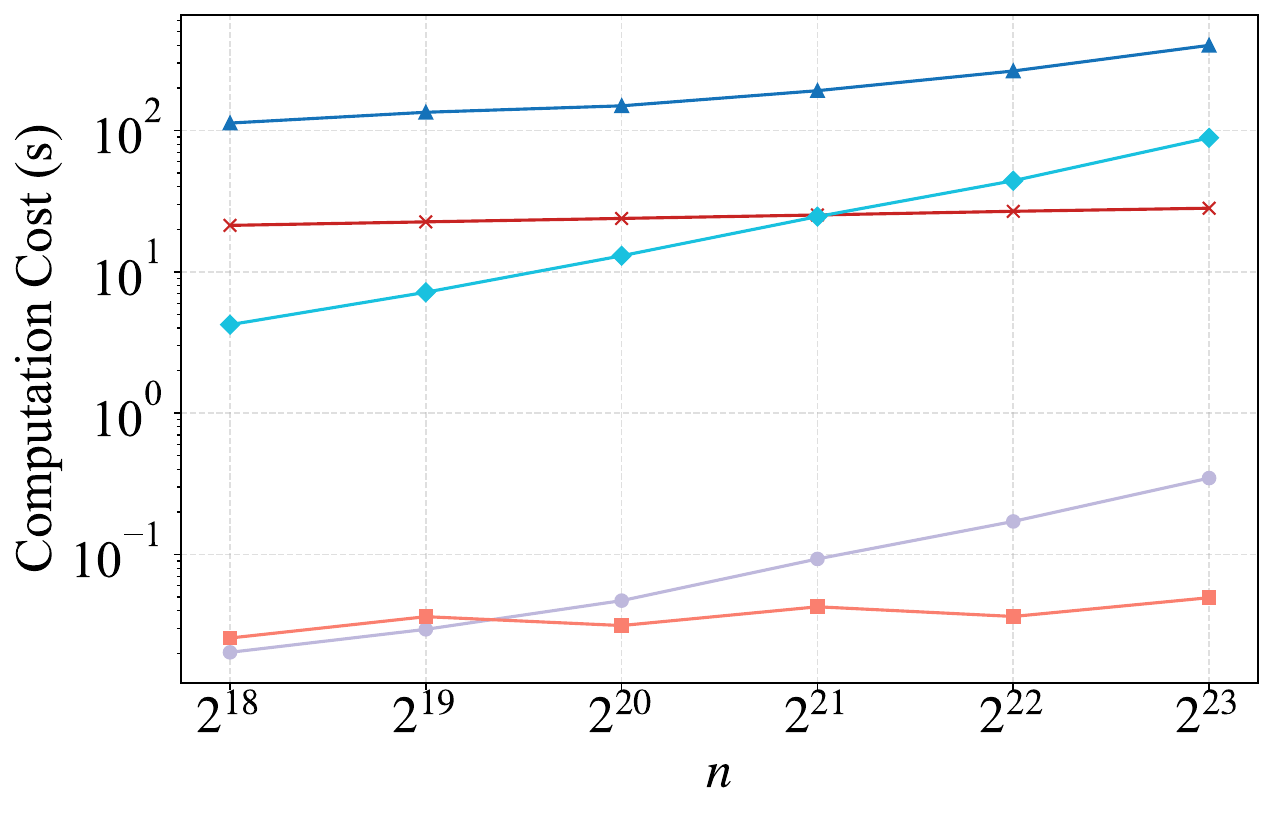}
    \end{minipage}
	}
    \subfloat[Total Time]
	{
 \label{fig:data_size_c}
		\begin{minipage}[b]{0.23\textwidth}
        \centering
        \includegraphics[width=\textwidth]{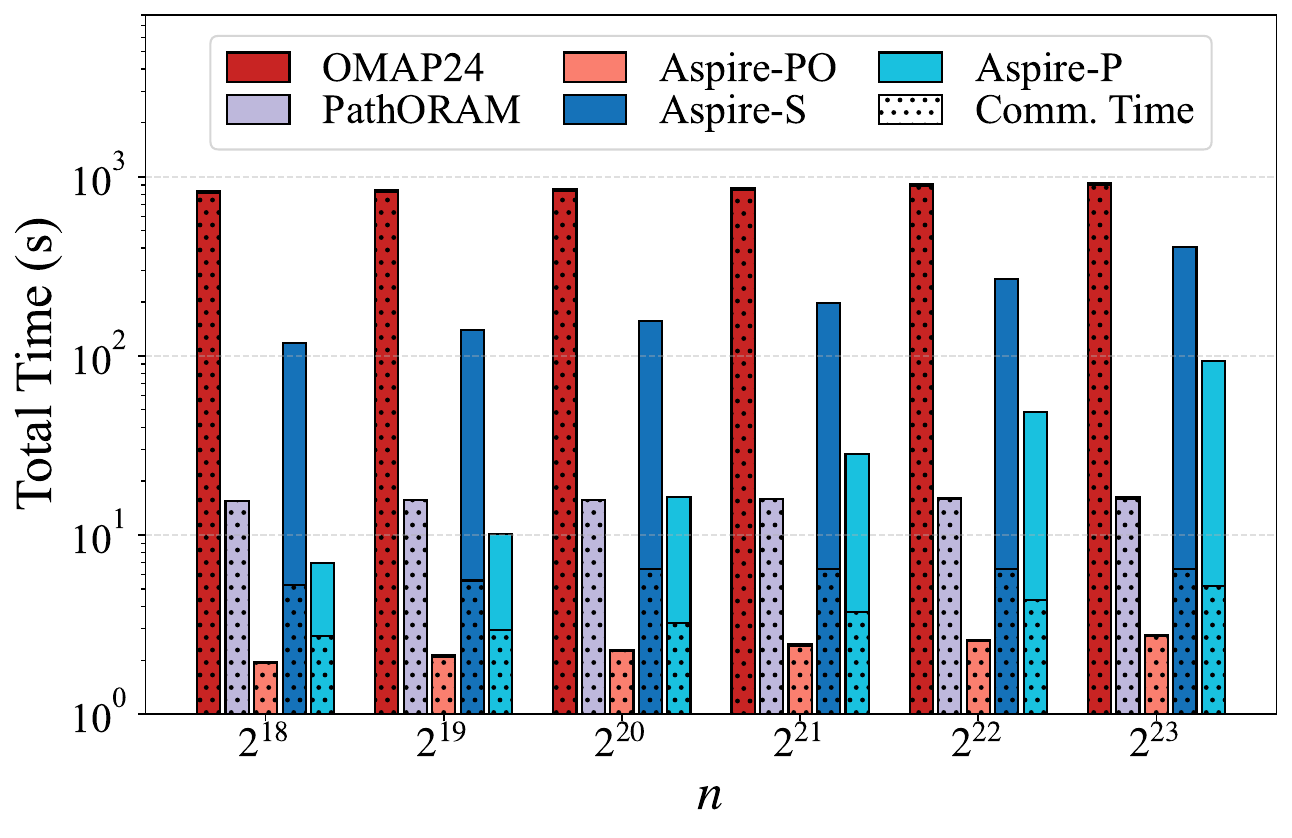}
    \end{minipage}
	}
    \subfloat[Total Time (Parallel)]
	{
 \label{fig:data_size_d}
		\begin{minipage}[b]{0.23\textwidth}
        \centering
        \includegraphics[width=\textwidth]{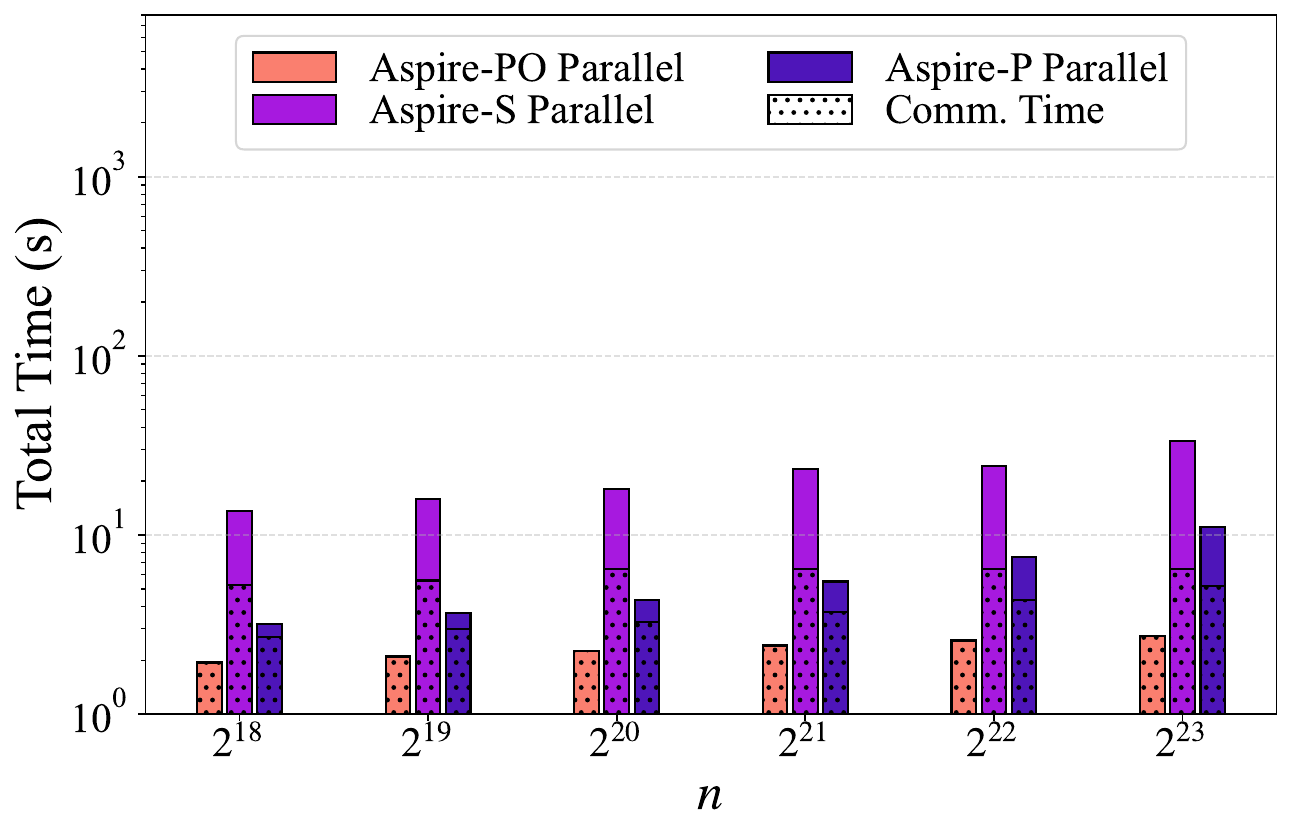}
    \end{minipage}
	}
    
    \caption{Query performance with different numbers of key-value pairs.}
    \label{fig:data_size}
\end{figure*}

\subsubsection{Local storage}
\label{sec:local_storage}
We vary the database size from $2^{18}$ to $2^{23}$ key-value pairs to evaluate and compare the local storage requirements of PathORAM, \textit{Aspire-PO}, \textit{Aspire-S}, and \textit{Aspire-P}.
The OMAP24 construction is characterized by minimal local storage requirements of $O(\log n)$. 
Although its local storage is not depicted in the figure, it exhibits substantial communication overhead and necessitates a large number of communication rounds to support queries, as we show later.
We assume $a=10$ for \textit{Aspire}, i.e., 10 $\mathsf{UpdateLocal}$ per $\mathsf{Update}$ for \textit{Aspire-S}, \textit{Aspire-P}, and \textit{Aspire-PO}.
We note that both PIRANA and Spiral require only constant local storage for the secret key, and \textit{Aspire-P} and \textit{Aspire-S} have the same local storage for the MPH and the stash. 
The results are shown in Figure \ref{fig:locol_storage_size}. 
For \textit{Aspire}, the MPH size increases as the database size grows but remains small. 
Even with $2^{23}$ key-value pairs, the MPH size is still less than 6 MB. 
The stash size of \textit{Aspire-P} and \textit{Aspire-S} changes little and remains small as the database scales. 
This is consistent with the theoretical analysis of the stash size as $O(ab)$, which is not relevant to the database size.
PathORAM needs to store a mapping between the keywords and the path as the position map, which has a size of $(|k|+\log n)n$. 
The \textit{Aspire-PO} also needs to store an MPH and a position map for each bucket.
The local storage size of PathORAM and \textit{Aspire-PO} grows rapidly as the database size increases, reaching over 50 MB and 100 MB, respectively, when the database size is $2^{23}$, while \textit{Aspire-S} and \textit{Aspire-P} take less than 8.5 MB of local storage.

\subsubsection{Setup performance} We assess the setup overhead for a range of database sizes, specifically from $2^{18}$ to $2^{23}$ key–value pairs.
The results are presented in Figure \ref{fig:setup}. 
When $n = 2^{23}$, the communication costs for \textit{Aspire-S} and \textit{Aspire-P} are approximately 12.1 GB and 7.1 GB, respectively.
In contrast, the communication costs for PathORAM and \textit{Aspire-PO} are higher. This is because PathORAM requires a substantial number of dummy items to reduce the probability of stash overflow. Moreover, \textit{Aspire-PO} replicates each key–value pair three times, and all buckets are padded to the same size, resulting in communication overheads for \textit{Aspire-PO} that are more than three times larger than those of PathORAM. The communication cost of OMAP24 is also considerable, amounting to approximately 44 GB when $n = 2^{23}$.

The computation times of \textit{Aspire-PO}, \textit{Aspire-P}, and \textit{Aspire-S} range from 500 s to 780 s when $n = 2^{23}$.  PathORAM requires only 40 s, as it does not replicate items or construct the MPH. The computation cost for OMAP24 is also substantial, exceeding 1,600 s.

\begin{figure*}[tbp]
    \centering

    {
        \scriptsize 
        \itshape      
        \color{gray}   
    Note: Aspire-S results are omitted for large numbers of buckets ($b \ge 2^{12}$) because the Spiral framework lacks corresponding parameters.
    }
    \vspace{-3mm}

    \subfloat[Communication Cost]
	{
 \label{fig:batch_size_a}
		\begin{minipage}[b]{0.23\textwidth}
        \centering
        \includegraphics[width=\textwidth]{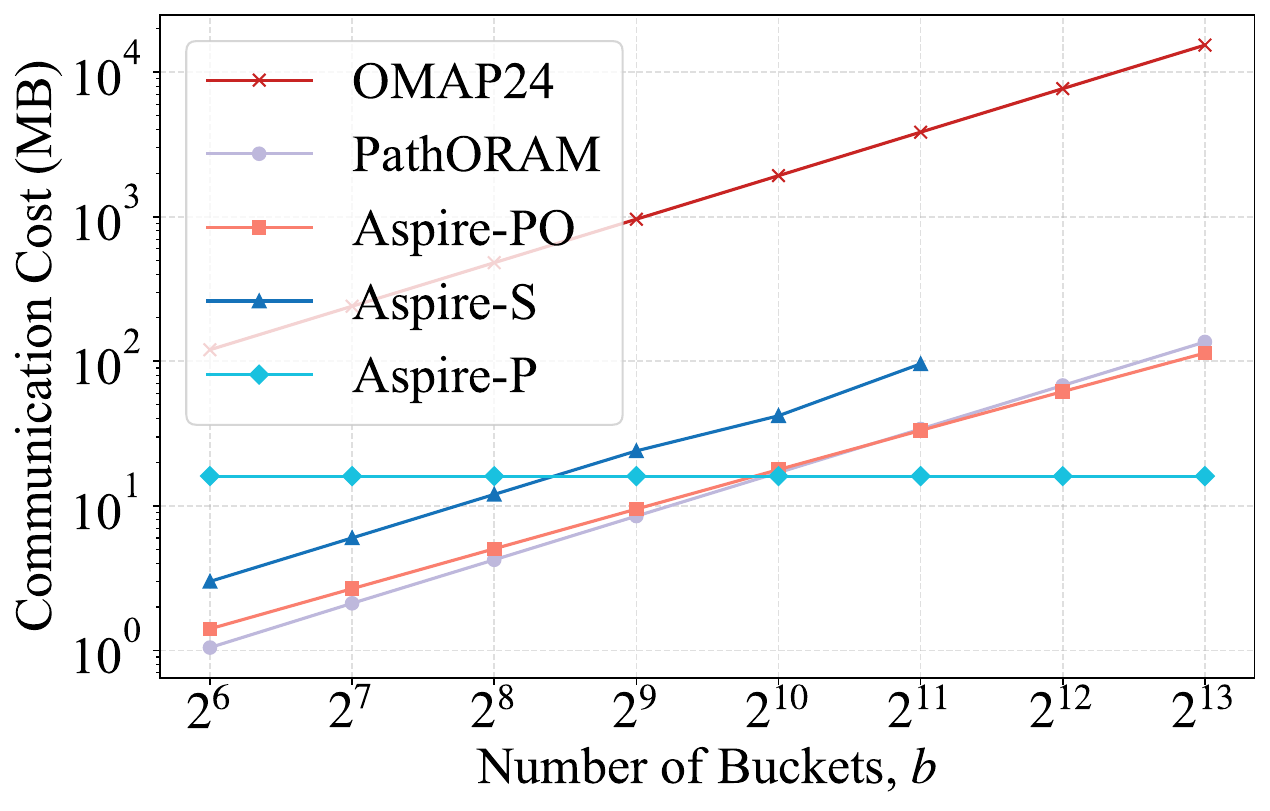}
    \end{minipage}
	}
    \subfloat[Computation Cost]
	{
 \label{fig:batch_size_b}
		\begin{minipage}[b]{0.23\textwidth}
        \centering
        \includegraphics[width=\textwidth]{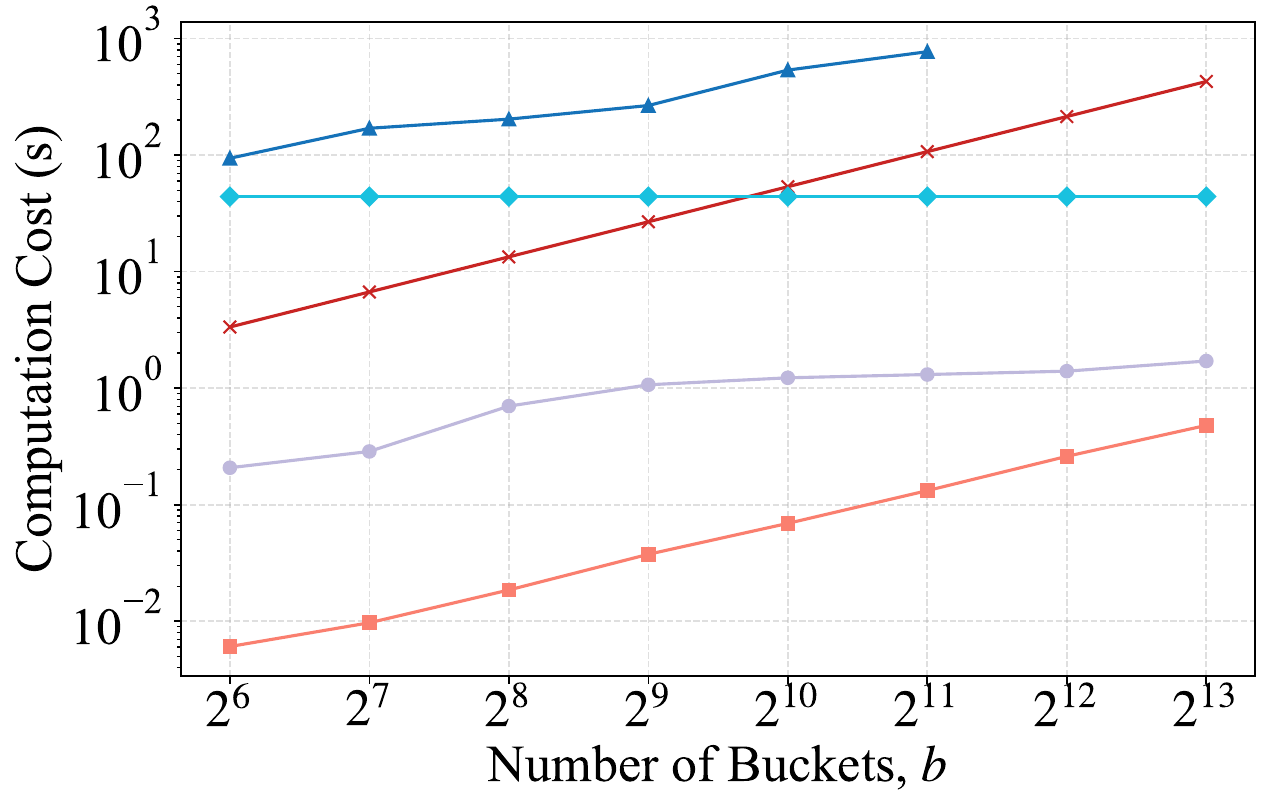}
    \end{minipage}
	}
    \subfloat[Total Time]
	{
 \label{fig:batch_size_c}
		\begin{minipage}[b]{0.23\textwidth}
        \centering
        \includegraphics[width=\textwidth]{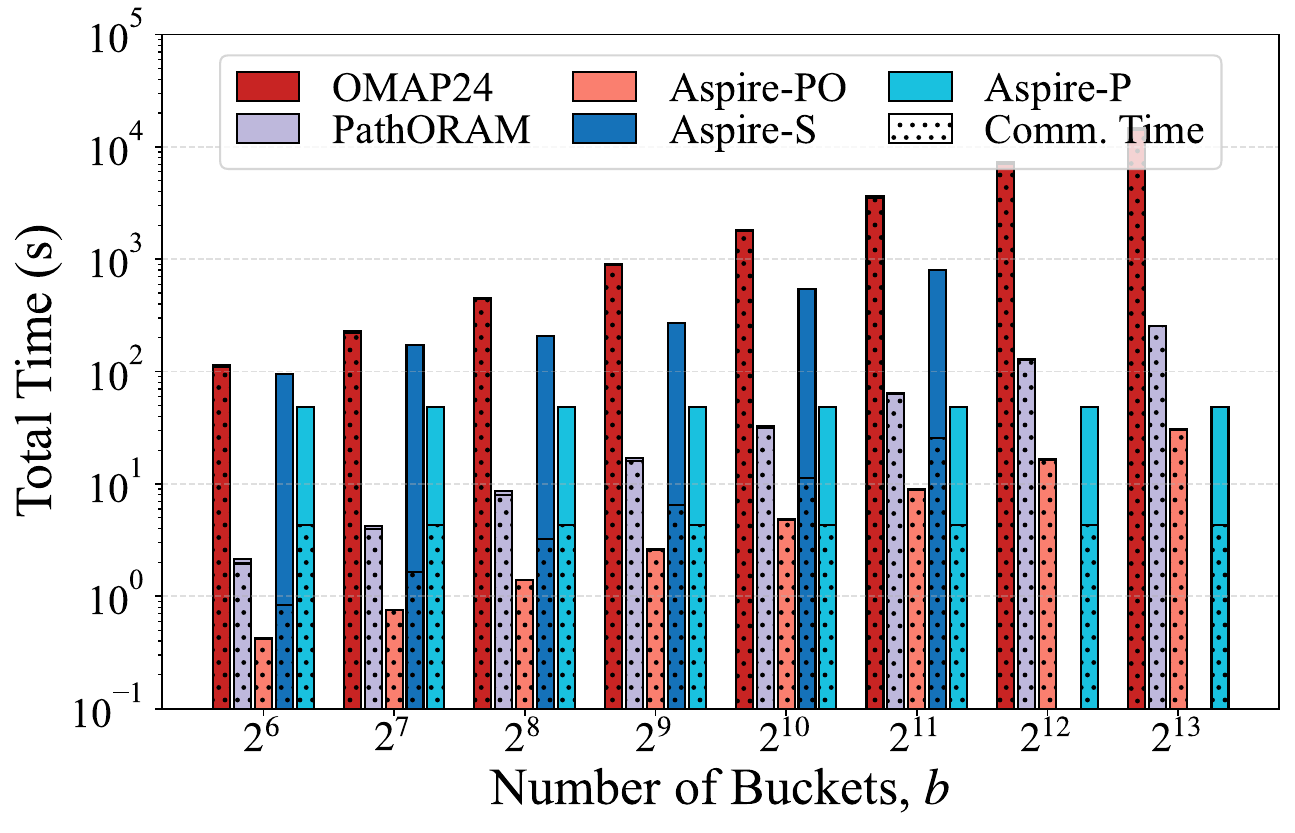}
    \end{minipage}
	}
    \subfloat[Total Time (Parallel)]
	{
 \label{fig:batch_size_d}
		\begin{minipage}[b]{0.23\textwidth}
        \centering
        \includegraphics[width=\textwidth]{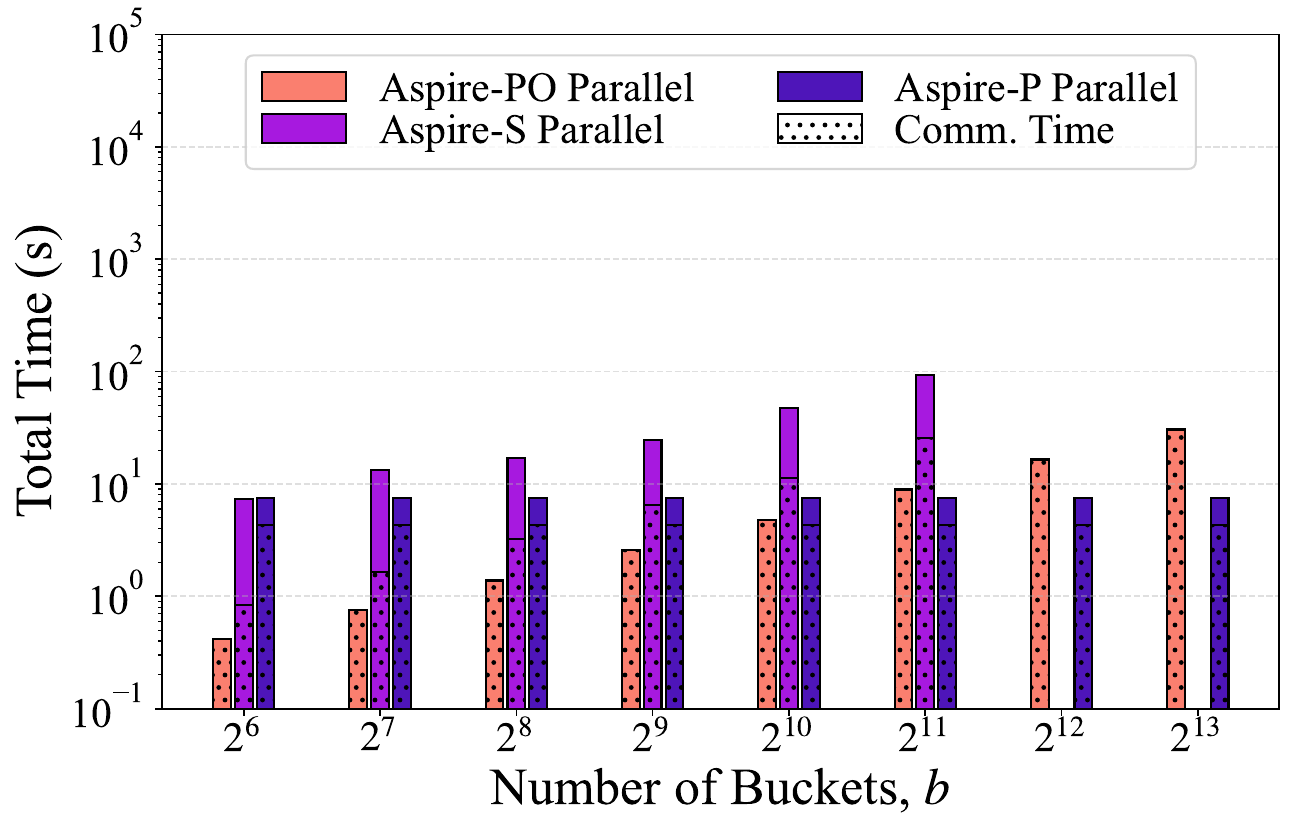}
    \end{minipage}
	}
    \caption{Query performance with different numbers of buckets.}
    \label{fig:batch_size}
\end{figure*}

\subsubsection{Query performance} To provide a comprehensive assessment of query efficiency, we evaluate performance along three dimensions: communication overhead, computational overhead, and end-to-end query latency. 
We conduct a series of experiments in which we systematically vary system parameters in three ways: (1) scaling the database size $n$ from $2^{18}$ to $2^{23}$ key-value pairs to assess scalability; (2) increasing the number of buckets $b$ from $2^6$ to $2^{13}$ (corresponding to batch queries over 42 to 5461 keywords); and (3) varying the value length $|v|$ from $2^10$ bits (128 bytes) to $2^{16}$ bits (8192 bytes). When studying the impact of one parameter, all remaining parameters are fixed at their default values.
To enable a fair comparison, for baseline schemes that do not employ PBC (namely, PathORAM and OMAP24), we perform $2b/3$ independent queries so that the total number of keywords queried matches the batch-query capacity of \textit{Aspire} with $b$ buckets.

\textbf{Database size.} We first analyze the impact of the number of key-value pairs $n$ in the database. 
As illustrated in Figure \ref{fig:data_size}, when $n$ increases from $2^{18}$ to $2^{23}$ (with 512 buckets supporting a batch query of 341 keywords), the communication costs of \textit{Aspire-P} and \textit{Aspire-S} remain relatively low, at about 20 MB and 24 MB when $n=2^{23}$. 
This indicates that the communication overhead of \textit{Aspire-S} and \textit{Aspire-P} is only mildly affected by the database size, as each batch PIR query in \textit{Aspire} incurs a nearly fixed communication cost for downloading the answer, while only the query request will slowly increase as $n$ grows. 
Additionally, PathORAM and \textit{Aspire-PO} also have a low communication cost of about 10 MB at $n=2^{23}$, which is consistent with their logarithmic access overhead in $n$.
In contrast, the OMAP24 exhibits explosive growth in communication cost, consuming over 1000 MB at $n=2^{23}$, as it restricts the local storage of the client and takes about 90 communication rounds for each keyword query.  

Regarding computation cost, PathORAM and \textit{Aspire-PO} are faster than PIR-based schemes, whose computation cost grows linearly with $n$.
This linear growth stems from the underlying PIR protocols in \textit{Aspire-P} and \textit{Aspire-S}, which require scanning all database elements to ensure zero-leakage security.
Nevertheless, the computational overhead of \textit{Aspire-P} and \textit{Aspire-S} is predominantly incurred on the server side, while the client incurs only a comparatively lightweight computational cost (see Appendix \ref{app:local_computation_cost}).

The total execution time in Figure \ref{fig:data_size_c} reveals that communication time is the dominant factor for PathORAM, \textit{Aspire-PO}, and OMAP24, which need to transmit multiple items in the database, while computation time accounts for at least half of the total latency for \textit{Aspire-S} and \textit{Aspire-P}. 
We further evaluate the total time of \textit{Aspire-S} Parallel, \textit{Aspire-P} Parallel, and \textit{Aspire-PO} Parallel under 16-thread parallelization. 
The computation time for \textit{Aspire-P} Parallel is reduced by approximately 15 times to 5.89 seconds at $n = 2^{23}$, compared to the 88.64 seconds required for \textit{Aspire-P}. Consequently, the amortized total time per query for \textit{Aspire-P} Parallel is less than 33 ms at $n = 2^{23}$.

\noindent \textbf{Number of buckets.} 
The number of buckets $b$ controls the number of keywords ($2b/3$) for each batch query. 
We further investigate the scalability of \textit{Aspire} by varying the number of buckets $b$ from $2^6$ to $2^{13}$, while fixing the database size at $n=2^{22}$.
The results of varying $b$ are shown in Figure \ref{fig:batch_size}.
The results show that the performance of \textit{Aspire-P} and \textit{Aspire-P} Parallel changes little, supporting up to $2^{13}$ buckets with about 16 MB of communication and 3.2 s of computation.
This is because \textit{Aspire-P} inherits the characteristics of the underlying PIRANA protocol, thereby enabling support for up to 8192 buckets without incurring any degradation in performance.
However, the cost of \textit{Aspire-S} grows as $b$ increases.
Additionally, Spiral does not provide parameters for small databases. As $b$ increases, each bucket contains fewer elements and exceeds the recommended parameters of Spiral. 

Other schemes' overheads also increase as $b$ increases.
\textit{Aspire-PO} needs to download and shuffle more blocks as $b$ increases, causing higher communication and computation costs. 
The communication cost for PathORAM and  \textit{Aspire-PO} increases linearly from approximately 1 MB to over 113 MB as $b$ grows.
OMAP24 exhibits explosive growth from about 120 MB at $b=2^6$ to about 15,000 MB at $b=2^{13}$ buckets. 
Even with much higher local storage, the total time of \textit{Aspire-PO} Parallel is higher than that of \textit{Aspire-P} Parallel when the bucket number is $2^{11}$. 
In this case, we note that the amortized total time for \textit{Aspire-P} Parallel to answer a query is less than 6 ms.

\begin{figure*}[tbp]
    \centering
     {
        \scriptsize 
        \itshape      
        \color{gray}   
        Note: Aspire-S results are omitted for large values ($|v| \ge 2^{15}$), as these experiments exceed available memory.
    }
    \vspace{-3mm} 
    
    \subfloat[Communication Cost]
	{
 \label{fig:value_size_a}
		\begin{minipage}[b]{0.23\textwidth}
        \centering
        \includegraphics[width=\textwidth]{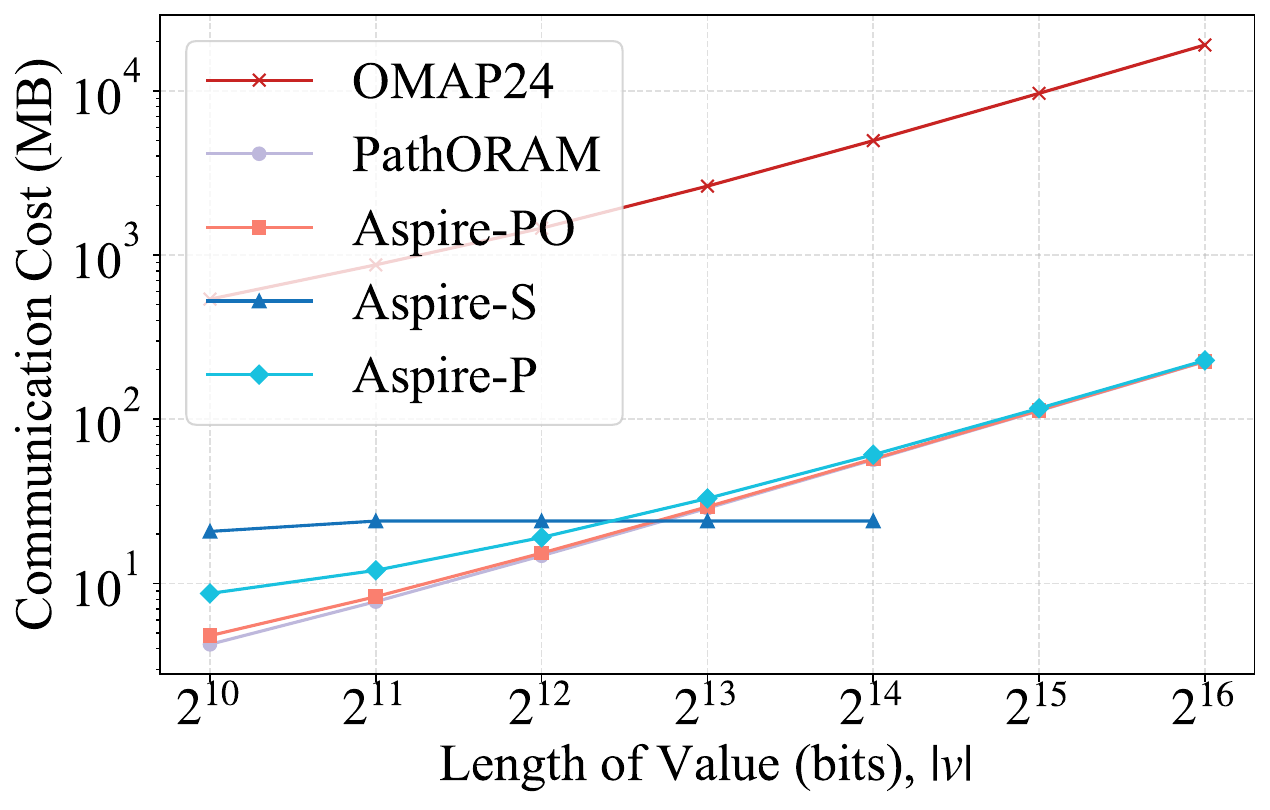}
    \end{minipage}
	}
    \subfloat[Computation Cost]
	{
 \label{fig:value_size_b}
		\begin{minipage}[b]{0.23\textwidth}
        \centering
        \includegraphics[width=\textwidth]{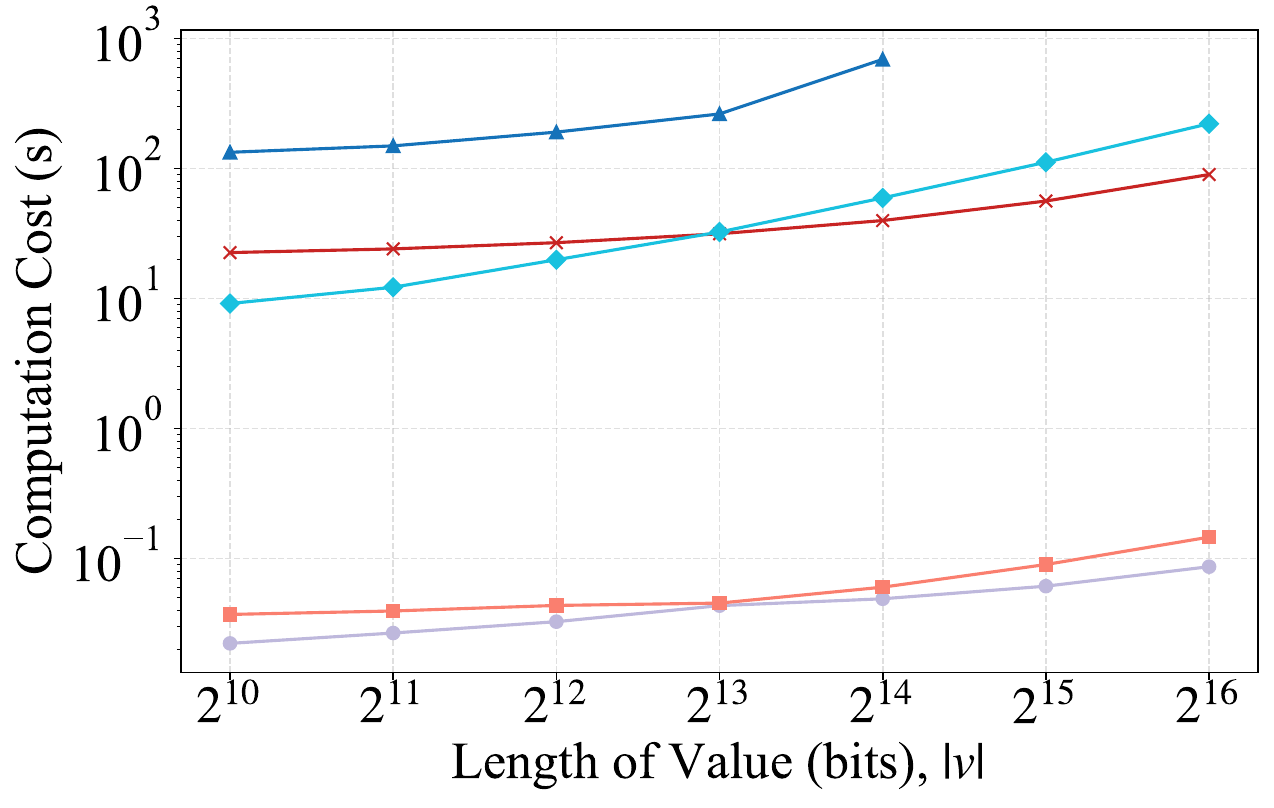}
    \end{minipage}
	}
    \subfloat[Total Time]
	{
 \label{fig:value_size_c}
		\begin{minipage}[b]{0.23\textwidth}
        \centering
        \includegraphics[width=\textwidth]{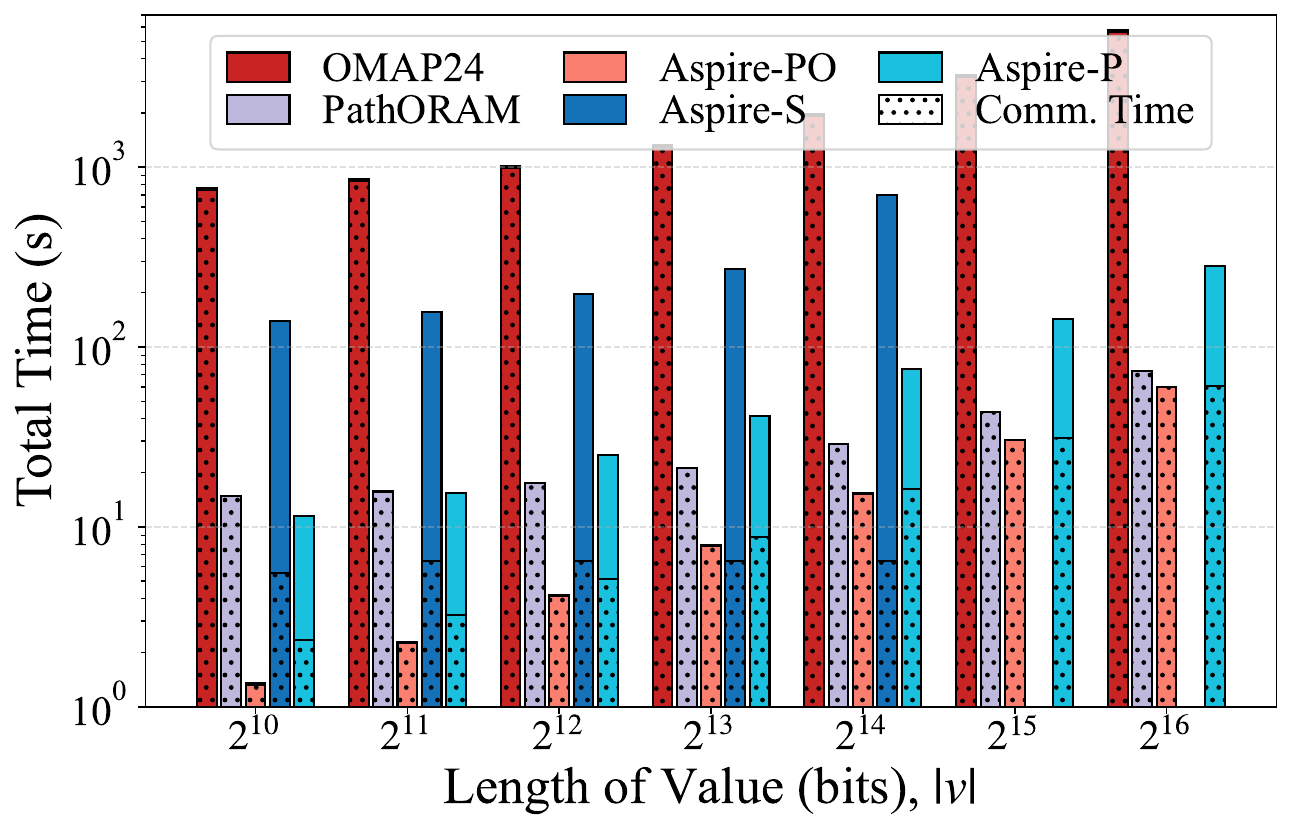}
    \end{minipage}
	}
    \subfloat[Total Time (Parallel)]
	{
 \label{fig:value_size_d}
		\begin{minipage}[b]{0.23\textwidth}
        \centering
        \includegraphics[width=\textwidth]{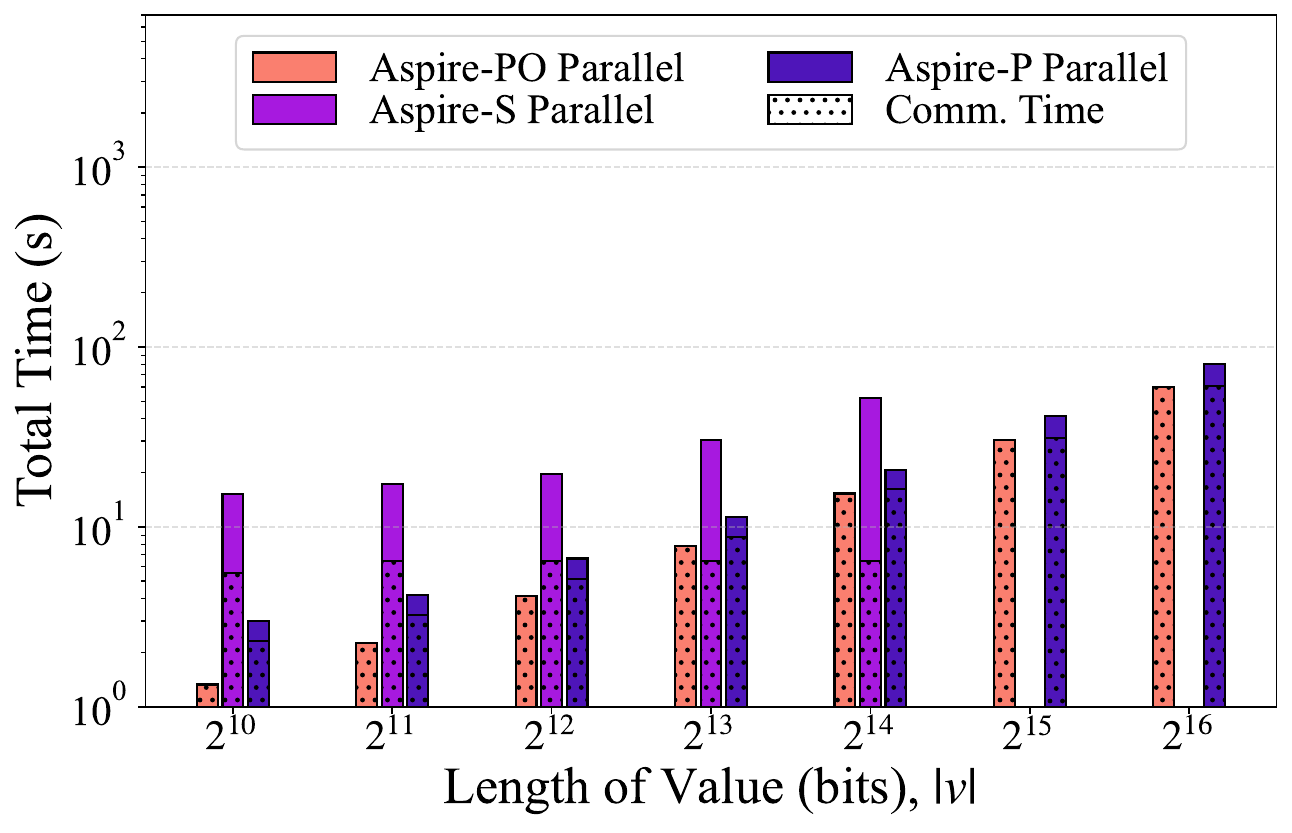}
    \end{minipage}
	}
    \caption{Query performance with different value lengths.}
    \label{fig:value_size}
\end{figure*}

\begin{figure}[tbp]
    \centering
    \subfloat[Communication Cost]
	{
 \label{fig:update_a}
		\begin{minipage}[b]{0.23\textwidth}
        \centering
        \includegraphics[width=\textwidth]{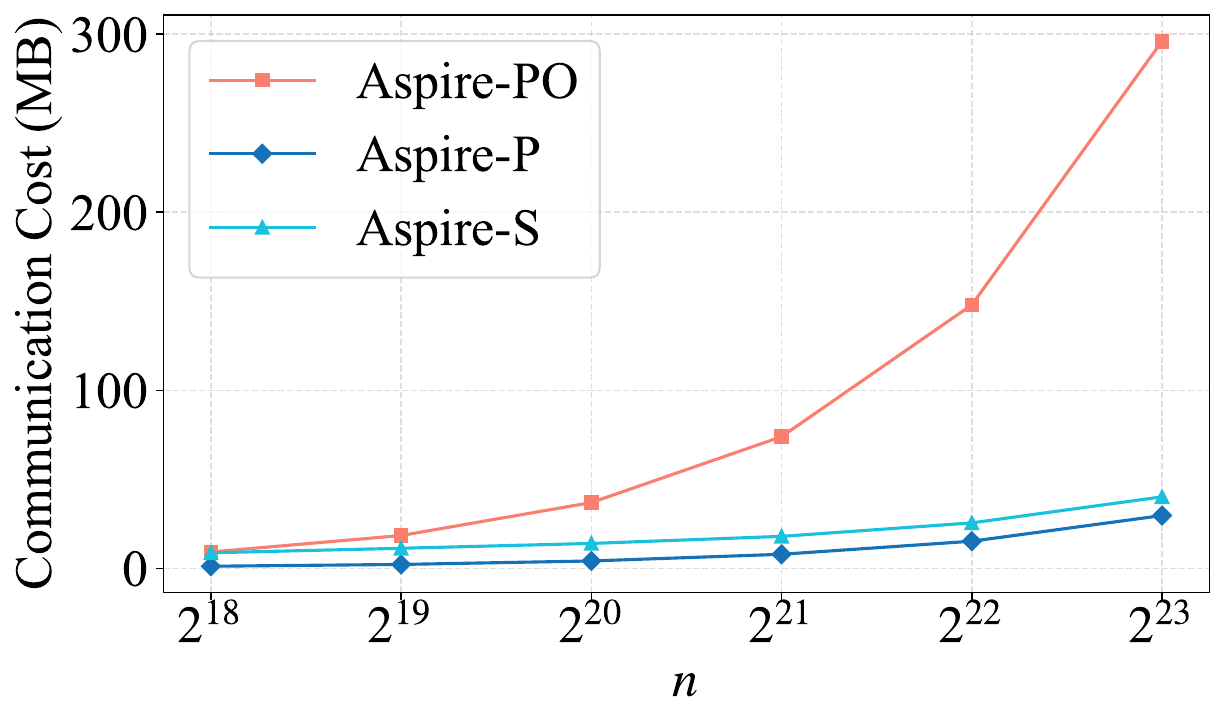}
    \end{minipage}
	}
    \subfloat[Computation Cost]
	{
 \label{fig:update_b}
		\begin{minipage}[b]{0.23\textwidth}
        \centering
        \includegraphics[width=\textwidth]{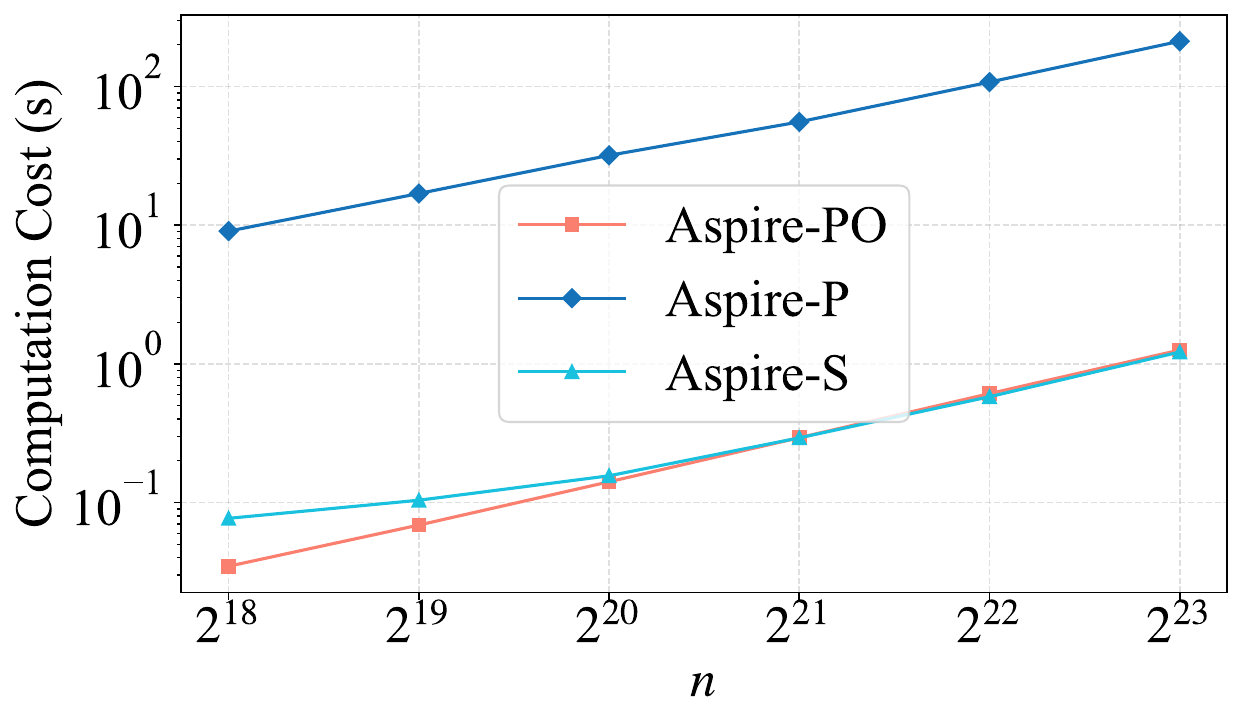}
    \end{minipage}
	}
    
    \caption{Update performance with different numbers of key-value pairs.}
    \label{fig:update}
\end{figure}

\noindent \textbf{Length of value.}
We then set $n=2^{20}$ and the number of buckets to $2^9$, and vary the value length $|v|$ (from $2^{10}$ bits to $2^{16}$ bits). 
The results are shown in Figure \ref{fig:value_size}. 
The \textit{Aspire-S} exhibits nearly constant communication overhead, maintaining a stable cost of approximately 24 MB for value lengths between $2^{11}$ and $2^{14}$ bits. 
This is because the underlying Spiral protocol encapsulates the query result within a ciphertext polynomial that has a fixed capacity of 8192 bytes and does not consider compressing the results from multiple buckets. 
Consequently, as long as the value size does not exceed this 8192-byte limit, the scheme returns the same number of ciphertexts with identical sizes, keeping the communication bandwidth invariant. 
Note that we only extended the evaluation of \textit{Aspire-S} up to a value length of $2^{14}$ bits, as testing larger values requires more working memory, which exceeds the RAM capacity. 
Except for \textit{Aspire-S}, the communication costs of all the schemes increase as the length of the value increases. The \textit{Aspire-P} shows communication costs similar to those of PathORAM and \textit{Aspire-PO}, at about 220 MB at $|v|=2^{16}$.

Also, the computation cost of all the schemes increases as $|v|$ increases. The \textit{Aspire-S} and \textit{Aspire-P} remain efficient even with a large value length. When $|v|=2^{13}$, the amortized total time of \textit{Aspire-S} and \textit{Aspire-P} is about 2050 ms and 222 ms, which decreases in parallel to approximately 153 ms and 61 ms.

\subsubsection{Update performance} We evaluate the update performance of \textit{Aspire-S}, \textit{Aspire-P}, and \textit{Aspire-PO} by varying the database size. 
Whenever the client performs an update, \textit{Aspire} temporarily stores it locally and subsequently issues $\mathsf{Update}$ operations after a predefined time interval or after a specified number of $\mathsf{UpdateLocal}$ to synchronize the local changes with the server. 
In contrast, PathORAM and OMAP24 do not employ update mechanisms; instead, they operate over data structures of fixed sizes.
The $\mathsf{UpdateLocal}$ of \textit{Aspire} only needs to store the updates locally; thus, it takes negligible computational time (the storage costs are presented in Section \ref{sec:local_storage}). 
We present the cost of $\mathsf{Update}$ in Figure \ref{fig:update}. 
For \textit{Aspire-P}, the client is required to download one bucket and upload a new bucket, which results in relatively low communication overhead (approximately 29.7 MB for $n = 2^{23}$). In contrast, \textit{Aspire-S} further necessitates the upload of additional evaluation keys, thereby increasing the total communication by approximately 10 MB compared to \textit{Aspire-P}.
Similar to the scenario described in $\mathsf{Setup}$, \textit{Aspire-PO} incurs substantially higher communication overhead, amounting to approximately 296 MB.
For the computation costs, \textit{Aspire-S} needs only a small amount of time to rebuild a PIR instance on a bucket (about 1.2 seconds at $n=2^{23}$).
While \textit{Aspire-P} requires re-encoding the whole database (about 216 seconds at $n=2^{23}$), since PIRANA encodes the contents from multiple buckets into a vector, changing a bucket results in re-encoding the entire vector.
Since the execution of $\mathsf{Update}$ is independent of $\mathsf{UpdateLocal}$, the computation and communication costs of $\mathsf{Update}$ can be amortized over the updates of multiple clients.

\subsection{Summary and Discussion} 

\textit{Aspire-PO}, \textit{Aspire-P}, and \textit{Aspire-S} each exhibit distinct performance characteristics and associated trade-offs. 
\textit{Aspire-PO} requires increased client-side local storage (over 100 MB when $n=2^{23}$), as well as higher communication and computational overhead during the setup phase; however, it achieves lower query latency in most of the evaluated scenarios.
In contrast, \textit{Aspire-P} and \textit{Aspire-S} incur reduced client storage requirements (about 8.4 MB) at the cost of higher computational overhead during query processing.
Overall, \textit{Aspire-P} attains a lower end-to-end execution time than \textit{Aspire-S}, particularly when the number of buckets is large. 
With parallelization enabled, Aspire attains an amortized query latency of 1.38 ms (7.53 s for a batch of 5,461 queries). 
Furthermore, the computational burden is predominantly placed on the server side, while the client incurs only a minor computational load.
Regarding update operations, \textit{Aspire-S} and \textit{Aspire-PO} support lightweight updates, whereas \textit{Aspire-P} exhibits higher update costs due to the underlying PIRANA mechanism, which encodes content from multiple buckets into a single vector.

{
The performance of \textit{Aspire} reflects a fundamental trade-off between efficiency and privacy. Whereas a plaintext key-value store can directly retrieve indexed records, \textit{Aspire} additionally performs cryptographic query processing and either a PIR scan or oblivious path accesses. This additional computation and communication enable \textit{Aspire} to hide the queried keys and plaintext-dependent access patterns, thereby achieving restricted zero-leakage security.

Compared with prior work, PathORAM is competitive and achieves lower computation time and, in some settings, lower query latency, whereas OMAP24 incurs substantially higher communication overhead and total latency because each lookup requires many interaction rounds. 
The $O(\log n)$ bandwidth complexity of PathORAM is attractive for individual accesses and makes it competitive in several scenarios. 
For a large batch of independent lookups, however, this cost is incurred for each lookup and therefore accumulates with the batch size. 
By amortizing the work across a batch, \textit{Aspire-P} maintains nearly constant communication as the batch size increases and, with parallelization, outperforms the ORAM-based alternatives for sufficiently large batches. Meanwhile, \textit{Aspire-PO} retains the computational advantages of ORAM at the cost of increased client storage and communication. \textit{Aspire} flexibly supports both PIR and ORAM instantiations, allowing users to choose the best fit for their network and computation constraints.
}

Consequently, \textit{Aspire-PO} is preferred when the client can provide large local storage and a high-bandwidth network. 
When client storage is constrained, \textit{Aspire-P} is recommended for workloads involving large batch queries over many keywords with infrequent updates, whereas \textit{Aspire-S} is more suitable for scenarios requiring frequent updates.

\section{Conclusion}
In this paper, we propose restricted zero-leakage security for encrypted databases, which allows no semantic information of the plaintext to be revealed.
The restricted zero-leakage security defends against LAA and can be practically achieved.
The PIR techniques are adopted to construct an encrypted key-value database with restricted zero-leakage security.
By transferring keyword search to index access with MPH and utilizing PBC for batch queries, the proposed \textit{Aspire} supports batch keyword search and dynamic updates. 
Experiments demonstrate that \textit{Aspire} achieves an amortized query latency of less than 1.38 ms (7.53 s for a batch of 5461 queries) on a database with $2^{22}$ key-value pairs.

\section*{Ethical Considerations}

This work does not involve human subjects, user data, or real-world systems; all experiments use randomly generated synthetic data for efficiency tests. The paper proposes only a defensive mechanism and does not introduce new attacks or analyze real-world vulnerabilities. Therefore, the potential risks are minimal, and the work contributes positively to improving system security and privacy.

\bibliographystyle{IEEEtran}
\bibliography{sample-base}
%



\appendices

\section{Notations and Definitions}
\label{app:notations}
\subsection{Notations}
We use $\{v_1,...,v_n\}$ for a list that contains $n$ elements. In particular, the $[n]$ denotes $\{1,2,...,n\}$. The $|K|$ represents the number of elements in a set $K$. Frequently used notations are summarized in Table \ref{tab:summary_of_notations}.

\begin{table}[t]
    \centering
    \begin{threeparttable}
    \caption{Summary of notations.}
    \label{tab:summary_of_notations}
    \begin{tabular}{c@{\hspace{3pt}}c}
    \hline
        Notation & Description \\
        \hline\hline
        $n$ & The (initial) number of pairs in the database.\\\hline
        $b$ & The number of buckets.\\\hline
        $B_{size}$ & The number of pairs in every encrypted bucket.\\\hline
        $Count$ & The current number of pairs in the database.\\\hline
        $Rcount$ & The number of times $\mathsf{Update}$ is executed.\\\hline
        $S_1,S_2$ & The local stash of the client.\\\hline
        $\widetilde{KV}_i$ & The $i$-th encrypted bucket.\\
        \hline\hline
        AP & \makecell{Access pattern, the contents accessed during search\\ or
update operations in the storage system.}\\\hline
        SP & \makecell{Search pattern, whether two queries are \\for the
same keyword.}\\\hline
        QC & \makecell{Query correlation, the probability of the keyword\\ chosen for the next query after the current query.}\\\hline
        VP & \makecell{Volume pattern, the number of files\\ in the response.}\\\hline
        FVP & \makecell{File volume pattern, the size of each file \\in the response.} \\\hline
        RSP &  \makecell{Response size pattern, the total size of files\\ in the response.}\\\hline
        UL & \makecell{Update leakages.}\\
        \hline

    \end{tabular}
    
    \end{threeparttable}
\end{table}

\subsection{Private Information Retrieval}
\label{app:PIR_preliminary}
PIR enables a client to retrieve elements from a \textit{public} database owned by a server without leaking the accessed indexes to the server. 
We consider PIR schemes with a setup phase in which the server and client cooperate to pre-process hints that can be used in the retrieval. 
In some PIR schemes, there is no setup phase, or the client does not participate in the setup phase; then the client takes $\bot$ as a hint.
A PIR scheme contains the following algorithms or protocols, which all take the computational security parameter $\lambda$ as an implicit input:
\begin{itemize}
    \item $(\mathsf{hint}_C,\mathsf{hint}_S) \gets \mathsf{PIR}.\mathsf{Setup}(DB)$, which takes the database $DB=\{v_1,...,v_n\}\in \mathcal{V}^n$ as input and outputs the $\mathsf{hint}_C$ for the client and the $\mathsf{hint}_S$ for the server.
    \item $(\mathsf{st}_C,\mathsf{qu})\gets \mathsf{PIR}.\mathsf{Query}(\mathsf{hint}_C,index)$, which takes an $index\in [n]$ and the $\mathsf{hint}_C$ from the client and outputs the state of the client $\mathsf{st}_C$ and the query $\mathsf{qu}$.
    \item $\mathsf{ans}\gets \mathsf{PIR}.\mathsf{Answer}(\mathsf{hint}_S,\mathsf{qu})$, which takes a query $\mathsf{qu}$ from the client and the $\mathsf{hint}_S$ of the server as input and outputs a response $\mathsf{ans}$ for the query.
    \item $v\gets \mathsf{PIR}.\mathsf{Recover}(\mathsf{st}_C,\mathsf{hint}_C,\mathsf{ans})$, which takes the response $\mathsf{ans}$, the client's state $\mathsf{st}_C$, and the hint from the client $\mathsf{hint}_C$ as input and outputs a value $v\in\mathcal{V}$.
\end{itemize}

\noindent\textbf{Correctness}. A PIR scheme $\mathsf{PIR}$ is \textit{correct} if, for any $DB=$$\{v_1,...,$ $v_n$ $\}$ and any $index\in[n]$, 
\begin{equation}
    \mathsf{PIR}.\mathsf{Recover}(\mathsf{st}_C,\mathsf{hint}_C,\mathsf{PIR}.\mathsf{Answer}(\mathsf{hint}_S,\mathsf{qu}))=v_{index},
\end{equation}
where $(\mathsf{hint}_C,\mathsf{hint}_S) $ $\gets \mathsf{PIR}.$ $\mathsf{Setup}(DB)$ and $(\mathsf{st}_C,\mathsf{qu})$ $\gets\mathsf{PIR}.$ $\mathsf{Query}$ $(\mathsf{hint}_C,index)$. 

\noindent\textbf{Security}. A PIR scheme $\mathsf{PIR}$ is $(T,\varepsilon)$-secure if, for any adversary $\mathcal{A}$ with a running time of $T$, and for all $i,j\in [n]$, it holds
\begin{equation}
	\begin{aligned}
		|\textbf{Pr}(\mathcal{A}(\mathsf{qu})=1|(\mathsf{st}_C,\mathsf{qu}\gets\mathsf{PIR}.\mathsf{Query}(\mathsf{hint}_C,i)))-
         \\
		\textbf{Pr}(\mathcal{A}(\mathsf{qu})=1|(\mathsf{st}_C,\mathsf{qu}\gets\mathsf{PIR}.\mathsf{Query}(\mathsf{hint}_C,j)))|\leq \varepsilon
	\end{aligned}
\end{equation}

Sometimes, the client wants to retrieve the value related to a given keyword without leaking the keyword, which is referred to as keyword PIR. A keyword PIR scheme $\mathsf{KPIR}$ contains the same algorithms as the PIR schemes, except that keyword PIR operates on a key-value database and queries for keywords instead of indexes. 
The security and correctness of $\mathsf{KPIR}$ are similar to $\mathsf{PIR}$. If the queried keyword is not in the database, the output of $\mathsf{Recover}$ is $\bot$ with a probability of $1-\mathsf{negl}(\lambda)$.

\subsection{Batch Code and Batch PIR}

A $(n,m,s,b)$-batch code $\mathsf{BC}$ takes $n$ elements and produces $b$ buckets with a total of $m$ codewords. 
The batch code ensures that any $s$ elements from the $n$ elements could be recovered from the $b$ buckets by retrieving at most one codeword from each bucket. Also, the batch code ensures that $m<s\cdot n$.
Batch PIR uses batch code for efficient retrieval of a batch of elements. 
For any $s$ elements to be retrieved, the client can retrieve one element from each of the $b$ buckets. The server only needs to go through the $m$ codewords once, which is smaller than calling PIR over $n$ elements $s$ times. 
In \cite{DBLP:conf/sp/AngelCLS18}, Angel et al. proposed the probabilistic batch code (PBC) based on 3-way cuckoo hashing, where the server uses 3 hash functions to map the elements in the database to $b$ buckets (each element is copied 3 times, $m=3n$), and the client uses cuckoo hashing with the same 3 hash functions to map the wanted elements to each bucket. 
The PBC contains the following functions:
\begin{itemize}
    \item $\{B_1,...,B_b\}\gets\mathsf{PBC}.\mathsf{Encode}(DB)$, which takes a database $DB=\{pl_1,...,pl_n\}$ and outputs $b$ buckets $\{B_1,...,B_b\}$, where $B_i$ is a vector of codewords in the $i$-th bucket.
    \item $\{\sigma,\bot\}\gets \mathsf{PBC}.\mathsf{GenSchedule}(I)$, which takes a set of $s$ elements $I=\{i^*_1,...,i_s^*\}$, where $i^*_j\in DB$, as input and outputs a schedule $\sigma$ that maps each element in $I$ to a bucket, with each bucket containing only one element mapped to it. It outputs $\bot$ if it cannot produce a schedule.
\end{itemize}
\noindent\textbf{Correctness}. A probabilistic batch code scheme $\mathsf{PBC}$ is said to be correct if, for any database $DB=\{pl_1,\dots,pl_n\}$, any element set $I=\{i_1^*,\dots,i_s^*\}$, and any schedule $\sigma$ outputted by the algorithm $\mathsf{PBC.GenSchedule}$ on input $I$, it holds that for every $i^* \in I$, the index $i^*$ is contained in the corresponding block, i.e., $i^* \in B_{\sigma[i^*]}$.

By using 3 hash functions and setting $b=1.5s$, the PBC fails to recover the set of $s$ elements with a probability of $p=2^{-40}$ while providing $m=3n$~\cite{DBLP:conf/sp/AngelCLS18}. 

\subsection{Minimal Perfect Hash}
\label{app:MPH_preliminary}
Minimal Perfect Hash (MPH) functions are static data structures that map a set $K\subset \mathcal{K}$ of keywords to the first $|K|$ natural numbers without collisions. MPH functions are allowed to return any value if the queried keyword is not in the original set $K$. 
MPH constructions achieve $O(n)$ bits of space, regardless of the length of the keywords. 
Recent works~\cite{DBLP:conf/alenex/EspositoGV20} show that the MPH could be stored with less than 2 bits per keyword and achieve expected linear construction time and constant lookup time. The minimal perfect hashing contains the following functions:
\begin{itemize}
    \item $mph\gets \mathsf{MPH}.\mathsf{Setup}(K)$, which takes a set $K\subset \mathcal{K}$ of $n$ keywords and outputs a minimal perfect hash function $mph$.
    \item $i\gets \mathsf{MPH}.\mathsf{Evaluate}(mph,k)$ takes a $mph$ and a keyword $k\in \mathcal{K}$ as input and outputs an integer $i$ in $[n]$ if $k$ is in $K$, where $n=|K|$.
\end{itemize}

\noindent\textbf{Correctness}. A minimal perfect hash scheme $\mathsf{MPH}$ is correct if, for any set $K\subset \mathcal{K}$, for any $k_1,k_2\in K$ and $k_1\neq k_2$, $i_1=\mathsf{MPH}.\mathsf{Evaluate}(mph,k_1)$ and $i_2=\mathsf{MPH}.\mathsf{Evaluate}(mph,k_2)$, the following holds:1) $i_1,i_2\in [n]$. 2) $i_1\neq i_2$.

\section{Analysis of Expected Bucket Overflow}
\label{app:overflow}
{
The number of elements in each bucket, denoted by $X$, follows a binomial distribution
$X \sim \operatorname{Binomial}(3n, 1/b)$. When the expected bucket load $\mu = 3n/b$ is sufficiently large and the number of buckets $b$ is not small, as in our case, $X$ can be approximated by a normal distribution. Its mean and variance are $\mu = 3n/b$ and $\sigma^2 = (3n/b)(1 - 1/b)$. For large $b$ in our case, we may further approximate it as $\sigma^2 \approx 3n/b$.

Let the capacity threshold of each bucket be
$B_{size} = {3n}/{b} + \sqrt{{6n\ln b}/{b}}$.
Under the normal approximation, the expected overflow of a single bucket can be expressed using the standard normal loss function. Therefore, the expected total overflow over all $b$ buckets is approximately
\[
E_{\mathrm{overflow}}
\approx
\sqrt{3nb}
\left[
\phi\!\left(\sqrt{2\ln b}\right)
-
\sqrt{2\ln b}
\left(1-\Phi\!\left(\sqrt{2\ln b}\right)\right)
\right],
\]
where $\phi(\cdot)$ and $\Phi(\cdot)$ denote the density function and cumulative distribution function of the standard normal distribution, respectively.
}
\section{Restricted Zero-Leakage with Leakage Functions}
 
\label{app:security_proof}
This appendix instantiates the restricted leakage function of Definition \ref{def:rzl} for encrypted key-value databases and proves that Aspire satisfies this leakage profile.

For the encrypted key-value database in Section \ref{sec:system_model}, we define
\[
\mathcal L_{\mathsf{Setup}}(DB)=L'(|DB|),
\qquad
\mathcal L_{\mathsf{Query}}(K)=\bot,
\]
\[
\mathcal L_{\mathsf{UpdateLocal}}(op,k,v)=\bot,
\qquad
\mathcal L_{\mathsf{Update}}()=L''(|DB|).
\]
Additionally, if $\mathcal L_{\mathsf{Update}}()=L''(|DB|)$ occurs only once every $\alpha$ updates, otherwise, $\mathcal L_{\mathsf{Update}}()=L''(\bot)$, we say that the encrypted key-value database additionally achieves $\alpha$-periodic size exposure.

\begin{lemma}
    The above leakage function is non-semantic and therefore a valid restricted leakage function for Definition \ref{def:rzl}.
\end{lemma}

\begin{proof}
{
Consider any two executions having the same public parameters, operation types, timing, and database-size evolution. Their setup leakages are identical because $L_{\mathsf{Setup}}$ depends only on the initial database size. Their $\mathsf{Query}$ and $\mathsf{UpdateLocal}$ leakages are both $\bot$, independent of the queried keywords and updated records. Finally, their Update leakages are identical because $L_{\mathsf{Update}}$ depends only on the current database size. Therefore, the complete leakage transcripts of the two executions are identical and contain no information about the plaintext keywords or values beyond the permitted metadata.
}
\end{proof}

\begin{theorem}
    Assuming the underlying $\mathsf{PIR}$ provides security {and the underlying encryption scheme is IND-CPA secure}, \textit{Aspire} offers restricted zero-leakage security and achieves $\alpha$-periodic size exposure with $\alpha=b$.
    \label{the:security_of_aspire}
\end{theorem}

\begin{proof}
For the security of our protocol, we construct a PPT simulator $\mathcal{S}$.

{Let $\mathbf{T}=(t_0,t_1,\ldots,t_q)$ denote the complete server-side transcript of the real execution, where $t_0$ is the setup transcript and $t_i$ is the transcript of the $i$-th operation. Let $\mathbf{T}'=(t'_0,t'_1,\ldots,t'_q)$ denote the corresponding simulated transcript. We construct a PPT simulator $\mathcal{S}$ that uses only the permitted leakage.}

{
\noindent \textbf{Setup.}
From the setup leakage, $\mathcal{S}$ obtains the padded bucket size determined by $|DB|$. It generates $\{\widetilde{KV}'_0,\ldots,\widetilde{KV}'_{b-1}\}$ by encrypting equal-length all-zero records under $sk'$ and pads each bucket to the same size as in the real setup. It then executes
lines \ref{alg1:sends_buckets}-\ref{alg1:server2} in Algorithm \ref{alg1} on these simulated buckets.
By IND-CPA security, the encryptions of the real records are computationally indistinguishable from the encryptions of equal-length all-zero records, i.e., $t'_0\approx_c t_0$.}

\noindent \textbf{Query.}
For any $t'_i$ where $o'_i$ is a $\mathsf{Query}$, the simulator generates PIR queries using line \ref{alg2:pir1} of Algorithm \ref{alg2} with uniformly random positions for each bucket and then runs lines \ref{alg2:sends_query}-\ref{alg2:sends_answer} of Algorithm \ref{alg2} to generate $t'_i$.
By the security of $\mathsf{PIR}$, queries for random positions are computationally indistinguishable from queries for real positions, and thus $t'_i \approx_c t_i$.

{
\noindent\textbf{UpdateLocal.}
For $op_i=\mathsf{UpdateLocal}$, no message is sent to the server.
Thus, $\mathcal{S}$ outputs an empty transcript, which is identical to
the real server-side transcript.
}

\noindent \textbf{Update.}
For any $t'_i$ where $o'_i$ is an $\mathsf{Update}$, the simulator maintains a counter $Rcount$ and a variable $B_{size}$.
If $Rcount\% b = 0$, the simulator sets $B_{size}$ according to $|DB|$ as line \ref{alg4:set_Bsize} of Algorithm \ref{alg4_v2}.
The simulator then generates a fresh bucket $\widetilde{KV}'$ of random ciphertexts with size $B_{size}$ and runs the line \ref{alg4:reset_PIR} of Algorithm \ref{alg4_v2}.
Since all uploaded buckets in $\mathsf{Update}$ are freshly re-encrypted, {by IND-CPA security,} $t'_i$ is computationally indistinguishable from the real transcript $t_i$.
As the simulator only needs to reset $B_{size}$ when $Rcount\%b=0$, \textit{Aspire} achieves $\alpha$-periodic size exposure with $\alpha=b$.

{
For every operation, the above indistinguishability holds conditioned on any preceding transcript and consistent protocol state. Since the execution contains only polynomially many operations, the total distinguishing advantage is bounded by the sum of polynomially many negligible advantages and is therefore negligible. Consequently, $\mathbf{T}'\approx_c\mathbf{T}$, even when each operation is chosen adaptively based on the preceding transcripts.}
Therefore, the simulated transcripts generated by $\mathcal{S}$ are computationally indistinguishable from the real execution transcripts, and \textit{Aspire} achieves adaptive restricted zero-leakage security.


\end{proof}

{
\section{Computation and Communication Breakdown of Aspire}
\label{app:breakdown}
To understand the sources of \textit{Aspire}'s overhead, we break down the total runtime of \textit{Aspire-P} into two components, the PBC plus MPH processing and the PIRANA execution within Aspire-P. As shown in Figure~\ref{fig:breakdown}, PBC and MPH processing take less than $1$ ms for all database sizes, whereas PIRANA execution, including query generation, takes $3.674$--$85.109$ s and dominates the overall runtime.

We further compare \textit{Aspire-P} with standalone PIRANA to quantify the overhead introduced beyond a basic PIR query. 
We note that standalone PIRANA provides private indexed retrieval but does not support \textit{Aspire}’s dynamic key-value functionality or its restricted zero-leakage security. Therefore, this comparison is intended only to quantify the additional computation and communication introduced by \textit{Aspire-P}.
We use the same parameters for PIRANA as for \textit{Aspire-P}.
Standalone PIRANA takes $3.472$--$80.067$ s, and \textit{Aspire-P} increases the runtime by only $4.84\%$--$7.74\%$. \textit{Aspire-P} also incurs only $0.79$--$1.00$ MB of additional communication, corresponding to an increase of $5.17\%$--$9.47\%$. Thus, the PBC and MPH operations introduce little computational cost, and incorporating PIRANA into \textit{Aspire-P} causes only limited overhead compared with standalone PIR queries.

\begin{figure}[tbp]
    \centering
    \subfloat[Computation Cost]
	{
 \label{fig:breakdown_a}
		\begin{minipage}[b]{0.23\textwidth}
        \centering
        \includegraphics[width=\textwidth]{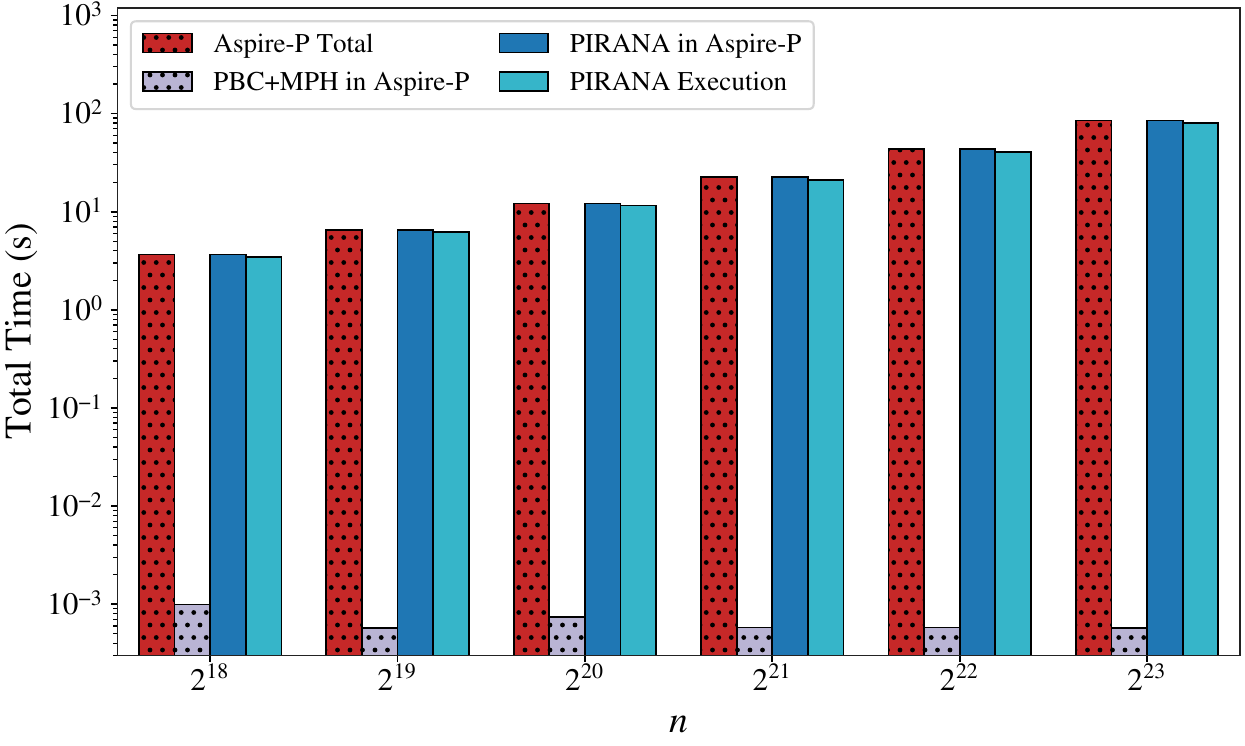}
    \end{minipage}
	}
    \subfloat[Communication Cost]
	{
 \label{fig:breakdown_b}
		\begin{minipage}[b]{0.23\textwidth}
        \centering
        \includegraphics[width=\textwidth]{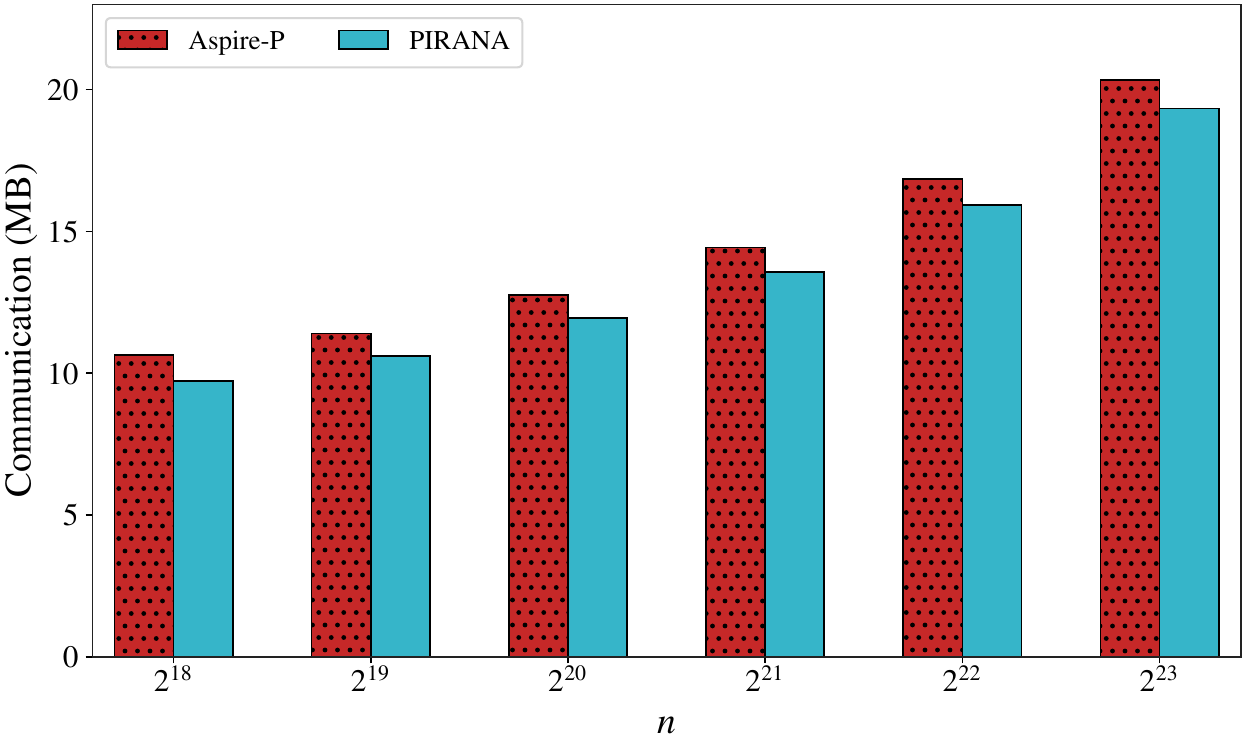}
    \end{minipage}
	}
    
    \caption{Runtime breakdown and communication comparison of \textit{Aspire-P}. The four bars in (a) represent the total runtime of \textit{Aspire-P}, the PBC plus MPH processing time in \textit{Aspire-P}, the PIRANA execution time in \textit{Aspire-P}, and the runtime of standalone PIRANA on a database of the same size, respectively.}
    \label{fig:breakdown}
\end{figure}

}

\section{Instantiations of Keyword PIR}
\subsection{Existing Keyword PIR}
\label{app:review_of_KPIR}
Keyword PIR was first proposed by Chor et al.~\cite{DBLP:journals/iacr/ChorGN98}. A naive solution is to download the entire keyword set, find the index of the querying keyword, and invoke a standard PIR. 
However, the keyword set may be prohibitively large to download. 
Chor et al.~\cite{DBLP:journals/iacr/ChorGN98} propose a solution that does not require downloading the keyword set; instead, it requires $O(\log n)$ rounds of standard PIR.

Recent works are based on hashing schemes to solve the above problems. Some works~\cite{DBLP:conf/uss/AliLP0SSY21,DBLP:conf/uss/0001LPZ000D25} utilize a public hashing function to distribute the entries of the database to a bucket. The client then retrieves the entire bucket corresponding to the querying keyword using a standard PIR. Ali et al.~\cite{DBLP:conf/uss/AliLP0SSY21} also use cuckoo hashing to ensure that each bucket contains only one entry. However, this solution requires multiple invocations of standard PIR and an increased number of buckets.

There are also works relying on homomorphically evaluating the equality operator~\cite{DBLP:conf/uss/MahdaviK22,DBLP:journals/pvldb/AhmadAAG22}. 
Mahdavi et al.~\cite{DBLP:conf/uss/MahdaviK22} use the constant-weight code to encode all the keywords into a representation with the same number of 1's, which reduces the multiplicative depth of equality circuits. 
Ahmad et al.~\cite{DBLP:journals/pvldb/AhmadAAG22} use Fermat's little theorem to determine whether two numbers are equal, where computation is performed in an integer space rather than requiring the evaluation of a bitwise boolean circuit. 
Nevertheless, the complicated homomorphic operations in the equality operator incur significantly higher computation and communication costs than existing standard PIR~\cite{DBLP:conf/uss/0001LPZ000D25}.

Other works encode the database as a linear combination, which can then be queried using standard PIR. 
For example, Patel et al.~\cite{DBLP:conf/uss/PatelSY23} proposed SparsePIR, where the client can follow the same encoding to retrieve the sum of entries after encoding and reconstruct the value corresponding to the keyword. SparsePIR is also compatible with important optimizations, such as recursion and batching. Similarly, ChalametPIR~\cite{DBLP:conf/ccs/CeliD24} uses binary fuse filters (BFF) to encode the database and an inner product between the BFF output and the query vector for decoding. 
Hao et al.~\cite{DBLP:conf/uss/0001LPZ000D25} utilize the sparse property of the BFF, and the client only needs to access a small constant number $\alpha$ ($\alpha=3$ in their work) of entries. 

Although instantiating the keyword PIR by directly sending the keyword-to-index mapping is communication expensive, there are several ways to decrease the size of the mapping.
In ~\cite{DBLP:journals/iacr/ChorGN98}, Chor et al. use perfect hashing to construct a keyword PIR based scheme. However, the proposed scheme requires two rounds of PIR and enlarges the database by three times.
Hao et al.~\cite{DBLP:conf/uss/0001LPZ000D25} also propose a keyword PIR that utilizes the approximate key-to-index mapping. 
In their experiments, the mapping costs 324 KB of storage for $2^{20}$ entries with $\varepsilon=4$. 
To retrieve the corresponding value of a keyword, Hao et al. propose a Row-KOPIR scheme to retrieve all the entries in a row of the database.

\subsection{New Keyword PIR based on MPH}
\label{app:keyword_PIR}
\input{Alg_KeywordPIR}
Keyword PIR can be instantiated with standard PIR by sending the keyword-to-index mapping to the client. We find that the mapping can be replaced by a minimal perfect hashing function, which incurs only minimal communication and storage costs for the client.

Our keyword PIR scheme is presented in Algorithm \ref{algkeywordPIR}. In the setup, the server extracts all the keywords from the database and calls the $\mathsf{MPH}.\mathsf{Setup}(DB.keyword)$ to generate a minimal perfect hash $mph$, which maps all the keywords to $[n]$ without collisions. 
The server then re-arranges the database as $DB'$, where $DB'[\mathsf{MPH}$ $.\mathsf{Evaluate}$ $(mph,k_i)]=k_i||v_i$ for all $i\in [n]$. The server then calls the standard PIR scheme $\mathsf{PIR}$ to set up on $DB'$. The $mph$ is then sent to the client as the client's hint.

When the client wants to query a keyword $k$, it uses the $\mathsf{MPH}$ $.\mathsf{Evaluate}$ $(mph,k)$ to get the index $ind$ of the queried keyword in $DB'$ and calls the standard PIR for the query. Also, the $\mathsf{Answer}$ of our keyword PIR protocol is the same as the standard PIR. 
To recover the queried content, the client calls $\mathsf{PIR}.$ $\mathsf{Recover}$ to get $k'||v$. If the $k'=k$, then the client gets the queried content $v$. Otherwise, the client has queried a keyword that is not in the database, and the client receives $\bot$.

Most of the steps in our keyword PIR are exactly the same as those in the standard PIR, except for the steps related to minimal perfect hashing. This means that the optimizations in the standard PIR could also be used in our keyword PIR.
Additionally, the minimal perfect hash function could be efficient and storage-friendly~\cite{DBLP:conf/alenex/EspositoGV20}.
The $mph$ could be constructed with less than 2 bits per keyword (approximately 0.22 MB for $2^{20}$ arbitrarily long keywords in our experiments) and achieves expected linear construction time and constant lookup time. Thus, our keyword PIR has computational efficiency comparable to that of standard PIR.

\begin{figure}[tp]
    \centering
    \includegraphics[width=0.8\linewidth]{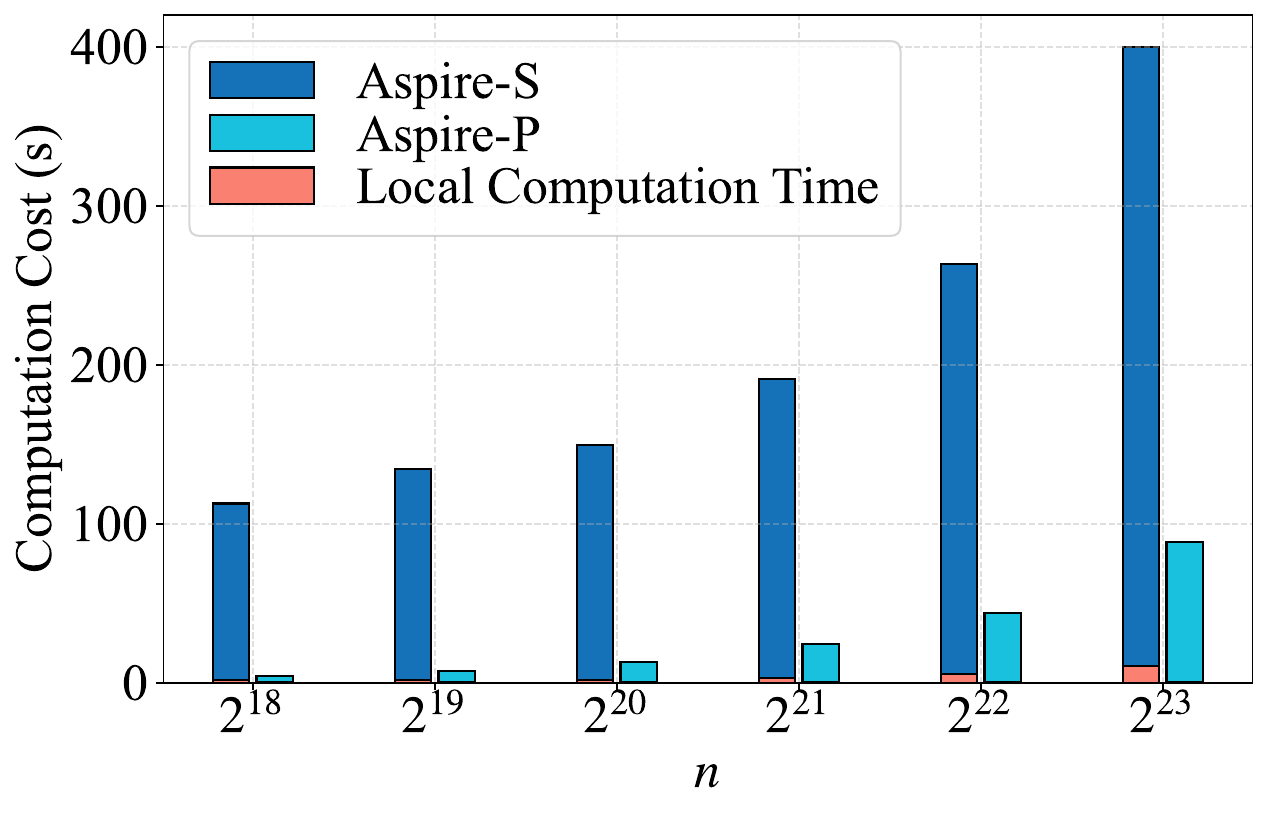}
    \caption{The local computation cost of \textit{Aspire-S} and \textit{Aspire-P} in search query.}
    \label{fig:local_computation_cost}
\end{figure}
\section{Local Query Costs}
\label{app:local_computation_cost}
This section demonstrates that the local computation overhead of \textit{Aspire} when employing PIR (i.e., \textit{Aspire-S} and \textit{Aspire-P}) constitutes only a small fraction of the overall computational cost.
We vary the number of key-value pairs from $2^{18}$ to $2^{23}$ (using 512 buckets to support batch queries of 341 keywords), and the corresponding results are presented in Figure \ref{fig:local_computation_cost}.
For \textit{Aspire-S}, the local computation time is approximately 1.5 s for $n = 2^{18}$ and 10.3 s for $n = 2^{23}$, which is less than 2.6\% of the total computation time in both cases.
For \textit{Aspire-P}, the local computation time is approximately 0.10 s for $n = 2^{18}$ and 0.25 s for $n = 2^{23}$, again accounting for less than 2.6\% of the total computation time.
These results indicate that the client in both \textit{Aspire-S} and \textit{Aspire-P} incurs only a modest computational burden when issuing search queries.

\section{Artifact Appendix}
\label{app:artifact}

\subsection{Description \& Requirements}

\subsubsection{How to access}
We have open-sourced the implementation of \textit{Aspire} at
\url{https://doi.org/10.5281/zenodo.22813536}.
The artifact contains the implementations of the three
instantiations evaluated in the paper:
PIRANA-based \textit{Aspire} (\textit{Aspire-P}),
Spiral-based \textit{Aspire} (\textit{Aspire-S}), and
PathORAM-based \textit{Aspire} (\textit{Aspire-PO}).

\subsubsection{Hardware dependencies}
Our experiments were conducted on a machine equipped with an
AMD EPYC 9354 32-core processor and 1\,TB of main memory.
The large amount of memory is particularly important when evaluating \textit{Aspire-S} as the Spiral-based configuration used by our benchmark requires approximately 1\,TB of RAM for the largest evaluated parameter sets.
The PIRANA and PathORAM-based implementations can be evaluated with substantially less memory.

\subsubsection{Software dependencies}
Our implementation is written in C++ and was evaluated on
Ubuntu~22.04.5.

The main software dependencies are:

\begin{itemize}
    \item Microsoft SEAL, version $\geq 4.0.0$;
    \item Microsoft Kuku, version 2.1.0;
    \item OpenSSL, version $\geq 3.0.2$;
    \item Crypto++, version 8.7.0;
    \item Intel HEXL, version 1.2.6; and
    \item Sux, version 1.0.3.
\end{itemize}

Detailed installation instructions are provided in the README file.
Git LFS, CMake, a C++ compiler, GNU Make, and Python~3 are also required to obtain, build, and execute the artifact.



\subsection{Artifact Installation \& Configuration}
The artifact uses Python scripts to compile the implementation and automate the experiments.

\subsubsection{Building the artifact}
The artifact provides a Python script for compiling the benchmark.

\begin{verbatim}
python3 compile.py
\end{verbatim}

After compilation, the main benchmark executable has the following interface:

\begin{verbatim}
./build/main <bucket_num> <query_num> \
  <log2(N)> <key_bits> <val_bits> \
  <dim0> <dim1> <threads> <pir_type>
\end{verbatim}

The parameters are defined as follows:

\begin{itemize}
    \item \texttt{bucket\_num}: the number of PBC buckets $b$.

    \item \texttt{query\_num}: the number of keywords queried in one batch. 

    \item \texttt{log2(N)}: the base-2 logarithm of the number of key--value pairs $n$ in the database. For example, \texttt{log2(N)=18} corresponds to $n=2^{18}$.

    \item \texttt{key\_bits}: the keyword length $|k|$ in bits.
    The default experiments use 32-bit keywords.

    \item \texttt{val\_bits}: the value length $|v|$ in bits.
    For example, \texttt{val\_bits=2048} corresponds to
    a 256-byte value.

    \item \texttt{dim0} and \texttt{dim1}: internal parameters of the
    Spiral PIR implementation. They are only relevant to
    \texttt{SpiralPIR} and should not be modified arbitrarily.

    \item \texttt{threads}: the number of threads used for parallel query processing. 

    \item \texttt{pir\_type}: the underlying oblivious-access
    mechanism. \texttt{PIRANA}, \texttt{SpiralPIR}, and
    \texttt{ORAM} correspond to \textit{Aspire-P},
    \textit{Aspire-S}, and \textit{Aspire-PO}, respectively.
\end{itemize}

\subsubsection{Quick functionality test}
A small-scale experiment can be executed as follows:

\begin{verbatim}
./build/main 512 341 18 32 2048 \ 
6 2 32 ORAM
./build/main 512 341 18 32 2048 \ 
6 2 32 PIRANA
./build/main 512 341 18 32 2048 \ 
6 2 32 SpiralPIR
\end{verbatim}


For comprehensive evaluation, the artifact provides the
\texttt{run.py} benchmark driver.

\begin{verbatim}
python3 run.py --type <TYPE> \
  --threads <THREADS> --testcase <ID>
\end{verbatim}

where \texttt{TYPE} is one of
\texttt{PIRANA}, \texttt{SpiralPIR}, and \texttt{ORAM}. Note that \texttt{SpiralPIR} requires substantial memory and may consume up to approximately 1\,TB of RAM for large parameter settings.
The argument \texttt{testcase} controls the parameter sweep:

\begin{itemize}
    \item \texttt{testcase=1}: vary the number of elements in
    the database (Figures 4, 5, 6, and 9);
    \item \texttt{testcase=2}: vary the value length (Figure. 8);
    \item \texttt{testcase=3}: vary the number of buckets (Figure. 7); and
    \item \texttt{testcase=0}: execute all of the above tests.
\end{itemize}

The default automated evaluation of the query uses 16 threads.  



\subsection{Major Claims}

The artifact supports the following major experimental claims:

\begin{itemize}
    \item \textbf{(C1) \textit{Aspire-P}:}
    \textit{Aspire-P} achieves efficient large-batch queries with low
    communication overhead and benefits significantly from parallel execution.

    \item \textbf{(C2) \textit{Aspire-S}:}
    \textit{Aspire-S} achieves communication-efficient queries and supports
    lightweight updates, consistent with the results reported in the paper.

    \item \textbf{(C3) \textit{Aspire-PO}:}
    \textit{Aspire-PO} achieves low query computation and latency by using PathORAM, at the cost of higher client storage and communication overhead.
\end{itemize}

\subsection{Evaluation}

\subsubsection{Experiment (E1): \textit{Aspire-P}
[5 human-minutes + 5 hours compute time]}

This experiment evaluates the PIRANA-based \textit{Aspire-P}
and validates Claim C1.

\textit{[Preparation]}
The artifact varies the database size, value length, and number of buckets.

\textit{[Execution]}
Execute:
\begin{verbatim}
python3 run.py --type PIRANA \
  --threads 1 --testcase 0
\end{verbatim}
\begin{verbatim}
python3 run.py --type PIRANA \
  --threads 16 --testcase 0
\end{verbatim}

\textit{[Results]}
The results should follow the \textit{Aspire-P} and  \textit{Aspire-P} Parallel curves in
Figures 4-9. In particular, \textit{Aspire-P} maintains low
communication overhead for large batches. Absolute runtimes may vary with hardware.

\subsubsection{Experiment (E2): \textit{Aspire-S}
[5 human-minutes + 3 hours compute time]}

This experiment evaluates the Spiral-based \textit{Aspire-S}
and validates Claim C2.

\textit{[Preparation]}
Use the same experimental settings as E1.
Note that \texttt{SpiralPIR} requires substantial memory and may consume approximately 1\,TB of RAM for large parameter settings.

\textit{[Execution]}
Execute:
\begin{verbatim}
python3 run.py --type SpiralPIR \
  --threads 1 --testcase 0
\end{verbatim}
\begin{verbatim}
python3 run.py --type SpiralPIR \
  --threads 16 --testcase 0
\end{verbatim}

\textit{[Results]}
The results should follow the \textit{Aspire-S} curves in
Figures 4-9. 

\subsubsection{Experiment (E3): \textit{Aspire-PO}
[5 human-minutes + 1 hours compute time]}

This experiment evaluates the PathORAM-based \textit{Aspire-PO}
and validates Claim C3.

\textit{[Preparation]}
Use the same database, query, and parameter settings as E1.

\textit{[Execution]}
Execute:
\begin{verbatim}
python3 run.py --type ORAM \
  --threads 1 --testcase 0
\end{verbatim}
\begin{verbatim}
python3 run.py --type ORAM \
  --threads 16 --testcase 0
\end{verbatim}

\textit{[Results]}
The results should follow the \textit{Aspire-PO} curves in
Figures 4-9.

%
%
%
%

\subsection{Acknowledgement}

We thank the artifact evaluators for their time and effort in
evaluating our implementation and reproducing the experimental
results.
\end{document}



\appendices
\section{Artifact Appendix}

The artifact appendix is meant to be a self-contained document presenting a roadmap for setting up and evaluating your artifact. It should provide the following elements:
\begin{enumerate}
\item a list of the hardware, software, and configuration \textbf{requirements} for running the artifact;
\item a clear description of how, and in what respects, the artifact \textbf{supports} the research presented in the paper;
\item a guide for how others can \textbf{execute and validate} the artifact for its functional and usability aspects;
\item the \textbf{major claims} made by your paper and a clear description of how to obtain data for each claim through your supplied artifact.
\end{enumerate}

Points 3-4 are not required if only the \textit{Available} badge is requested. Linking the paper claims to the artifact is a necessary step that allows artifact evaluators to functionally test and reproduce your results. Towards that end, explicitly list down items (e.g., results, plots, tables) from the paper and cross-reference them with the experiments to be reproduced.
 
Unless otherwise stated, filling every section below with the requested contents is mandatory for participating in the artifact evaluation process. Section titles cannot be changed.

\subsection{Description \& Requirements}

This section should list all the information necessary to recreate the experimental setup you used to run your artifact.

Provide also a link to an archival repository where all the artifact's main components (software, datasets, documentation, etc.) can be accessed and, where this applies, the minimal hardware and software requirements to run your artifact.

It is also very good practice to list and describe in this section benchmarks where those are part of your artifacts or simply have been used to produce results with it.

\subsubsection{How to access}
Describe here how to access your artifact. During the artifact evaluation, in case of a private repository, you should provide instructions on how to access it. For the camera-ready version of the appendix, you must provide a DOI link to the AEC-approved artifact version that you must upload by then to permanent storage.

\subsubsection{Hardware dependencies}
Simply write ``None." where this does not apply to your artifact.

\subsubsection{Software dependencies} 
Simply write ``None." where this does not apply to your artifact.

\subsubsection{Benchmarks} 
Describe here any data (e.g., datasets, models, workloads, etc.) required by the experiments with this artifact reported in your paper. Simply write ``None." where this does not apply to your artifact.

\subsection{Artifact Installation \& Configuration}

This section should include all the high-level installation and configuration steps required to prepare the environment to be used for the evaluation of your artifact.

\subsection{Experiment Workflow}

This section should provide a high-level view of the experimental workflow and how it is implemented, invoked, and (if needed) customized. The section is optional if the experiment workflow can be easily embedded in the Evaluation section.

\textit{Detailed information for submission, reviewing, and badging process followed for the evaluation of NDSS 2024 artifacts can be found at:} \url{https://secartifacts.github.io/ndss2024/}.

\subsection{Major Claims}
Enumerate here the major claims (Cx) made in your paper. For this purpose, we ask you to use the this and the following section and cross-reference the items therein, as explained next.

Follows an example:

\begin{itemize}
    \item (C1): \textsc{System} achieves the same accuracy of the state-of-the-art systems for a task X while saving 2x storage resources. This is proven by the experiment (E1) whose results are illustrated/reported in [refer to your paper's plots, tables, sections, etc.].
    \item (C2): \textsc{System} has been used to uncover new bugs in the Y software. This is proven by the experiments (E2) and (E3) in [ibid].
\end{itemize}

\subsection{Evaluation}

This section include all the operational steps and experiments which must be performed to evaluate if your artifact is functional and validate the proposed/presented experiments. For this purpose, we ask you to use the this and the preceding section and cross-reference the items therein.

If the \textit{Reproduced} badge is requested, the section should reference your paper's key results: describe the expected numbers and, where applicable, the maximum variation expected (particularly important for performance numbers). If an experiment is a scaled-down version of the one run for the paper's evaluation, provide a clear indication of it with a justification of the scaling-down being meaningful.

We also highly encourage you to provide your estimates of human- and compute-time for each of the listed experiments.

Follows an exemplary structure for one experiment (Ey):

\subsubsection{Experiment (E1)}
[Optional Name] [30 human-minutes + 1 compute-hour]: provide a short explanation of the experiment and expected results.

\textit{[How to]} Describe thoroughly the steps to perform the experiment and to collect and organize the results as expected from your paper. We encourage you to use the following structure with three blocks for the description of an experiment (and to provide precise indications about the expected outcome for each of the steps from the blocks):

\textit{[Preparation]}
Describe the steps required to prepare and configure the environment for this experiment.

\textit{[Execution]}
Describe the steps to run this experiment.

\textit{[Results]}
Describe the steps required to collect and interpret the results for this experiment.

\subsection{Customization}
Provide here notes on how to customize your experiments, when applicable. The section is optional.

\subsection{Notes}
This section is meant to allow authors to include any further important notes that may not fall within any of the previous categories. We kindly encourage you to remove this section where this sort of content may not be strictly needed (rather than filling it with unnecessary or redundant information).


%% file: Alg_Setup.tex
\begin{algorithm}[!t]
 
    \caption{Setup.}
    \label{alg1}
    \SetKwFunction{FMain}{Setup}
    The client holds a $S_1=\{s_0,...,s_{b-1}\}$, where $s_i$ is initialized as empty, a $Count=n$, which records the number of key-value pairs in database. The server and the client both hold a $Rcount$, which records the reconstruction times and is initialized as 0\;
    \SetKwProg{Bn}{procedure}{}{end}
    \Bn{\FMain{$DB={(k_1,v_1),\dots,(k_n,v_n)},\lambda$}}{
        \textbf{Client}:\\
        Sets $K$ as $\{k_1,...,k_n\}$\label{alg1:ekeywords}\;
        $\{K_0,K_1,\dots,K_{b-1}\}\gets\mathsf{PBC.Encode}(K)$\label{alg1:PBCencode}\;
        Sets the bucket size $B_{size}$ to $3n/b+\sqrt{(6n\log b)/b}$\label{alg1:set_bsize}\;
        \For {$i\in \{0,...,b-1\}$}{
            \While{$|K_i|>B_{size}$}{\label{alg1:bsize_b}
                Removes a random keyword $k$ from $K_i$ and adds $(k,j,v,ADD)$ to $S_1[i]$, where the $j$-th hash function in PBC maps $k$ to $i$ and $v$ is the corresponding value of $k$\;
            }\label{alg1:bsize_e}
            $mph_i\gets \mathsf{MPH.Setup}(K_i)$\label{alg1:setmph}\;
            Rearranges $K_i$ according to $mph_i$\label{alg1:rearrange}\;
            \For{$j \in [|K_i|]$}{\label{alg1:enc1}
                Encrypts $K_i[j]||v||p$ as $\widetilde{KV}_{i,j}$, where $v$ is the corresponding value of keyword $K_i[j]$ and the $p$-th hash function in PBC maps $K_i[j]$ to the $i$-th bucket\;
            }\label{alg1:enc2}
            Pads $\widetilde{KV}_i$ to the size of $B_{size}$ with {encryptions of all-zero records} if $|\widetilde{KV}_i|<B_{size}$\label{alg1:pad}\;
        }

        Sends $\{\widetilde{KV}_0,...,\widetilde{KV}_{b-1}\}$ to the server\label{alg1:sends_buckets}\;
        \textbf{Server}:\\
        \For{$i\in \{0,...,b-1\}$}{\label{alg1:server1}
            $(\mathsf{hintc}_i,\mathsf{hints}_i) \gets {\mathsf{PIR.Setup}(\widetilde{KV}_i)}$\;
        }
        Sends $\{\mathsf{hintc}_0,...,\mathsf{hintc}_{b-1}\}$ to the client\label{alg1:server2}\;
    }

\end{algorithm}

%% file: Alg_Query.tex
\begin{algorithm}[!t]
    
    \caption{Query.}
    \label{alg2}
    \SetKwFunction{FMain}{Query}
    \SetKwProg{Bn}{procedure}{}{end}
    \Bn{\FMain{$Q=\{q_1,...,q_{s}\}$}}{
        \textbf{Client}:\\

        $\sigma \gets \mathsf{PBC.GenSchedule}(q_1,...,q_{s})$, where $\sigma$ maps $q_1,...,q_{s}$ to different buckets\label{alg2:pbcgen}\;
        \For{$i\in \{0,...,b-1\}$}{
            \eIf{$\sigma^{-1}[i]\neq \bot$}{\label{alg2:mph1}
            $pos_i\gets mph_i(\sigma^{-1}[i])$\;}{
            $pos_i\gets r$\;
            }\label{alg2:mph2}
            $(\mathsf{stc}_i,\mathsf{qu}_i)\gets \mathsf{PIR.Query}(\mathsf{hintc}_i,pos_i)$\label{alg2:pir1}\;
        }
        Sends $\{\mathsf{qu}_0,...,\mathsf{qu}_{b-1}\}$ to the server\label{alg2:sends_query}\;
        \textbf{Server}:\\
        \For{$i\in \{0,...,b-1\}$}{
            $\mathsf{ans}_i\gets {{\mathsf{PIR}}.\mathsf{Answer}(\mathsf{hints}_i,\mathsf{qu}_i)}$\;
        }
        Sends $\{\mathsf{ans}_0,...,\mathsf{ans}_{b-1}\}$ to the client\label{alg2:sends_answer}\;
        \textbf{Client}:\\
        Initializes $V$ as the retrieved value for $q_1,...,q_{s}$\;
        \For{$i\in \{0,...,b-1\}$}{
            $k_i||v_i||p_i\gets {{\mathsf{PIR}}.\mathsf{Recover}(\mathsf{stc}_i,\mathsf{hintc}_i,\mathsf{ans}_i)}$\;
        }\label{alg2:pir2}
        \For{$q_i\in \{q_1,...,q_s\}$}{\label{alg2:check1}
            $index_i\gets \sigma(q_i)$\;
            Sets $V[i]$ as $v_{index_i}$ if $k_{index_i}=q_i$\; 
        }
        Checks the local stash $S_1$ for updates on $Q$ and updates the $V$ correspondingly\label{alg2:check2}\;
    }

\end{algorithm}

%% file: Alg_Update_v2.tex
\begin{algorithm}[!t]
    
    \caption{UpdateLocal.}
    \label{alg3_v2}
    \SetKwFunction{FMain}{UpdateLocal}
    \SetKwProg{Bn}{procedure}{}{end}
    \Bn{\FMain{$k,v,op$}}{
        \textbf{Client}:\\
        Computes $Pos={pos_1,...,pos_3}$, where $pos_i$ is $hash_i(k)$\label{alg3:pbc}\;\Comment{$hash_i$ is the $i$-th hash function in PBC}\;

            \eIf{$op$ is ADD}{\label{alg3:up1}
                $Count\gets Count+1$\;
                Adds $(k,i,v,ADD)$ to $S_1[pos_i]$ for every $i\in \{1,2,3\}$.
            }{
                $Count\gets Count-1$\;
                Adds $(k,i,\bot,DEL)$ to $S_1[pos_i]$ for every $i\in \{1,2,3\}$.
            }   \label{alg3:up2}
        
    }
    
\end{algorithm}

%% file: Alg_Reconstruct_v2.tex
\begin{algorithm}[!t]
  
    \caption{Update.}
    \label{alg4_v2}
    \SetKwFunction{FMain}{Update}
    \SetKwProg{Bn}{procedure}{}{end}
    \Bn{\FMain{}}{
        \textbf{Client}:\\
        \If{$Rcount\%b=0$}{\label{alg4:setBsize1}
            $B_{size}=3Count/b+\sqrt{(6Count\log b)/b}$\label{alg4:set_Bsize}\;
        }\label{alg4:setBsize2}
        
        Downloads the $\widetilde{KV}_{Rcount\%b}$ from the server and decrypts it as $KV'$\label{alg4:download}\;
        \For{$(k,i,v,op)\in S_1[Rcount\%b]$}{\label{alg4:updates_1}
            \eIf{$op$ is ADD}{
                Adds $k||v||i$ to $KV'$\;
            }{
                Deletes the entry starting with $k$ from $KV'$\;
            }
        }
        $S_1[Rcount\%b]\gets \emptyset$\label{alg4:updates_2}\;

             
        \While{$|KV'|>B_{size}$}{\label{alg4:remove1}
            Removes a random element $k||v||i$ from $KV'$ and adds $(k,i,v,ADD)$ to $S_1[Rcount\%b]$\;
            }\label{alg4:remove2}
            
        Sets $K$ as the keywords in $KV'$\label{alg4:genmph1}\;
        $mph_{Rcount\%b}\gets \mathsf{MPH.Setup}(K)$\;
        Rearranges the $KV'$ according to $mph_{Rcount\%b}$\label{alg4:genmph2}\;
        Encrypts $KV'$ as $\widetilde{KV}'$ and pads $\widetilde{KV}'$ to the size of $B_{size}$ with {encryptions of all-zero records} if $|\widetilde{KV}'|<B_{size}$\; 
        Uploads $\widetilde{KV}'$ to the server\;
        
        \textbf{Server and Client}:\\
        Setup a PIR on $\widetilde{KV}'$ to replace the original PIR on bucket $Rcount\% b$\label{alg4:reset_PIR}\;

        $Rcount=Rcount+1$\label{alg4:set_Rcount}\;
    }
    
\end{algorithm}

%% file: Alg_KeywordPIR.tex
\begin{algorithm}[!t]
   
    \caption{Keyword PIR via MPH.}
    \label{algkeywordPIR}
    \SetKwFunction{FMain}{Setup}
    
    \SetKwProg{Bn}{procedure}{}{end}
    \Bn{\FMain{$DB={(k_1,v_1),\dots,(k_n,v_n)}$}}{
        $mph\gets\mathsf{MPH}.\mathsf{Setup}(DB.keyword)$\;
        Constructs $DB'$, where $DB'[\mathsf{MPH}.\mathsf{Evaluate}(mph,k_i)]=k_i||v_i$ for $i\in [n]$\;
        $\mathsf{hint_{PIR_C}},\mathsf{hint_{PIR_S}}\gets \mathsf{PIR}.\mathsf{Setup}(DB')$\;
        Return $(\mathsf{hint_C}:=(\mathsf{hint_{PIR_C}},mph),\mathsf{hint_S}:=\mathsf{hint_{PIR_S}})$\;
    }
    \SetKwFunction{FMain}{Query}
    \SetKwProg{Bn}{procedure}{}{end}
    \Bn{\FMain{$\mathsf{hint_C},k$}}{
        Parse $\mathsf{hint_C}$ as $(\mathsf{hint_{PIR_C}},mph)$\;
        $ind\gets \mathsf{MPH}.\mathsf{Evaluate}(mph,k)$\;
        $(\mathsf{st_{PIR_C}},\mathsf{qu})\gets \mathsf{PIR}.\mathsf{Query}(\mathsf{hint_{PIR_C}},ind)$\;
        Return $(\mathsf{st_C}:=(\mathsf{st_{PIR_C}},k),\mathsf{qu})$\;
    }
    \SetKwFunction{FMain}{Answer}

    \SetKwProg{Bn}{procedure}{}{end}
    \Bn{\FMain{$\mathsf{qu},\mathsf{hint_S}$}}{
        $\mathsf{ans}\gets \mathsf{PIR}.\mathsf{Answer}(\mathsf{qu},\mathsf{hint_S})$\;
        Return $\mathsf{ans}$\;
    }

    \SetKwFunction{FMain}{Recover}
    
    \SetKwProg{Bn}{procedure}{}{end}
    \Bn{\FMain{$\mathsf{st_C},\mathsf{hint_C},\mathsf{ans}$}}{
        Parse $\mathsf{st_C}$ as $(\mathsf{st_{PIR_C}},k)$ and parse $\mathsf{hint_C}$ as $(\mathsf{hint_{PIR_C}},mph)$\;
        $k'||v\gets \mathsf{PIR}.\mathsf{Recover}(\mathsf{st_{PIR_C}},\mathsf{hint_{PIR_C}},\mathsf{ans})$\;
        If $k'=k$, return $v$, else return $\bot$\;
    }
    
\end{algorithm}